\documentclass[a4paper]{amsart}
\usepackage[margin=1in]{geometry}
\usepackage{amsfonts}
\usepackage{amssymb}
\usepackage{graphicx}
\usepackage[all]{xy}
\input {diagxy} 
\usepackage{xcolor}

\def\>#1{{\bf #1}}                

\def\dim{\operatorname{dim}}              

\newtheorem{theorem}{Theorem}
\newtheorem{corollary}{Corollary}
\newtheorem{proposition}{Proposition}
\newtheorem{definition}{Definition}
\newtheorem{lemma}{Lemma}

\newtheorem{remark}{Remark}
\newtheorem{convention}{Convention}

\def\dem{\noindent\textbf{\emph{Proof. }}}\def\findem{$~\hfill\Box$}

\begin{document}
\title{Dirac Structures and Reduction of Optimal Control Problems with symmetries}
\author[A.~Ibort]{Alberto Ibort}
\address{Departamento de Matem\'aticas, Universidad Carlos III de Madrid,
Avda.\ de la Universidad 30, 28911 Legan\'es, Madrid, Spain}
\email{albertoi@math.uc3m.es}

\author[M.~Leok]{Melvin Leok}
\address{Department of Mathematics, University of California, San Diego,
9500 Gilman Drive, La Jolla, CA 92093-0112, USA}
\email{mleok@ucsd.edu}

\author[R.~Salmoni]{Rebecca Salmoni}
\address{Departamento de Matem\'aticas, Universidad Carlos III de Madrid,
Avda.\ de la Universidad 30, 28911 Legan\'es, Madrid, Spain}

\author[C.~Sard\'on]{Cristina Sard\'on}
\address{Departamento de Matem\'atica Aplicada, Universidad Polit\'ecnica de Madrid,
Avenida Juan de Herrera 6, 28040 Madrid, Spain}
\email{mariacristina.sardon@upm.es}

\date{\today}


\begin{abstract}
We study the symmetry reduction of geometric optimal control problems by means of
Dirac structures. The first-order conditions of the Pontryagin Maximum Principle are
formulated as an implicit Dirac system on the Pontryagin bundle. For a problem
invariant under a Lie group we reduce both the Lagrange--Pontryagin variational
principle and the underlying Dirac structure, the reduced canonical Dirac structure
being derived rather than postulated, and we prove that the two procedures yield the
same reduced equations. The reduction picture is completed by a reconstruction
theorem, by the independence of the reduced dynamics from the choice of principal
connection, and by the relation between criticality and optimality. The reduced
equations are identified with the Hamiltonian equations on the dual of the Atiyah
algebroid, recovering the formulations of Mart\'{\i}nez and of Grabowski and
J\'o\'zwikowski.

We then drop the regularity assumption, so that the control cannot be eliminated and
a constraint algorithm is required. The algorithm of Gotay, Nester and Hinds is given
a Dirac-geometric form in which each stage is the Dirac structure induced by the
constraint distribution, the constraint forces lying in its annihilator, and we prove
that symmetry reduction commutes with the algorithm: it may be run on the reduced
space, the unreduced constraint submanifolds being recovered by pullback. For
control-linear systems whose admissible directions form the horizontal distribution of
a principal connection, the abnormal extremals occupy the leaf on which the base
costate vanishes, and the first secondary constraint is the degeneracy of the
curvature pairing appearing in the reduced costate equation.

The theory is illustrated by three examples: a rigid spacecraft with internal rotors,
for which the connection is flat; the sub-Riemannian geodesic problem on the
Heisenberg group, for which the curvature is nowhere zero; and the Martinet problem,
for which the curvature degenerates on a hypersurface and the abnormal extremals are
exactly the admissible curves lying in it.
\end{abstract}

\maketitle



\markboth{Dirac structures}{Reduction of optimal control problems}


\section{Introduction}
\label{intro}

In this paper we study the role of Dirac structures in the geometric formulation and symmetry reduction of a class of optimal control problems. More specifically, we show how the reduction theory for Dirac structures developed by Yoshimura and Marsden \cite{Yo09} can be extended to optimal control systems admitting a Lie group of symmetries. We prove that, as in the case of implicit Hamiltonian and Lagrangian systems, the reduction of the variational principle and the reduction of the Dirac structure underlying the Pontryagin Maximum Principle (PMP) lead to the same reduced equations, and we then use the Dirac formulation to treat the singular case, in which the control cannot be eliminated and a constraint algorithm is required.

The reduced equations themselves are not new. They are the Hamiltonian equations on the dual of the Atiyah algebroid $TP/G$, obtained in that setting by Mart\'{\i}nez \cite{Mr04} and, in the wider context of almost Lie algebroids, by Grabowski and J\'o\'zwikowski \cite{GJ11}; for a Lie group they go back to Jurdjevic \cite{Jurbook}. Section~\ref{algebroid} records the dictionary explicitly. What the Dirac formulation adds is a framework in which the constraint algorithm and the symmetry reduction interact, and it is there that our main new result lies.

Optimal control problems often possess a rich geometric structure. Throughout this paper we restrict our attention to smooth geometric optimal control problems, where both the state space and the control space are smooth manifolds (assumed, for simplicity, to be without boundary), the controlled dynamics is described by a control-dependent vector field, and the cost functional is defined by a Lagrangian density. This setting naturally admits a differential-geometric description and provides the appropriate framework for studying additional geometric features, such as symmetries and reduction.

Suppose that a Lie group acts on the state and control spaces in such a way that both the dynamics and the cost functional are invariant under the action. We then say that the optimal control problem is symmetric with respect to the given Lie group. The presence of symmetry is a powerful tool for simplifying the analysis of dynamical systems. In Hamiltonian and Lagrangian mechanics, conserved quantities associated with symmetries often allow one to reduce the dimension of the problem. More generally, reduction techniques replace the original system with a family of lower-dimensional systems whose solutions can later be reconstructed to recover the solutions of the original dynamics.

The theory of reduction for systems with symmetry has a long history, and we do not attempt to review it here. Instead, we focus on the work of Yoshimura and Marsden, who developed a unified framework for implicit Lagrangian and Hamiltonian systems based on Dirac structures. Their contributions include the formulation of implicit Lagrangian mechanics using Dirac structures \cite{Yo06a,Yo06b}, the reduction of the canonical Dirac structure on the cotangent bundle of a Lie group \cite{Yo07a}, the treatment of implicit Hamiltonian and Lagrangian systems \cite{Yo07b}, and, ultimately, the reduction of Dirac structures on cotangent bundles \cite{Yo09}. These developments provide the main geometric inspiration for the present work.

The Pontryagin Maximum Principle admits a natural presymplectic formulation, which has been extensively developed in the literature (see, for example, \cite{Ba07,Ba08} and the references therein). Since every presymplectic manifold carries a canonical Dirac structure, the PMP can also be formulated entirely in terms of Dirac geometry. This observation makes it natural to ask whether the reduction theory for Dirac structures can be used to reduce optimal control problems while preserving not only the dynamics but also their variational character.

A straightforward approach would be to perform presymplectic reduction of the Hamiltonian system associated with the PMP. Although this reduces the number of degrees of freedom, it generally loses the variational information encoded in the original optimal control problem. Consequently, the reduced dynamics alone does not guarantee that its solutions correspond to extremals of the original cost functional. What is needed is a reduction procedure that simultaneously preserves the geometric dynamics and the underlying variational principle.

Dirac structures provide precisely such a framework. As shown by Yoshimura and Marsden, the reduction of Dirac structures is compatible with the reduction of the associated variational principles. This makes Dirac geometry particularly well suited to the study of symmetric optimal control problems, where one seeks a reduced formulation that retains both the geometric and variational content of the original problem.

Following this perspective, we extend the Dirac reduction framework of \cite{Yo09} to geometric optimal control. The geometric formulation of the PMP naturally leads to the introduction of the \emph{Pontryagin bundle}, a bundle over the state manifold that plays a role analogous to the Pontryagin bundle $TQ\oplus T^*Q$ arising in implicit Lagrangian mechanics. We formulate the first-order necessary conditions of the PMP as an implicit Hamiltonian system determined by the Pontryagin Hamiltonian and the canonical Dirac structure on this bundle. Assuming the optimal control problem admits a symmetry group $G$, we study the reduction of both the variational principle and the associated Dirac structure.

The reduced PMP equations are obtained in two independent ways. First, we derive them by reducing the variational principle obtained through the method of Lagrange multipliers. Second, we obtain the same equations by directly reducing the Dirac structure defining the PMP, the reduced canonical Dirac structure being derived in Section~\ref{reduced_canonical} rather than postulated. The agreement between these two constructions shows that the reduced Dirac system provides the correct geometric realization of the reduced optimal control problem.

It may be helpful to state at the outset what we claim and what we do not. Our main new result is Theorem~\ref{reduction_commutes}: symmetry reduction commutes with the constraint algorithm, so that for a singular problem the algorithm may be run on the reduced space and the unreduced constraint submanifolds recovered by pulling back. To formulate it we give a Dirac-geometric form of the algorithm, in which each stage is the Dirac structure induced by the constraint distribution and the constraint forces appear in its annihilator. Alongside this we supply the reduction picture in full: reconstruction (Theorem~\ref{reconstruction_theorem}), independence of the choice of principal connection (Theorem~\ref{connection_independent}), and the relation between criticality and optimality (Proposition~\ref{criticality_optimality}). A second new result is Theorem~\ref{strictness_criterion}, which decides in reduced terms whether an abnormal extremal is \emph{strictly} abnormal: it is not, precisely when its projection to the base satisfies Wong's equation for some parallel coadjoint multiplier. The criterion separates cleanly what the curvature alone determines from what the cost contributes, and it is decided by a second constraint algorithm run in the coadjoint fiber (Proposition~\ref{strictness_algorithm}). Section~\ref{index_form} then supplies the second-order theory in the same variables: for an abelian structure group the index form of an abnormal extremal is $\int\langle\mu,F_A(\nabla_{\dot z}v,v)+(\nabla_vF_A)(\dot z,v)\rangle\,dt$ (Theorem~\ref{reduced_index_form}), which involves the curvature and its covariant derivative but not the metric. The abelian hypothesis is not removable: for nonabelian holonomy the second variation acquires time-ordered commutator terms that we do not compute, and Remark~\ref{nonabelian_index} exhibits an $SO(3)$ example in which they destroy positivity on long arcs. Location, rigidity and strictness are then three reduced tests (Theorem~\ref{trichotomy}), the first and third criteria and the second a sufficient certificate, and the deformed Martinet family of \cite{ABCK97} together with the rank-two example of Section~\ref{strictness} shows that the last two are genuinely independent.

By contrast, several ingredients are recastings of known results in the present language, and we have tried to attribute them carefully. The constraint algorithm is that of Gotay, Nester and Hinds \cite{Go79}, applied to the PMP in the presymplectic setting by Barbero-Li\~n\'an and Mu\~noz-Lecanda \cite{Ba08}, who also identified its secondary constraints with the classical Goh and generalized Legendre--Clebsch conditions. The device of adjoining the running cost to the state, which brings abnormal extremals within the same framework, is due to Ohsawa \cite{Oh15} and J\'o\'zwikowski and Respondek \cite{JR16}; we only observe that it is itself an instance of the reduction theory developed here, the cost multiplier being a momentum map. The characterization of abnormal extremals in Section~\ref{abnormal_curvature_sec} is the specialization to a principal bundle of criteria that are classical \cite{AgrachevSachkov,JR16}; what is ours is the observation that the obstruction is exactly the curvature term by which the reduced PMP differs from the unreduced one.

The remainder of the paper is organized as follows. In Section~2 we introduce the geometric framework for the class of optimal control problems under consideration. We review the necessary notions of Pontryagin bundles, Dirac structures, and symmetries, and establish the notation used throughout the paper.

In Section~3 we study the reduction of the variational formulation of the Pontryagin Maximum Principle. We derive the reduced Pontryagin bundle, formulate the reduced variational principle, and obtain the corresponding first-order necessary conditions, which we refer to as the reduced Pontryagin Maximum Principle. The section closes with the four results that complete the reduction picture: reconstruction, independence of the connection, the relation between criticality and optimality, and the comparison with Lie algebroid reduction.

Section~4 approaches the reduction problem from the perspective of Dirac geometry. We construct the reduced Dirac structure associated with the Pontryagin bundle, derive its connection-dependent description, and relate it to the reduction of the canonical Dirac structure on the cotangent bundle. We show that the reduced Dirac formulation yields exactly the same reduced equations obtained from the variational approach.

We then illustrate the theory with two complementary examples. In Section~5 we treat the optimal control of a rigid spacecraft carrying three internal reaction wheels, in which the symmetry group $SO(3)$ is nonabelian but the connection is flat, so that the reduced equations are governed by the coadjoint term. In Section~6 we treat the sub-Riemannian geodesic problem on the Heisenberg group, in which the symmetry group is the abelian center and the coadjoint term vanishes, while the curvature of the natural connection is nowhere zero. Between them the two examples exercise both of the terms by which the reduced Pontryagin Maximum Principle differs from its unreduced counterpart.

Section~7 removes the regularity assumption made in Sections~2--6. It is here that the Dirac formulation does work that the presymplectic one does not: we give the constraint algorithm in Dirac form, prove that reduction commutes with it, and show that for control-linear systems the abnormal extremals occupy the leaf on which the base costate vanishes, the first secondary constraint being the degeneracy of the curvature pairing. The Martinet problem, in which the curvature vanishes on a hypersurface and the abnormal extremals are exactly the curves lying in it, is worked out as a third example, in both its abnormal and its normal leaf. The section closes with the strictness criterion, with an example, of structure group $\mathbb R^2$, in which strictness is decided only at the second stage of the algorithm, and with the reduced index form and the resulting trichotomy.

\section{Geometrical and Variational Aspects of Optimal Control Problems with Symmetry}
\label{ocpsec}

As discussed in the introduction, we restrict our attention to a class of geometric optimal control problems for which the underlying differential-geometric structures can be described in a transparent way. Throughout the paper, all manifolds are assumed to be smooth and without boundary, and all maps, vector fields, differential forms, and curves are assumed to be smooth.

For simplicity, we consider in Sections~\ref{ocpsec}--\ref{heisenberg} only regular optimal control problems, excluding abnormal extremals and singular control systems (see, for example, \cite{AgrachevSachkov,Jurdjevic}); the regularity assumption is removed in Section~\ref{singular}. Moreover, admissible controls are assumed to be smooth. These assumptions allow us to focus on the geometric aspects of the theory without addressing the additional analytical difficulties arising from nonsmooth controls or more general admissible classes. Most of the constructions presented here extend, with appropriate modifications, to more general settings, including systems defined on manifolds with boundary or singular optimal control problems.

We begin by recalling the geometric framework and fixing the notation used throughout the paper.


\subsection{A geometric formulation of optimal control problems}

We follow the standard geometric formulation of optimal control (see, for instance, \cite{AgrachevSachkov,Jurdjevic,BulloLewis}).

\begin{definition}[Control system]
A \emph{control system} is a triple $(P,C,\Gamma)$ consisting of:
\begin{enumerate}
\item[(i)] A smooth manifold $P$, called the \emph{state space}. Local coordinates on $P$ are denoted by $(x^i)$, $i=1,\ldots,n$, and a generic point by $x$.
\item[(ii)] A locally trivial fiber bundle $\pi:C\longrightarrow P,$ called the \emph{control bundle}. For each state $x\in P$, the fiber $C_x=\pi^{-1}(x)$ is the manifold of admissible controls at $x$. Adapted local coordinates on $C$ are denoted by $(x^i,u^a)$, where $a=1,\ldots,r$, and a generic point is written $(x,u)$.
\item[(iii)] A vector field along the projection $\pi:C\longrightarrow P,$ that is, a smooth map
$\Gamma:C\longrightarrow TP$ satisfying $\tau_P\circ\Gamma=\pi,$
where $\tau_P:TP\rightarrow P$ denotes the canonical projection. The map $\Gamma$ describes the controlled dynamics; in local coordinates, $\Gamma=f^i(x,u)\,\partial/\partial x^i$, where natural coordinates on $TP$ are denoted by $(x^i,v^i)$, or equivalently by $(x^i,\dot{x}^i)$. An admissible control is a smooth curve $c:I\longrightarrow C$, defined on an interval $I\subset\mathbb{R}$, with state trajectory $x=\pi\circ c:I\rightarrow P$. Writing locally $c(t)=\bigl(x(t),u(t)\bigr)$, the controlled dynamics is described by
\begin{equation}\label{cs}
\dot{x}^i(t)=f^i\bigl(x(t),u(t)\bigr),
\qquad i=1,\ldots,n.
\end{equation}
\end{enumerate}
\end{definition}
\begin{definition}[Optimal control problem]
\label{ocp}
An \emph{optimal control problem} (OCP) is a quadruple $(P,C,\Gamma,L)$ consisting of
\begin{enumerate}
\item[(i)] A control system $(P,C,\Gamma)$.
\item[(ii)] A smooth cost function $L:C\longrightarrow\mathbb{R}$,
called the \emph{running cost} (or \emph{Lagrangian}), which defines the objective functional
\begin{equation}
\label{of}
S(c)=\int_0^T L(c(t))\,dt
=\int_0^T L\bigl(x(t),u(t)\bigr)\,dt,
\end{equation}
where $c:[0,T]\longrightarrow C$ is an admissible control satisfying $\pi(c(t))=x(t),$ and the associated state trajectory $x(t)$ satisfies the controlled dynamics
\eqref{cs} together with the prescribed boundary conditions $x(0)=x_0$, $x(T)=x_T$.
The optimal control problem consists of finding admissible controls minimizing the functional \eqref{of}.
\end{enumerate}

\end{definition}


\subsection{The presymplectic formulation of optimal control problems}
\label{PMP}

The fundamental tool in optimal control theory is the Pontryagin Maximum Principle (PMP), which provides first-order necessary conditions for optimality. Throughout this paper we assume that optimal trajectories exist and, until Section~\ref{singular}, restrict our attention to regular optimal control problems.

The geometric formulation of the PMP naturally leads to the introduction of the following bundle.

\begin{definition}[Pontryagin bundle]
The \emph{Pontryagin bundle} associated with the control system $(P,C,\Gamma)$ is the fiber product $M=T^*P\times_PC,$
whose fiber over a point $x\in P$ is $M_x=T_x^*P\times C_x.$
Equivalently,
\[
M=\{(\alpha_x,c_x)\in T^*P\times C\mid
\pi_P(\alpha_x)=\pi(c_x)=x\}.
\]
\end{definition}

The Pontryagin bundle inherits natural local coordinates $(x^i,p_i,u^a)$, induced from the canonical coordinates $(x^i,p_i)$ on $T^*P$ and the adapted coordinates $(x^i,u^a)$ on the control bundle $C$.

We denote by $\mathrm{pr}_P:M\to T^*P$ and $\mathrm{pr}_C:M\to C$ the canonical projections onto the two factors,
\begin{equation}
\label{projections}
\mathrm{pr}_P(x,p,u)=(x,p),
\qquad
\mathrm{pr}_C(x,p,u)=(x,u).
\end{equation}
These maps fit into the commutative diagram
\[
\bfig
\xymatrix{
&M=T^*P\times_PC\ar[dl]_{\mathrm{pr}_P}\ar[dr]^{\mathrm{pr}_C}&\\
T^*P\ar[dr]_{\pi_P}&&C\ar[dl]^{\pi}\\
&P&
}
\efig
\]
The Pontryagin bundle carries a natural presymplectic form obtained by pulling back the canonical symplectic form $\omega_P$ on $T^*P$, namely
\begin{equation}
\label{presympl}
\Omega=\mathrm{pr}_P^{\,*}\omega_P
      =dx^i\wedge dp_i .
\end{equation}
The control vector field $\Gamma$ and the running cost $L$ determine the \emph{Pontryagin Hamiltonian}
\begin{equation}
\label{pont_ham}
H_P(x,p,u)
=\langle p,\Gamma(x,u)\rangle-L(x,u)
=p_if^i(x,u)-L(x,u).
\end{equation}

The presymplectic system $(M,\Omega,H_P)$ provides a geometric formulation of the Pontryagin Maximum Principle. Indeed, its integral curves are precisely the extremals of the optimal control problem $(P,C,\Gamma,L)$. The characteristic distribution of $\Omega$ is
\[
K=\ker\Omega=\ker(\mathrm{pr}_P)_*
=\operatorname{span}\left\{\frac{\partial}{\partial u^a}\; ;\; a=1,\ldots,r\right\},
\]
that is, the vertical bundle of the projection $\mathrm{pr}_P:M\to T^*P$.

The dynamics is determined by the presymplectic equation
\begin{equation}
\label{presym}
i_{X_P}\Omega=dH_P.
\end{equation}

Since $\Omega$ is degenerate, solutions of \eqref{presym} do not exist everywhere. A necessary condition for solvability \cite{Go79} is $\langle Z,dH_P\rangle=0, \forall\,Z\in K,$ which, in local coordinates, becomes
\begin{equation}
\label{constr}
\varphi_a=\frac{\partial H_P}{\partial u^a}=0,
\qquad
a=1,\ldots,r.
\end{equation}

The functions $\varphi_a$ are called the \emph{primary constraints}, and their common zero set
\begin{equation}
\label{M1}
M_1=\left\{(x,p,u)\in M\mid
\varphi_a(x,p,u)=0,\;
a=1,\ldots,r
\right\}
\end{equation}
is the \emph{primary constraint submanifold}. Under the regularity assumptions considered in this paper, $M_1$ is a smooth submanifold of codimension $r$.

Let $i_1:M_1\hookrightarrow M$ denote the canonical inclusion. We write
\begin{equation}
\label{Omega1}
\Omega_1=i_1^*\Omega,
\qquad
H_1=i_1^*H_P,
\end{equation}
for the induced presymplectic form and Hamiltonian on $M_1$, and by
$p_1=\mathrm{pr}_P\circ i_1:M_1\longrightarrow T^*P$ the restriction of $\mathrm{pr}_P$ to $M_1$. Then
\[
p_1^*\omega_P=\Omega_1,
\qquad
K_1:=\ker(p_1)_*=TM_1\cap K\ \subseteq\ \ker\Omega_1,
\]
with equality precisely at those points where $p_{1*}(TM_1)$ is a
symplectic subspace of $T(T^*P)$, hence in particular at every regular point
(see Proposition~\ref{regular_points} below). The inclusion may be strict at singular points: for $n=r=1$ and
$H_P=pu-\tfrac13u^3$ one has $M_1=\{p=u^2\}$ and
$\Omega_1=2u\,dx\wedge du$, so that at $u=0$ the radical
$\ker\Omega_1=TM_1$ is two-dimensional whereas $K_1=\mathrm{span}\{\partial_u\}$.

Where the differentials $d\varphi_1,\dots,d\varphi_r$ are linearly independent, so that $TM_1=\ker d\varphi$, a vector $C^a\,\partial/\partial u^a$ of $K$ is tangent to $M_1$ if and only if $C^aW_{ab}=0$, where
\[
W=W_{ab}\,du^a\otimes du^b,
\qquad
W_{ab}
=
\frac{\partial^2H_P}
{\partial u^a\partial u^b},
\]
is the symmetric bilinear form along the fibers of $\mathrm{pr}_P$; hence $K_1(\xi)\neq\{0\}$ if and only if $\det W_{ab}(\xi)=0$. We say that a point $\xi\in M_1$ is \emph{regular} if $W$ is nondegenerate at $\xi$, and \emph{singular} otherwise. The following characterization is proved in \cite{De03}.
\begin{proposition}
\label{regular_points}
Let $\xi\in M_1$ and assume that the differentials $d\varphi_1,\dots,d\varphi_r$ are linearly independent at $\xi$, so that near $\xi$ the set $M_1$ is an
embedded submanifold of $M$ of codimension $r$ with $TM_1=\ker d\varphi$ (this is
automatic at a regular point, since the $du$-block of the matrix of the $d\varphi_b$ is $W$). Then the following statements are equivalent:
(i) $\xi$ is a regular point;
(ii) $K_1(\xi)=\{0\}$, that is, $M_1$ is transverse at $\xi$ to the fibers of $\mathrm{pr}_P$;
(iii) $\ker\Omega_1(\xi)=\{0\}$;
(iv) $p_1$ is a local symplectomorphism between $(M_1,\Omega_1)$ and $(T^*P,\omega_P)$ in a neighborhood of $\xi$.
In that case the equation $i_X\Omega_1=dH_1$ has a unique local solution $X$ near $\xi$, whose integral curve through $\xi$ is the unique extremal of the optimal control problem through $\xi$.
\end{proposition}
If $\xi$ is a singular point, then $K_1(\xi)\neq\{0\}$ and the presymplectic equation
$i_{X_P}\Omega_1=dH_1$ may fail to admit a local solution, so that there may exist no optimal trajectory passing through $\xi$. Extremals passing through singular points are called \emph{singular arcs}.

Throughout Sections~\ref{ocpsec}--\ref{heisenberg} we restrict our attention to \emph{regular} optimal control problems, namely those for which every point of $M_1$ is regular. The Dirac formulation developed below remains valid in the singular case, which is taken up in Section~\ref{singular}.

Since $\Omega=\mathrm{pr}_P^{\,*}\omega_P$, a solution $X_P$ of \eqref{presym} enters the equation only through its projection $T\mathrm{pr}_P\circ X_P$. The dynamics of the Pontryagin Maximum Principle is therefore naturally regarded as a vector field along the projection $\mathrm{pr}_P:M\longrightarrow T^*P$, defined on the locus $M_1$ on which \eqref{presym} is solvable, that is, as a smooth map $X:M_1\longrightarrow T(T^*P)$ satisfying $\tau_{T^*P}\circ X=\mathrm{pr}_P\circ i_1=p_1$; this is the form in which the dynamics enters Sections~\ref{regular_dirac} and~\ref{ldrpmp}.


\subsection{The Lagrange--Pontryagin principle}

We now describe the variational formulation underlying the presymplectic picture developed above.

Consider the optimal control problem $(P,C,\Gamma,L)$ and let $M=T^*P\times_PC$ be its Pontryagin bundle. The optimal control problem may be regarded as a constrained variational problem, where the functional \eqref{of} is minimized subject to the controlled dynamics \eqref{cs}. Under suitable regularity assumptions, the method of Lagrange multipliers shows that the extremals are precisely the critical points of the functional
$\mathbb{S}:\mathcal{C}_{x_0,x_T}(M)\longrightarrow\mathbb{R}$ defined by
\begin{equation}
\label{LPvariational}
\mathbb{S}(\gamma)
=
\int_0^T
\left(
L(x(t),u(t))
+
\langle p(t),
\dot x(t)-\Gamma(x(t),u(t))
\rangle
\right)\,dt,
\end{equation}
where $\mathcal{C}_{x_0,x_T}(M)$ denotes the space of smooth curves
$\gamma:[0,T]\longrightarrow M$, $\gamma(t)=\bigl(x(t),p(t),u(t)\bigr)$,
such that $x(0)=x_0$ and $x(T)=x_T$.

Although the objective of an optimal control problem is to minimize the functional \eqref{of}, in the present geometric framework we are interested in the extremals of $\mathbb{S}$, assuming that a minimum exists in the class of admissible curves.

Let $\xi\in T_\gamma\mathcal{C}_{x_0,x_T}(M)$ be a variation of the curve $\gamma$. A standard computation (see, for example, \cite{De04}) shows that
\[
d_\xi\mathbb{S}(\gamma)
=
\int_0^T
\left\langle
i_{\dot\gamma(t)}\Omega_{\gamma(t)}
-
dH_P(\gamma(t)),
\xi(t)
\right\rangle
dt.
\]

Thus, the differential of $\mathbb{S}$ is locally represented by the one-form along $\gamma$
\[
i_{\dot\gamma}\Omega-dH_P,
\]
which is precisely the geometric object defining the presymplectic dynamics. Consequently, a curve $\gamma$ is an extremal of $\mathbb{S}$ if and only if it satisfies
\begin{equation}
\label{excur}
i_{\dot\gamma(t)}
\Omega_{\gamma(t)}
=
dH_P(\gamma(t)),
\qquad
t\in[0,T].
\end{equation}

Following Yoshimura and Marsden \cite{Yo09}, a variational principle of this type is usually called the \emph{Hamilton--Pontryagin principle}. Here, however, we shall refer to \eqref{LPvariational} as the \emph{Lagrange--Pontryagin principle}, since it is naturally derived from the Lagrangian cost function and is defined on the Pontryagin bundle; rearranging the integrand as $\langle p,\dot x\rangle-H_P(x,p,u)$ gives the equivalent \emph{Hamilton--Pontryagin} form.
 

\subsection{Dirac structures and optimal control}

The presymplectic and variational formulations developed above admit a natural interpretation in terms of Dirac geometry. We briefly recall the basic definitions needed in the sequel; for further details we refer to \cite{CW88,Co90}.

\begin{definition}[Dirac structure on a vector bundle]
\label{Dirac_vector_bundle}
Let $E\to N$ be a vector bundle. A \emph{Dirac structure} on $E$ is a maximally isotropic subbundle
\[
D\subset E\oplus E^*
\]
with respect to the symmetric bilinear form
\[
\langle\!\langle (u,\alpha),(v,\beta)\rangle\!\rangle
=
\langle\alpha,v\rangle
+
\langle\beta,u\rangle,
\]
where $u,v\in E$ and $\alpha,\beta\in E^*$.
\end{definition}

\begin{definition}[Integrability]
\label{integrability}
A Dirac structure $D\subset TN\oplus T^*N$ on the tangent bundle of a smooth manifold $N$ is \emph{integrable} if its space of sections is closed under the Courant bracket \cite{Co90}
\[
[(X,\alpha),(Y,\beta)]
=
\left(
[X,Y],
\mathcal{L}_X\beta
-
\mathcal{L}_Y\alpha
+
\frac12\,d\bigl(\alpha(Y)-\beta(X)\bigr)
\right);
\]
an integrable Dirac structure on $TN$ is simply called a \emph{Dirac structure on the manifold} $N$.
\end{definition}

Two important classes of Dirac structures are provided by presymplectic and Poisson manifolds. In particular, if $\Omega$ is a closed $2$-form on a smooth manifold $N$, the graph of $\Omega$ defines the Dirac structure
\[
D_\Omega
=
\left\{
(X,\alpha)\in TN\oplus T^*N
\mid
\alpha=i_X\Omega
\right\},
\]
and conversely the integrability of $D_\Omega$ is equivalent to the closedness of $\Omega$. If $D$ is a Dirac structure on $N$ and $H\in C^\infty(N)$, the corresponding dynamics is defined by the vector fields $X$ satisfying $(X,dH)\in D$.

\begin{definition}[Canonical Dirac structure associated with an optimal control problem]
Let $(P,C,\Gamma,L)$ be an optimal control problem and let $M=T^*P\times_PC$ be its Pontryagin bundle endowed with the canonical presymplectic form $\Omega=\mathrm{pr}_P^{\,*}\omega_P$. The associated Dirac structure $D_\Omega\subset TM\oplus T^*M$ will be called the \emph{canonical Dirac structure associated with the optimal control problem}.
\end{definition}

With this definition, the presymplectic equation \eqref{presym} reads $(X_P,dH_P)\in D_\Omega$: the dynamics of the Pontryagin Maximum Principle is the Dirac system on the Pontryagin bundle determined by the canonical Dirac structure and the Pontryagin Hamiltonian, and its integral curves are precisely the extremals of the optimal control problem.

\subsection{Dirac structures for regular optimal control problems}
\label{regular_dirac}

For regular optimal control problems the previous Dirac formulation admits a more intrinsic description, which will play a fundamental role in the reduction process developed later. Let $D_P\subset T(T^*P)\oplus T^*(T^*P)$ denote the canonical Dirac structure determined by the canonical symplectic form $\omega_P$ on $T^*P$. Since the problem is regular, $p_1:M_1\to T^*P$ is a local symplectomorphism (Proposition~\ref{regular_points}) and the presymplectic equation admits a unique solution at every point of $M_1$, so that the dynamics of Section~\ref{PMP} is a vector field $X:M_1\to T(T^*P)$ along $p_1$. Likewise, since $p_1^*:T^*(T^*P)\to T^*M_1$ is an isomorphism of vector bundles, the differential of $H_1=i_1^*H_P$ may be regarded as a covector on $T^*P$ along $p_1$,
\[
dH_1
=
(p_1^*)^{-1}(i_1^*dH_P):M_1\longrightarrow T^*(T^*P),
\]
and we use the symbol $dH_1$ for both; the constraint manifold $M_1$ thus plays the role of the graph of the Legendre transformation in the implicit Lagrangian formalism of Yoshimura and Marsden \cite{Yo09}. The implicit Dirac formulation of the regular optimal control problem is
\[
(X,dH_1)\in D_P
\qquad
\text{on }M_1,
\]
both entries being taken along $p_1$. This formulation is completely equivalent to the presymplectic description obtained after imposing the primary constraints \eqref{constr}.

\subsection{Symmetries of optimal control problems}
\label{symmetric_ocp}

We now introduce the notion of symmetry for an optimal control problem. Let $(P,C,\Gamma,L)$ be an optimal control problem, and let $\pi:C\longrightarrow P$ be the corresponding control bundle.

Suppose that a Lie group $G$ acts on the right on the bundle $C$ by fiber-preserving transformations. We denote this action by
$\Psi:C\times G\longrightarrow C$, $(c,g)\longmapsto c\cdot g$,
and the induced action on the state manifold by
$\Phi:P\times G\longrightarrow P$, $(x,g)\longmapsto x\cdot g$;
writing $\Psi_g(c)=\Psi(c,g)$ and $\Phi_g(x)=\Phi(x,g)$, the fact that the action on $C$ is fibered means that
$\pi\circ\Psi_g=\Phi_g\circ\pi$ for all $g\in G$.

\begin{definition}\label{ocp_symmetry}
An optimal control problem $(P,C,\Gamma,L)$ is said to be \emph{$G$-invariant}, or equivalently to admit $G$ as a symmetry group, if the fibered action $(\Phi,\Psi)$ satisfies

\begin{enumerate}
\item[(i)] the running cost is invariant,
\[
L\circ\Psi_g=L,
\qquad
\forall\,g\in G;
\]

\item[(ii)] the controlled dynamics is equivariant,
\[
T\Phi_g\circ\Gamma
=
\Gamma\circ\Psi_g,
\qquad
\forall\,g\in G.
\]
\end{enumerate}
\end{definition}

The action of $G$ on $P$ induces, by tangent and cotangent lift, right actions on $TP$ and $T^*P$, namely $v\cdot g=T\Phi_g(v)$ for $v\in T_xP$ and $p\cdot g=(T\Phi_{g^{-1}})^*(p)$ for $p\in T_x^*P$. Consequently, $G$ acts naturally on the Pontryagin bundle
$M=T^*P\times_PC$ according to
\[
(x,p,u)\cdot g
=
(x\cdot g,\,
p\cdot g,\,
u\cdot g)
=
\left(
\Phi_g(x),
(T\Phi_{g^{-1}})^*p,
\Psi_g(u)
\right).
\]

The presymplectic formulation of the optimal control problem is naturally compatible with the action of the symmetry group.

\begin{proposition}
\label{invariance_prop}
Let $(P,C,\Gamma,L)$ be a $G$-invariant optimal control problem. Then the induced action of $G$ on the Pontryagin bundle $M=T^*P\times_PC$ preserves both the canonical presymplectic form
$\Omega=\mathrm{pr}_P^{\,*}\omega_P$ and the Pontryagin Hamiltonian $H_P$; in particular,
\begin{equation}
\label{Hinvar}
H_P\bigl((x,p,u)\cdot g\bigr)
=
H_P(x,p,u),
\qquad
\forall\,g\in G.
\end{equation}
Moreover, the action admits the equivariant momentum map
$J_M=\mathrm{pr}_P^{\,*}J_P:M\longrightarrow\mathfrak{g}^*$, $\langle J_M(x,p,u),\xi\rangle=\langle p,\xi_P(x)\rangle$,
where $J_P:T^*P\to\mathfrak g^*$ is the canonical momentum map of the cotangent-lifted action and $\xi_P$ the infinitesimal generator on $P$; it satisfies
$i_{\xi_M}\Omega=d\langle J_M,\xi\rangle$ for all $\xi\in\mathfrak{g}$, where $\xi_M$ denotes the infinitesimal generator of the action on $M$.
\end{proposition}

\begin{proof}
Cotangent lifts are symplectic, so $\Omega=\mathrm{pr}_P^{\,*}\omega_P$ is $G$-invariant, and $J_M$ is the pull-back of the canonical momentum map, with which it shares its equivariance. The invariance of $H_P$ follows from the equivariance of $\Gamma$ and the invariance of $L$:
$H_P((x,p,u)\cdot g)=\langle (T\Phi_{g^{-1}})^*p,\,T\Phi_g\Gamma(x,u)\rangle-L(x,u)=H_P(x,p,u)$.
\end{proof}

Throughout the paper we shall assume that the quotient space $P/G$ is a smooth manifold and that the canonical projection $\rho:P\longrightarrow P/G$ is a surjective submersion. Whenever explicit local expressions are required, we shall further assume that the action of $G$ on $P$ is free and proper, so that
$\rho$ is a principal $G$-bundle.

For the same reason, we assume that the control bundle is obtained from an associated bundle. More precisely, let $U$ be a smooth manifold endowed with a left action
$\Sigma:G\times U\longrightarrow U$. The associated bundle over the quotient is
$\widetilde U=P\times_GU=(P\times U)/G$,
and the control bundle is assumed to be its pull-back to $P$, $C\simeq\rho^*\widetilde U$.
Whenever no confusion is possible, we shall use the same notation $\widetilde U$ for its pull-back to $P$ or to $T^*P$. Notice that, as bundles over $P$, $\rho^*\widetilde U\simeq P\times U$,
although the notation $\widetilde U$ is convenient for emphasizing its origin as the pull-back of the associated bundle over $P/G$.


\section{Reduction of the Lagrange--Pontryagin principle}
\label{lp_reduction}

Let $(P,C,\Gamma,L)$ be a $G$-invariant optimal control problem. Since the running cost $L$ is $G$-invariant and the controlled dynamics $\Gamma$ is $G$-equivariant, the Lagrange--Pontryagin functional
$\mathbb{S}:\mathcal{C}(M)\longrightarrow\mathbb{R}$
defined in \eqref{LPvariational} is invariant under the natural action of $G$ on the space of smooth curves in the Pontryagin bundle,
$(\gamma\cdot g)(t)=\bigl(x(t)\cdot g,\;p(t)\cdot g,\;u(t)\cdot g\bigr)$,
that is, $\mathbb{S}(\gamma\cdot g)=\mathbb{S}(\gamma)$ for all $\gamma\in\mathcal{C}(M)$ and $g\in G$.

Two endpoint conditions must be distinguished here, because the space of
curves with genuinely fixed endpoints $x(0)=x_0$, $x(T)=x_T$ is \emph{not}
preserved by the $G$-action: for a free action a nontrivial group element fixes
no endpoint. The only formulation on which $G$ acts is the
\emph{orbit-endpoint} problem, in which $x(0)$ and $x(T)$ are required only to
lie in the orbits $G\!\cdot x_0$ and $G\!\cdot x_T$; there $\mathbb S$ descends
to the quotient of the corresponding curve space. Its endpoint variations are
vertical, so that $\delta z(0)=\delta z(T)=0$ automatically, while the reduced
vertical variations $\eta(0)$ and $\eta(T)$ of Section~\ref{rlpp} remain free
and produce the transversality conditions $\mu(0)=\mu(T)=0$.

The fixed-endpoint problem is treated by reduction as follows. Its admissible
curves are exactly the curves obtained by reconstructing a reduced curve from
the fixed initial lift $x(0)=x_0$ and imposing on it the terminal holonomy
condition $x(T)=x_T$, and its variations are those with
$\delta z(0)=\delta z(T)=0$ \emph{and} $\eta(0)=\eta(T)=0$. The two problems
therefore have the same interior equations, namely the reduced
Lagrange--Pontryagin equations \eqref{RHP} derived below, and differ only in
their boundary conditions: no condition on $\mu$ at the endpoints in the fixed
case, and $\mu(0)=\mu(T)=0$ in the orbit case. Except where stated otherwise we
work with the fixed-endpoint problem, and all the endpoint conditions written
below, in Section~\ref{rlpp} and in Proposition~\ref{criticality_optimality},
are its. Strictly speaking, then, it is the interior variational problem, and
not the fixed-endpoint functional itself, that is $G$-invariant; this is what is
used in what follows.

On the unconstrained curve space $\mathcal C(M)$, and likewise on the
orbit-endpoint curve space, the functional $\mathbb S$ therefore descends to the quotient $\mathcal{C}(M)/G$. We denote the induced functional by $[\mathbb{S}]_G$ and refer to it as the \emph{reduced Lagrange--Pontryagin functional}. Since the quotient map $\mathcal{C}(M)\rightarrow\mathcal{C}(M)/G$ identifies curves belonging to the same $G$-orbit, the critical points of $[\mathbb{S}]_G$ are in one-to-one correspondence with $G$-equivalence classes of extremals of $\mathbb{S}$. For the fixed-endpoint problem the corresponding statement is the one supplied by Theorem~\ref{reconstruction_theorem}: reduced solutions correspond to extremals through a prescribed initial point $x_0$, and the terminal condition $x(T)=x_T$ is imposed afterwards.

Our goal is to characterize these critical points in terms of the geometry of the reduced Pontryagin bundle. The corresponding Euler--Lagrange equations on the quotient space will be called the \emph{reduced Lagrange--Pontryagin equations}. Since they represent the reduction of the first-order necessary conditions of the Pontryagin Maximum Principle, we shall also refer to them simply as the \emph{reduced Pontryagin Maximum Principle}.


\subsection{Canonical reduced bundles}
\label{reduced_bundles}

We begin by introducing several canonical bundles naturally associated with the symmetry reduction, which will be used throughout the remainder of this section.

Unless otherwise stated, we assume that the right action of the Lie group $G$ on $P$ is free and proper. Hence the quotient
$B=P/G$
is a smooth manifold and the canonical projection
$\rho:P\rightarrow B$
defines a principal $G$-bundle. The orbit through a point $x\in P$ will be denoted either by $x\!\cdot G$ or by
$[x]=\{xg\mid g\in G\}$,
and we shall write $z=[x]\in B$. Local coordinates on $B=P/G$ will be denoted by $(z^\alpha)$, $\alpha=1,\ldots,\dim B$.


\subsubsection{The adjoint and coadjoint bundles}

Associated with the principal $G$-bundle $\rho:P\rightarrow B$ are the adjoint and coadjoint vector bundles, corresponding to the adjoint and coadjoint representations of $G$ on its Lie algebra $\mathfrak g$ and its dual $\mathfrak g^*$, respectively.

The adjoint bundle is the associated vector bundle
$\operatorname{Ad}P=P\times_G\mathfrak g=(P\times\mathfrak g)/G$,
where the quotient is taken with respect to the left action
\begin{equation}
\label{adjoint_action}
g\cdot(x,\xi)=\bigl(xg^{-1},\operatorname{Ad}_g\xi\bigr).
\end{equation}
Following the notation introduced in Section~\ref{symmetric_ocp}, we shall denote this bundle by $\widetilde{\mathfrak g}$. An element of $\widetilde{\mathfrak g}$ will be written simply as $(z,\xi)$, where $z=[x]\in B$.

Similarly, the coadjoint bundle is the associated vector bundle $\operatorname{Ad}^*P=P\times_G\mathfrak g^*$, where the quotient is defined by the left action
\begin{equation}
\label{coadjoint_action}
g\cdot(x,\mu)=\bigl(xg^{-1},\operatorname{Ad}_{g^{-1}}^*\mu\bigr).
\end{equation}
Since the coadjoint representation is the dual of the adjoint representation, $\operatorname{Ad}^*P$ is naturally the dual vector bundle of $\operatorname{Ad}P$. Accordingly, we shall denote it by $\widetilde{\mathfrak g}^{\,*}$; equivalently,
\[
\widetilde{\mathfrak g}^{\,*}\cong\widetilde{\mathfrak g^*}.
\]

\subsubsection{The principal bundles $TP$ and $T^*P$}

Since the action of $G$ on $P$ is free and proper, its tangent lift $(x,v_x)\cdot g=(xg,T\Phi_g(v_x))$ and its cotangent lift $(x,p_x)\cdot g=\bigl(xg,(T\Phi_{g^{-1}})^*p_x\bigr)$ are free and proper right actions, so that $TP\rightarrow TP/G$ and $T^*P\rightarrow T^*P/G$ are again principal $G$-bundles, and the tangent map $T\rho:TP\rightarrow T(P/G)$ covers $\rho$. The tangent group $TG$ is identified, by left trivialization $(g,v_g)\mapsto(g,TL_{g^{-1}}v_g)$, with the semidirect product $G\ltimes\mathfrak g$, whose multiplication is $(g,\xi)(h,\zeta)=(gh,\operatorname{Ad}_{h^{-1}}\xi+\zeta)$, and the induced right action of $TG$ on $TP$ reads
\begin{equation}
\label{tangent_action}
(x,v_x)\cdot(g,\xi)
=
\bigl(xg,\,
T\Phi_g(v_x)
+
\xi_P(xg)
\bigr),
\end{equation}
where $\xi_P$ denotes the fundamental vector field on $P$ generated by $\xi\in\mathfrak g$: for any curve with $g(0)=g$ and $TL_{g^{-1}}\dot g(0)=\xi$ one has
$\frac{d}{dt}\big|_{t=0}x\,g(t)=\frac{d}{dt}\big|_{t=0}(xg)\exp(t\xi)=\xi_P(xg)$.
With this convention \eqref{tangent_action} is a right action of the group above, since $T\Phi_h\bigl(\xi_P(xg)\bigr)=(\operatorname{Ad}_{h^{-1}}\xi)_P(xgh)$.

\subsubsection{The reduced tangent and cotangent bundles}
\label{atiyah_bundle}
\label{weinstein_space}

The quotient $TP/G$ is the middle term of the \emph{Atiyah exact sequence} of vector bundles over $B=P/G$,
\begin{equation}
\label{atiyah_sequence}
0\longrightarrow\widetilde{\mathfrak g}\longrightarrow TP/G\stackrel{T\rho}{\longrightarrow}TB\longrightarrow0 ,
\end{equation}
and is therefore canonically an affine bundle over $T(P/G)$ modeled on the pullback $\tau_{P/G}^*(\widetilde{\mathfrak g})$ of the adjoint bundle, where $\tau_{P/G}:T(P/G)\rightarrow P/G$ is the tangent bundle projection: two classes $[x,v_x]$ and $[x,w_x]$ lying over the same $v_z\in T_zB$ differ by the vertical class $[x,v_x-w_x]\in\widetilde{\mathfrak g}_z$, but no class over $v_z$ is distinguished. Its identification with $T(P/G)\oplus\widetilde{\mathfrak g}$ requires a splitting of \eqref{atiyah_sequence}, that is, a principal connection on $\rho:P\rightarrow P/G$, and is therefore not canonical. As before, we shall use the same notation $\widetilde{\mathfrak g}$ for the adjoint bundle and its pullback whenever no confusion is likely to arise.

Let $A$ be a principal connection on $P$, and consider the map
\[
\alpha_A:TP\longrightarrow T(P/G)\times(P\times\mathfrak g),
\qquad
\alpha_A(x,v_x)=\bigl(T\rho(v_x),\,x,\,A_x(v_x)\bigr).
\]
The group $G$ acts trivially on the factor $T(P/G)$ and on $P\times\mathfrak g$ by the adjoint action \eqref{adjoint_action}. Since $\alpha_A$ is $G$-equivariant, it descends to the quotient, giving rise to the vector bundle isomorphism
\begin{equation}
\label{tangent_quotient}
[\alpha_A]:
TP/G
\longrightarrow
T(P/G)\oplus\widetilde{\mathfrak g},
\qquad
[\alpha_A]\bigl([x,v_x]\bigr)
=
\bigl(T\rho(v_x),[x,A_x(v_x)]\bigr),
\end{equation}
under which the class $[x,v_x]$ is represented by the triple $(z,v_z,\zeta)$, where $z=[x]\in P/G$, $v_z=T\rho(v_x)\in T_z(P/G)$ and $\zeta=[x,A_x(v_x)]\in\widetilde{\mathfrak g}_z$.

By duality, the connection identifies the quotient $T^*P/G$, the \emph{Weinstein space} \cite{We78}, with $T^*(P/G)\oplus\widetilde{\mathfrak g}^{\,*}$, which is a vector bundle over $T^*(P/G)$, namely the pullback of the coadjoint bundle $\widetilde{\mathfrak g}^{\,*}$ by the cotangent bundle projection $\pi_{P/G}:T^*(P/G)\rightarrow P/G$; as before, we shall use the notation $T^*(P/G)\oplus\widetilde{\mathfrak g}^{\,*}$ for this bundle whenever no confusion may arise. Explicitly, let
$H_x=\ker A_x\subset T_xP$
be the horizontal subspace determined by the connection. Since the restriction of
$T\rho_x:H_x\rightarrow T_{[x]}(P/G)$
is an isomorphism, it induces an identification
$H_x^*\cong T_{[x]}^*(P/G)$, and the decomposition $T_xP=H_x\oplus V_x$, where $V_x$ denotes the vertical subspace, defines a projection
$\widetilde{\pi}_G:T^*P\longrightarrow T^*(P/G)$, $\widetilde{\pi}_G(x,p_x)=(z,p_z)$,
obtained by restricting covectors to $H_x$ and identifying them with covectors on $T_{[x]}(P/G)$.
The momentum map $J:T^*P\longrightarrow\mathfrak g^*$
associated with the cotangent lift of the $G$-action on $P$ then provides the second component of the identification. More precisely, consider the map
\[
\beta_J:T^*P\longrightarrow T^*(P/G)\times(P\times\mathfrak g^*),
\qquad
\beta_J(x,p_x)=\bigl(z,p_z;\,x,J(x,p_x)\bigr).
\]
Since $\beta_J$ is $G$-equivariant, where $G$ acts on $P\times\mathfrak g^*$ by the coadjoint action \eqref{coadjoint_action}, it descends to the quotient, yielding the vector bundle isomorphism
\begin{equation}
\label{cotangent_quotient}
[\beta_J]:
T^*P/G
\longrightarrow
T^*(P/G)\oplus\widetilde{\mathfrak g}^{\,*},
\qquad
[\beta_J]\bigl([x,p_x]\bigr)
=
\bigl(z,p_z;\,[x,J(x,p_x)]\bigr).
\end{equation}

Accordingly, an element of $T^*P/G$ will be represented by a triple
$(z,p_z,\mu)$, where
$(z,p_z)\in T^*(P/G)$ and
$\mu\in\widetilde{\mathfrak g}^{\,*}_z$.
\subsubsection{The reduced control bundle}

If $G$ is a symmetry group of the optimal control problem $(P,C,\Gamma,L)$, the quotient
$\widetilde C=C/G$ is a bundle over the orbit space $P/G$ with projection
$\widetilde\pi([x,u])=[x]$. As explained in Section~\ref{symmetric_ocp}, we assume that the control bundle
$\pi:C\rightarrow P$ is the pullback of an associated bundle
$\widetilde U=P\times_GU\rightarrow P/G$, where $G$ acts on the typical fiber $U$, so that $C=\rho^*(\widetilde U)\cong P\times U$ (a trivialization which is not $G$-equivariant), and passing to the quotient yields the associated bundle
\begin{equation}
\label{control_quotient}
C/G\cong\widetilde U,
\qquad
[x,u]\longmapsto[x,u],
\end{equation}
which will be regarded as the reduced control bundle over $P/G$. The right-hand side is the class of $(x,u)$ in $P\times_GU$ and not the pair $([x],u)$, which for a nontrivial action on $U$ is defined only once a representative $x$ of $z=[x]$ has been chosen. We write $\widetilde u\in\widetilde U_z$ for the reduced control, and revert to the notation $u$ only in a local trivialization of $\widetilde U$ or when the action on $U$ is trivial.

Unless the contrary is stated, $U$ is only a smooth $G$-manifold, so that
$\widetilde U$ is a fiber bundle over $B=P/G$ and not a vector bundle. From
Section~\ref{drec} onwards, and in particular in Sections~\ref{strictness} and~\ref{index_form}, we assume in addition
that $U$ is a vector space carrying a linear $G$-action, in which case
$\widetilde U$ is an associated vector bundle.


\subsubsection{The reduced Pontryagin bundle}
\label{MG}

Since the action of $G$ on the Pontryagin bundle $M=T^*P\times_PC$ covers the action on $P$, the quotient $M/G$ is naturally a bundle over $P/G$, and combining the identifications
\eqref{cotangent_quotient} and \eqref{control_quotient} we obtain the
isomorphism of fiber bundles over $B=P/G$
\begin{equation}
\label{reduced_pontryagin}
M/G
\cong
\bigl(T^*(P/G)\oplus\widetilde{\mathfrak g}^{\,*}\bigr)\times_B\widetilde U .
\end{equation}
The fiber product $\times_B$ reduces to the Whitney sum $T^*(P/G)\oplus\widetilde{\mathfrak g}^{\,*}\oplus\widetilde U$ when $U$ is a vector space with a linear $G$-action, and we write the right-hand side of \eqref{reduced_pontryagin} additively throughout, since that hypothesis is in force from Section~\ref{drec} onwards; in Section~\ref{rlpp}, where it is not yet needed, the additive notation is to be read as the fiber product, no operation being performed on the third factor. An element of the reduced Pontryagin bundle will be represented by a quadruple
$(z,p_z,\mu,\widetilde u)$,
where $(z,p_z)\in T^*(P/G)$, $\mu\in\widetilde{\mathfrak g}^{\,*}_z$, and $\widetilde u\in\widetilde U_z$.



\subsection{The reduced Lagrange--Pontryagin principle}
\label{rlpp}

We are now in a position to derive the reduced Lagrange--Pontryagin variational principle.

Since the running cost $L$ is $G$-invariant, it descends to a smooth function
$l:\widetilde U=C/G\rightarrow\mathbb R$, given by
$l(z,u)=L(x,u)$, where $z=[x]\in P/G$. This definition is independent of the chosen representative because of the $G$-invariance of $L$.

Similarly, the equivariance of the control vector field
$\Gamma:C\rightarrow TP$
induces a bundle map
$[\Gamma]:\widetilde U\rightarrow TP/G$.
After choosing a principal connection and using the identification
$TP/G\cong T(P/G)\oplus\widetilde{\mathfrak g}$,
the reduced dynamics decomposes as
\[
[\Gamma](z,u)=\bigl(\widetilde\Gamma(z,u),\widetilde\gamma(z,u)\bigr),
\]
where $\widetilde\Gamma:\widetilde U\rightarrow T(P/G)$ is a vector field along the bundle projection and $\widetilde\gamma:\widetilde U\rightarrow\widetilde{\mathfrak g}$ is a section of the adjoint bundle. In terms of the connection,
\[
\widetilde\gamma(z,u)=[x,A(\Gamma(x,u))],
\]
which is well defined by the equivariance of the connection form.

Using the decompositions
$TP/G\cong T(P/G)\oplus\widetilde{\mathfrak g}$ and
$T^*P/G\cong T^*(P/G)\oplus\widetilde{\mathfrak g}^{\,*}$,
the density of the Lagrange--Pontryagin functional becomes
\begin{equation}
\label{red_lag}
\mathbb L=
l(z,u)
+\langle p_z,\dot z-\widetilde\Gamma(z,u)\rangle
+\langle\mu,\xi-\widetilde\gamma(z,u)\rangle,
\end{equation}
where $\dot z=T\rho(\dot x)$, $\xi=A(\dot x)$ and
$p_x\leftrightarrow(p_z,\mu)$ under the above identification.

The reduced Lagrange--Pontryagin principle is formulated in terms of a curve
$\bigl(z(t),p_z(t),\mu(t),\widetilde u(t)\bigr)$
in the reduced Pontryagin bundle
$M/G\cong T^*(P/G)\oplus\widetilde{\mathfrak g}^{\,*}\oplus\widetilde U$,
together with the reduced velocity
$(\dot z(t),\xi(t))\in T(P/G)\oplus\widetilde{\mathfrak g}$,
where $\xi$ is the vertical component of the velocity determined by the
chosen principal connection. The reduced Lagrange--Pontryagin functional is therefore
\begin{equation}
\label{red_funct}
S_G
=
\int_0^T
\left[
l(z,u)
+
\left\langle p_z,
\dot z-\widetilde\Gamma(z,u)
\right\rangle
+
\left\langle\mu,
\xi-\widetilde\gamma(z,u)
\right\rangle
\right]dt.
\end{equation}
Equivalently,
\begin{equation}
\label{red_funct_hamiltonian}
S_G
=
\int_0^T
\left[
\langle p_z,\dot z\rangle
+
\langle\mu,\xi\rangle
-
h(z,p_z,\mu,u)
\right]dt,
\end{equation}
where the reduced Pontryagin Hamiltonian is
\begin{equation}
\label{reduced_pontryagin_hamiltonian}
h(z,p_z,\mu,u)
=
\langle p_z,\widetilde\Gamma(z,u)\rangle
+
\langle\mu,\widetilde\gamma(z,u)\rangle
-
l(z,u).
\end{equation}

The variations of the base variables satisfy $\delta z(0)=\delta z(T)=0$,
while the variations $\delta p_z$, $\delta\mu$ and $\delta u$ are arbitrary.
The variation of the reduced vertical velocity $\xi$ is constrained by the
Lagrange--Poincaré reduction \cite{CMR01}: if
$\eta(t)\in\widetilde{\mathfrak g}_{z(t)}$
denotes the reduced vertical variation, with $\eta(0)=\eta(T)=0$, then
\begin{equation}
\label{reduced_vertical_variation}
\delta\xi
=
\frac{D\eta}{Dt}
+
[\eta,\xi]
+
F_A(\delta z,\dot z),
\end{equation}
where $F_A$ denotes the curvature of the chosen principal connection $A$.

Two covariant derivatives appear in the computation that follows, and we fix
them now. The principal connection $A$ induces a linear connection $\nabla^A$
on the adjoint bundle $\widetilde{\mathfrak g}$ and, by duality, on
$\widetilde{\mathfrak g}^{\,*}$; the symbols $D\eta/Dt$ and $D\mu/Dt$ denote
the corresponding covariant derivatives along $z(\cdot)$. In a local
trivialization over a chart $(z^\alpha)$ on $B$, with $\{e_i\}$ a local frame
of $\widetilde{\mathfrak g}$ and
$\nabla^A_{\partial_\alpha}e_i=\Gamma^k_{\alpha i}\,e_k$,
\[
\Bigl(\frac{D\eta}{Dt}\Bigr)^{\!k}
=\dot\eta^k+\Gamma^k_{\alpha i}\,\dot z^\alpha\eta^i,
\qquad
\Bigl(\frac{D\mu}{Dt}\Bigr)_{\!i}
=\dot\mu_i-\Gamma^k_{\alpha i}\,\mu_k\dot z^\alpha .
\]
The connection $A$ does not, however, induce a connection on the base
$B=P/G$, and the symbols $Dp_z/Dt$, $\delta\dot z$ and $\partial h/\partial z$
below require one. We therefore fix, once and for all, an auxiliary
torsion-free linear connection $\nabla$ on $B$ (any one will do, for instance
the Levi-Civita connection of a Riemannian metric; in
Sections~\ref{strictness} and~\ref{index_form} it will be the Levi-Civita
connection of the reduced metric), and we understand: $Dp_z/Dt$ as the
covariant derivative of $p_z(\cdot)$ along $z(\cdot)$ with respect to the
connection induced by $\nabla$ on $T^*B$; $\delta\dot z$ as the covariant
variation $D(\delta z)/Dt$, which agrees with the covariant
$\varepsilon$-derivative of $\dot z$ because $\nabla$ is torsion-free; and
$\partial h/\partial z$ as the covariant partial derivative of $h$, namely the
derivative of $h$ along $\nabla\oplus\nabla^A$-parallel translation of the
fiber variables $(p_z,\mu)$, the control $u$ being held fixed in a local
trivialization of $\widetilde U$. The fiber derivatives
$\partial h/\partial p_z$, $\partial h/\partial\mu$ and
$\partial h/\partial u$ require no connection. This is the same device as in
the Hamilton--Poincaré equations of Cendra, Marsden, Pekarsky and Ratiu
\cite{CMPR03}. The choice of $\nabla$ affects the two symbols separately but
not the equations: under a change $\nabla\mapsto\nabla+S$, with $S$ a
symmetric $(1,2)$-tensor field, $Dp_z/Dt$ changes by
$-\langle p_z,S(\dot z,\cdot)\rangle$ and $\partial h/\partial z$ by
$+\langle p_z,S(\cdot\,,\partial h/\partial p_z)\rangle$, and the two changes
cancel in the reduced equations \eqref{RHP} below by the symmetry of $S$,
along the first equation $\dot z=\partial h/\partial p_z$.


\subsection{The reduced Lagrange--Pontryagin equations}
\label{pmpreduc}

We now derive the reduced Pontryagin Maximum Principle directly from the
reduced variational principle. Taking the variation of
\eqref{red_funct_hamiltonian}, we obtain
\begin{align*}
\delta S_G
=
\int_0^T
\Bigg[
&
\langle\delta p_z,\dot z\rangle
+
\langle p_z,\delta\dot z\rangle
+
\langle\delta\mu,\xi\rangle
+
\langle\mu,\delta\xi\rangle
\\
&
-
\left\langle
\frac{\partial h}{\partial z},
\delta z
\right\rangle
-
\left\langle
\frac{\partial h}{\partial p_z},
\delta p_z
\right\rangle
-
\left\langle
\frac{\partial h}{\partial\mu},
\delta\mu
\right\rangle
-
\left\langle
\frac{\partial h}{\partial u},
\delta u
\right\rangle
\Bigg]dt.
\end{align*}

We substitute \eqref{reduced_vertical_variation} for $\delta\xi$, integrate by parts the terms containing $\delta\dot z$ and $D\eta/Dt$, which is legitimate since $\delta z$ and $\eta$ vanish at the endpoints, and use
$\langle\mu,[\eta,\xi]\rangle=-\langle\operatorname{ad}_{\xi}^{*}\mu,\eta\rangle$,
which is our convention for the coadjoint action. Collecting all terms, we obtain
\begin{align*}
\delta S_G
=
\int_0^T
\Bigg[
&
\left\langle
\dot z-\frac{\partial h}{\partial p_z},
\delta p_z
\right\rangle
+
\left\langle
\xi-\frac{\partial h}{\partial\mu},
\delta\mu
\right\rangle
\\
&
-
\left\langle
\frac{Dp_z}{Dt}
+
\frac{\partial h}{\partial z},
\delta z
\right\rangle
+
\langle\mu,F_A(\delta z,\dot z)\rangle
\\
&
-
\left\langle
\frac{D\mu}{Dt}
+
\operatorname{ad}_{\xi}^{*}\mu,
\eta
\right\rangle
-
\left\langle
\frac{\partial h}{\partial u},
\delta u
\right\rangle
\Bigg]dt.
\end{align*}

Since $\delta z$, $\delta p_z$, $\delta\mu$, $\delta u$ and $\eta$ are
arbitrary, the stationarity condition $\delta S_G=0$ yields
\begin{equation}
\label{RHP}
\left\{
\begin{aligned}
\dot z
&=
\frac{\partial h}{\partial p_z},
\\[0.6em]
\frac{Dp_z}{Dt}
&=
-\frac{\partial h}{\partial z}
+
\left\langle
\mu,F_A(\,\cdot\,,\dot z)
\right\rangle,
\\[0.6em]
\frac{D\mu}{Dt}
&=
-\operatorname{ad}_{\xi}^{*}\mu,
\\[0.6em]
\xi
&=
\frac{\partial h}{\partial\mu},
\\[0.6em]
0
&=
\frac{\partial h}{\partial u}.
\end{aligned}
\right.
\end{equation}

These equations constitute the reduced first-order necessary conditions of
the Pontryagin Maximum Principle; by the antisymmetry of the curvature two-form, the curvature term in the second equation may equally be written as $-\langle\mu,F_A(\dot z,\,\cdot\,)\rangle$.

\begin{remark}[the left-action convention]
\label{left_action_variational}
Suppose that $G$ acts on $P$ on the left. The connection is then normalized by
$A(\zeta_P)=\zeta$ with
$\zeta_P(x)=\frac{d}{dt}\big|_{t=0}\exp(t\zeta)\cdot x$, the adjoint bundle
carries the equivalence $(x,\zeta)\sim(g\cdot x,\operatorname{Ad}_g\zeta)$, the
induced covariant derivative on $\widetilde{\mathfrak g}$ reads
$D/Dt=\partial_t-\operatorname{ad}_\xi$, and the structure equation becomes
$F_A=dA-\tfrac12[A,A]$, since the orbit map now pulls $A$ back to the right,
rather than the left, Maurer--Cartan form. With these conventions the
constrained variations \eqref{reduced_vertical_variation} read
$\delta\xi=D\eta/Dt+[\xi,\eta]+F_A(\delta z,\dot z)$, and the reduced Lagrange--Pontryagin equations are \eqref{RHP} with the sign of the coadjoint equation reversed,
\begin{equation}
\label{RHP_left}
\frac{D\mu}{Dt}
=
+\operatorname{ad}_{\xi}^{*}\mu ;
\end{equation}
the other four equations, and in particular the curvature term, are unaffected,
since the latter arises from $\int_0^T\langle\mu,F_A(\delta z,\dot z)\rangle\,dt$ in
either convention. This is the form taken by the reduced
Pontryagin Maximum Principle in body variables, and is consistent with the
implicit Lie--Poisson equations of Yoshimura and Marsden
\cite[Prop.~6.5]{Yo07a}.
\end{remark}


\subsection{Reconstruction}
\label{reconstruction}

Reduction is only half of the correspondence. We now show that a solution of the
reduced equations \eqref{RHP}, together with an initial condition in $P$,
determines a unique extremal of the original optimal control problem, so that no
information is lost in passing to the quotient.

Let $(z(t),p_z(t),\mu(t),u(t))$ be a solution of \eqref{RHP} on $[0,T]$, and let
$\xi(t)=\partial h/\partial\mu\in\widetilde{\mathfrak g}_{z(t)}$ be the
associated reduced vertical velocity. Fix $x_0\in P$ with $\rho(x_0)=z(0)$.

\begin{definition}
\label{def_horizontal_lift}
The \emph{horizontal lift} of $z(\cdot)$ through $x_0$ is the unique curve
$x^h(\cdot)$ in $P$ with
\[
\rho\bigl(x^h(t)\bigr)=z(t),
\qquad
A\bigl(\dot x^h(t)\bigr)=0,
\qquad
x^h(0)=x_0 .
\]
Along it, write $\xi^h(t)\in\mathfrak g$ for the representative of $\xi(t)$
determined by $x^h(t)$, that is, $\xi(t)=\bigl[x^h(t),\xi^h(t)\bigr]$.
\end{definition}

\begin{theorem}[Reconstruction]
\label{reconstruction_theorem}
Let $g(\cdot)$ be the solution of the initial value problem in $G$
\begin{equation}
\label{reconstruction_equation}
\dot g(t)=T R_{g(t)}\,\xi^h(t),
\qquad
g(0)=e ,
\end{equation}
and set $x(t)=x^h(t)\cdot g(t)$. Then $x(\cdot)$ is the unique curve in $P$ with
$x(0)=x_0$ whose class in $TP/G\cong T(P/G)\oplus\widetilde{\mathfrak g}$ is
$(\dot z(t),\xi(t))$. Let $\hat\mu(t)\in\mathfrak g^*$ be the representative of
$\mu(t)$ determined by $x(t)$, that is, $\mu(t)=[x(t),\hat\mu(t)]$, and let
$c(t)\in C_{x(t)}$ be the representative of the reduced control
$\widetilde u(t)\in\widetilde U_{z(t)}$ determined by $x(t)$ under
\eqref{control_quotient}. Define the costate $p(t)\in T^*_{x(t)}P$ by
\begin{equation}
\label{costate_reconstruction}
\bigl\langle p(t),v\bigr\rangle
=
\bigl\langle p_z(t),T\rho(v)\bigr\rangle
+
\bigl\langle\hat\mu(t),A(v)\bigr\rangle,
\qquad
v\in T_{x(t)}P,
\end{equation}
which is the unique covector over $x(t)$ whose class is
$[\beta_J]^{-1}(z(t),p_z(t),\mu(t))\in T^*P/G$. Then the curve
$\bigl(x(t),p(t),c(t)\bigr)$ is an extremal of the unreduced Pontryagin Maximum
Principle issuing from $x_0$; it solves the fixed-endpoint problem with terminal
state $x_T$ precisely when the terminal holonomy condition $x(T)=x_T$ holds,
which is a condition on $z(\cdot)$, $\xi(\cdot)$ and $x_0$ and not merely on
$z(T)$. Conversely, every extremal projects
to a solution of \eqref{RHP}, and two extremals project to the same solution if
and only if they lie on the same $G$-orbit.
\end{theorem}

The three representatives $\hat\mu(t)$, $c(t)$ and $p(t)$ are the point at which
the quotient must be undone: $\mu(t)$ and $\widetilde u(t)$ are classes, and it
is the choice of the reconstructed point $x(t)$ in the fiber, not the reduced
data alone, that selects representatives of them.

\dem
Write $x=x^h g$. By \eqref{tangent_action},
\[
\dot x
=
T\Phi_{g}\bigl(\dot x^h\bigr)+\bigl(TL_{g^{-1}}\dot g\bigr)_P(x^hg),
\]
so that, using the equivariance $A_{xg}\circ T\Phi_g=\operatorname{Ad}_{g^{-1}}
\circ A_x$ and $A(\dot x^h)=0$,
\[
A_{x(t)}\bigl(\dot x(t)\bigr)
=
\operatorname{Ad}_{g^{-1}}A\bigl(\dot x^h\bigr)+TL_{g^{-1}}\dot g
=
TL_{g^{-1}}\dot g .
\]
On the other hand the representative of $\xi(t)$ determined by $x(t)=x^hg$ is
$\operatorname{Ad}_{g^{-1}}\xi^h(t)$, by the equivalence
$(x,\zeta)\sim(xg,\operatorname{Ad}_{g^{-1}}\zeta)$ of \eqref{adjoint_action}.
The condition $A(\dot x)=\operatorname{Ad}_{g^{-1}}\xi^h$ is therefore
$TL_{g^{-1}}\dot g=\operatorname{Ad}_{g^{-1}}\xi^h$, which, since
$\operatorname{Ad}_{g^{-1}}=TL_{g^{-1}}\circ TR_g$, is exactly
\eqref{reconstruction_equation}. Moreover $T\rho(\dot x)=T\rho(\dot x^h)=\dot z$,
so the class of $\dot x$ in $TP/G$ is $(\dot z,\xi)$ under \eqref{tangent_quotient}.
Equation \eqref{reconstruction_equation} is a time-dependent right-invariant
equation on $G$ and therefore has a unique solution on all of $[0,T]$; uniqueness
of $x(\cdot)$ follows since the pair $(\dot z,\xi)$ determines $\dot x$ pointwise.

The costate \eqref{costate_reconstruction} is well defined because $[\beta_J]$ of
\eqref{cotangent_quotient} is a vector bundle isomorphism, so that
$[\beta_J]^{-1}(z,p_z,\mu)$ is a well-defined class in $T^*P/G$; the pairing
\eqref{costate_reconstruction} is exactly the representative of that class lying
over the chosen point $x(t)$, and it is the only one. The same remark applies to
$c(t)$: the class $\widetilde u(t)=[x(t),c(t)]$ has exactly one representative
in the fiber $C_{x(t)}$.

That the resulting curve satisfies the unreduced first-order conditions is
seen in two steps, neither of which uses the Dirac-geometric results of
Section~\ref{drec}. First, the state equation. By
\eqref{reduced_pontryagin_hamiltonian}, the first and fourth equations of
\eqref{RHP} read $\dot z=\widetilde\Gamma(z,u)$ and
$\xi=\widetilde\gamma(z,u)$, so the class of $\dot x$ in
$TP/G\cong T(P/G)\oplus\widetilde{\mathfrak g}$, shown above to be
$(\dot z,\xi)$, coincides with the class of $\Gamma(x,u)$; two vectors of
$T_xP$ with the same class are equal, the action being free, whence
$\dot x=\Gamma(x,u)$. Second, the extremal property. By construction
$\mathbb S(\gamma)=S_G([\gamma])$ for $\gamma=(x,p,u)$, and the constrained
variations \eqref{reduced_vertical_variation}, together with free variations
of $(p_z,\mu,u)$, are exactly the projections of the free variations of curves
in $M$ with fixed endpoints in $P$, every such reduced variation admitting a
lift \cite{CMR01}. Stationarity of $S_G$ at $[\gamma]$ under the constrained
variations is the system \eqref{RHP}, by Section~\ref{pmpreduc}, while
stationarity of $\mathbb S$ at $\gamma$ under free variations is the unreduced
first-order system \eqref{excur}; the correspondence of variations therefore
converts the one into the other. The converse statements are immediate from
the $G$-invariance of the data.
\findem

\begin{remark}
\label{reconstruction_left}
Once Theorem~\ref{red_Dirac_ham} is available, the same correspondence can be
read off the Dirac side, since the reduced equations \eqref{RHP} are equivalent to
the inclusion $([X]_G,dH_G)\in[D_P]_G$, which holds if and only if its representatives
$(X,dH_P)$ lie in $D_P$; the proof given above is independent of
that result. In the left-action convention of Remark~\ref{left_action_variational} the same
computation gives $\dot g=TL_{g}\,\xi^h$ in place of
\eqref{reconstruction_equation}; the examples of Sections~\ref{sec:spacecraft} and~\ref{heisenberg} specialize accordingly. Finally, Theorem~\ref{reconstruction_theorem} makes precise the assertion of
Section~\ref{lp_reduction} that the critical points of $[\mathbb S]_G$ are in
bijection with $G$-equivalence classes of extremals of $\mathbb S$: a reduced solution
determines a whole $G$-orbit of extremals, one for each choice of initial point
in the fiber over $z(0)$, and for the fixed-endpoint problem, whose curve space is not itself $G$-invariant, the
terminal condition $x(T)=x_T$ selects, within the reduced solutions, those whose
reconstruction through $x_0$ closes up.
\end{remark}


\subsection{Independence of the connection}
\label{connection_independence}

The reduced equations \eqref{RHP} were written using a principal connection $A$,
which is not part of the data of the optimal control problem. We now verify that
the choice affects only the presentation.

Let $A$ and $A'$ be two principal connections on $\rho:P\rightarrow B$. Their
difference $\tau=A'-A$ is a $G$-equivariant horizontal $\mathfrak g$-valued
one-form on $P$, hence corresponds to a one-form $\widetilde\tau$ on $B$ with
values in the adjoint bundle, $\widetilde\tau\in\Omega^1(B;\widetilde{\mathfrak
g})$.

\begin{proposition}
\label{connection_change}
Under the change of connection $A\mapsto A'=A+\tau$ the identifications
\eqref{tangent_quotient} and \eqref{cotangent_quotient} change by
\begin{equation}
\label{splitting_change}
\xi'=\xi+\widetilde\tau(\dot z),
\qquad
p_z'=p_z-\bigl\langle\mu,\widetilde\tau\bigr\rangle,
\qquad
\mu'=\mu,
\end{equation}
while the reduced Hamiltonian is unchanged as a function on $M/G$,
\begin{equation}
\label{h_invariance}
h^{A'}\bigl(z,p_z',\mu,u\bigr)=h^{A}\bigl(z,p_z,\mu,u\bigr).
\end{equation}
\end{proposition}

\dem
The momentum map $J$ is defined by evaluation on fundamental vector fields and so
does not involve the connection; hence $\mu'=\mu$. For $\alpha\in T^*_xP$ and
$v\in T_xP$ one has $\langle p_z,T\rho(v)\rangle=\langle\alpha,v\rangle-
\langle\mu,A(v)\rangle$, and subtracting the same identity for $A'$ gives
$p_z'=p_z-\langle\mu,\widetilde\tau\rangle$. Likewise
$\xi'=[x,A'(\dot x)]=\xi+\widetilde\tau(\dot z)$.

For \eqref{h_invariance}, note that $\widetilde\Gamma=T\rho\circ\Gamma$ does not
depend on the connection, whereas
$\widetilde\gamma^{A'}=\widetilde\gamma^{A}+\widetilde\tau(\widetilde\Gamma)$.
Hence, by \eqref{reduced_pontryagin_hamiltonian},
\[
h^{A'}
=
\bigl\langle p_z-\langle\mu,\widetilde\tau\rangle,\widetilde\Gamma\bigr\rangle
+\bigl\langle\mu,\widetilde\gamma^{A}+\widetilde\tau(\widetilde\Gamma)\bigr\rangle-l
=
\bigl\langle p_z,\widetilde\Gamma\bigr\rangle
+\bigl\langle\mu,\widetilde\gamma^{A}\bigr\rangle-l
=
h^{A},
\]
the two $\widetilde\tau$ terms canceling.
\findem

\begin{theorem}
\label{connection_independent}
The reduced dynamics is independent of the choice of principal connection. More
precisely, the map \eqref{splitting_change} is an isomorphism of vector bundles
over $B$ carrying the reduced Dirac structure and the reduced equations
\eqref{RHP} written with respect to $A$ to those written with respect to $A'$,
and it commutes with reconstruction: the unreduced extremal produced by
Theorem~\ref{reconstruction_theorem} is the same for both connections.
\end{theorem}

\dem
The reduced Dirac structure $[D_\Omega]_G$ and the reduced Hamiltonian $H_G$ are
defined without reference to a connection, as is the intrinsic reduced dynamics
$\bigl([X]_G,[dH]\bigr)\in[D_\Omega]_G$ of \eqref{Dirac_reduced_equations}. The
connection enters only through the isomorphism $[\beta_J]$ of
\eqref{cotangent_quotient} used to write that condition in the explicit form
\eqref{reduced_canonical_dirac_explicit}, and \eqref{splitting_change} is exactly
the composite $[\beta_J^{A'}]\circ[\beta_J^{A}]^{-1}$. The two systems of
equations are therefore the same condition expressed in two coordinatizations,
and \eqref{h_invariance} shows that the Hamiltonian is literally the same
function. Reconstruction is likewise connection-free: the class of $\dot x$ in
$TP/G$ is determined by $\dot x$, and $A$ enters
Theorem~\ref{reconstruction_theorem} only through the splitting of that class
into a horizontal lift and a group factor.
\findem


\subsection{Criticality and optimality}
\label{criticality}

The Pontryagin Maximum Principle provides necessary conditions, and everything
above concerns critical points of the functionals rather than minimizers. We make
the relationship explicit, and record what is and is not claimed.

Three standing conventions should be stated plainly. First, we work with the
\emph{weak} form of the maximum principle: the control set is the whole fiber
$C_x$, controls are smooth, and the maximum condition is replaced by the
stationarity condition $\partial H_P/\partial u^a=0$. Second, the normalization
$p_0=-1$ implicit in $H_P=\langle p,\Gamma\rangle-L$ selects the normal
extremals; the abnormal case $p_0=0$ is treated in Section~\ref{singular} through
the cost extension of Proposition~\ref{p0_momentum}. Third, the endpoints
$x_0,x_T\in P$ are prescribed, so that in the reduced principle the vertical
variation satisfies $\eta(0)=\eta(T)=0$.

\begin{proposition}
\label{criticality_optimality}
Let $\gamma$ be an admissible curve of $(P,C,\Gamma,L)$ from $x_0$ to $x_T$ and
let $[\gamma]$ be its reduction. Then
\begin{enumerate}
\item[(i)] $\mathbb S(\gamma)=S_G([\gamma])$, and this holds for every lift of
$[\gamma]$;
\item[(ii)] $\gamma$ is a critical point of $\mathbb S$ if and only if
$[\gamma]$ is a critical point of $S_G$ with respect to the constrained
variations \eqref{reduced_vertical_variation};
\item[(iii)] $\gamma$ is a local minimizer of $\mathbb S$ among admissible curves
from $x_0$ to $x_T$ if and only if $[\gamma]$ is a local minimizer of $S_G$ among
reduced admissible curves from $z_0$ to $z_T$ whose reconstruction through $x_0$,
in the sense of Theorem~\ref{reconstruction_theorem}, terminates at $x_T$;
\item[(iv)] if $\gamma$ is a critical point of $\mathbb S$, then the second
variation of $\mathbb S$ at $\gamma$ descends to the second variation of $S_G$ at
$[\gamma]$; and for the orbit-endpoint problem introduced at the beginning of
Section~\ref{lp_reduction} the tangent
space to the $G$-orbit of $\gamma$ lies in the kernel of the former.
\end{enumerate}
\end{proposition}

\dem
Part (i) is the definition $l(z,u)=L(x,u)$ together with the $G$-invariance of
$L$. Part (ii) is the content of Section~\ref{pmpreduc} combined with
Theorem~\ref{reconstruction_theorem}: the constrained variations
\eqref{reduced_vertical_variation} are exactly the projections of the free
variations of curves in $P$ with fixed endpoints, and $\eta(0)=\eta(T)=0$ encodes
the fixed endpoints. Part (iii) follows from (i) and the bijection of
Theorem~\ref{reconstruction_theorem}, which identifies admissible curves from
$x_0$ to $x_T$ with reduced admissible curves whose reconstruction through $x_0$
terminates at $x_T$; the terminal condition is a genuine restriction, since a reduced curve from $z_0$ to $z_T$ reconstructs to a curve ending anywhere in the fiber over $z_T$, the discrepancy being the group element $g(T)$ of \eqref{reconstruction_equation} (for the Heisenberg problem, the signed area swept by the plane curve, which is the enclosed area when the curve is closed), and this is why the vertical variable $\xi$ is retained in the reduced principle \eqref{red_funct} rather than eliminated. Part (iv) follows by differentiating (i) twice. Two qualifications are needed
for the second assertion of (iv), and both are essential. First, criticality
cannot be dropped: at a noncritical point the second derivative along an orbit
is $\mathrm{Hess}\,\mathbb S(\dot\gamma,\dot\gamma)+d\mathbb S(\ddot\gamma)$, and
invariance annihilates the sum rather than the first term. The
rotation-invariant function $x^2+y^2$ on $\mathbb R^2$ at $(1,0)$, whose Hessian
is $2I$, is the elementary instance. Second, the orbit direction is tangent to
the space of competitors only for the orbit-endpoint problem, since
$\gamma\cdot g$ does not have the endpoints of $\gamma$; for the fixed-endpoint
problem the assertion holds with $G$ replaced by the common stabilizer
$G_{x_0}\cap G_{x_T}$, which is trivial when the action is free.
\findem

Proposition~\ref{criticality_optimality} asserts only that reduction neither creates nor destroys minimizers; no sufficient conditions are claimed here, and the question of minimality for the abnormal extremals of Section~\ref{singular} is taken up in Section~\ref{index_form}; in general the regularity of abnormal minimizers remains open \cite{AgrachevSachkov,Mo02}.


\subsection{Relation to Lie algebroid reduction}
\label{algebroid}

Mart\'{\i}nez \cite{Mr04} formulates optimal control on Lie algebroids and
obtains reduced equations by regarding a $G$-invariant problem on $P$ as a
problem on the Atiyah algebroid $TP/G$. The correspondence between that
formulation and the present one is known: Grabowski and J\'o\'zwikowski
\cite[\S7]{GJ11}, in the wider setting of almost Lie algebroids, reduce an
invariant optimal control problem on a principal bundle to the Atiyah algebroid
and compute the resulting Poisson tensor on $E^*\cong T^*M\times\mathfrak g^*$ in
a connection-adapted splitting, obtaining a canonical part, a curvature term and
a Lie--Poisson part; for $P=G$ this recovers results of Jurdjevic
\cite{Jurbook}. The purpose of this subsection is not to add to those results but
to record the dictionary explicitly for \eqref{RHP}, since the correspondence is
asserted in Section~\ref{intro} and has not so far been exhibited.

Recall that the Atiyah algebroid of $\rho:P\rightarrow B$ is the vector bundle
$E=TP/G\rightarrow B$ of Section~\ref{atiyah_bundle}, whose sections are the
$G$-invariant vector fields on $P$, with the Lie bracket of vector fields and
with anchor $\varrho=T\rho$. Its dual is
\[
E^*=T^*P/G\cong T^*(P/G)\oplus\widetilde{\mathfrak g}^{\,*},
\]
by \eqref{cotangent_quotient}. Consequently the reduced Pontryagin bundle
\eqref{reduced_pontryagin} is
\begin{equation}
\label{algebroid_pontryagin}
M/G\;\cong\;E^*\times_B\widetilde U ,
\end{equation}
the Pontryagin bundle of the Atiyah algebroid, and the reduced Pontryagin
Hamiltonian $h$ is a function on it.

For a Lie algebroid $E\rightarrow B$ with local coordinates $(z^\alpha)$ on $B$
and a local basis of sections $\{e_a\}$ with
\[
\varrho(e_a)=\varrho^\alpha_a\frac{\partial}{\partial z^\alpha},
\qquad
[e_a,e_b]_E=\mathcal C^c_{ab}\,e_c ,
\]
the Hamiltonian equations on $E^*$, with fiber coordinates $y_a$, are
\begin{equation}
\label{algebroid_hamilton}
\dot z^\alpha=\varrho^\alpha_a\frac{\partial h}{\partial y_a},
\qquad
\dot y_a=-\varrho^\alpha_a\frac{\partial h}{\partial z^\alpha}
-\mathcal C^c_{ab}\,y_c\frac{\partial h}{\partial y_b},
\end{equation}
together with $\partial h/\partial u^a=0$ in the controlled case \cite{Mr04}.

\begin{proposition}
\label{atiyah_structure}
Let $A$ be a principal connection and let $\{X^h_\alpha\}$ be the horizontal
lifts of a coordinate frame on $B$, $\{e_i\}$ a local frame of
$\widetilde{\mathfrak g}$. In the adapted basis $\{X^h_\alpha,e_i\}$ of $E$ the
anchor and structure functions are
\[
\varrho(X^h_\alpha)=\frac{\partial}{\partial z^\alpha},
\qquad
\varrho(e_i)=0,
\]
\[
\mathcal C^i_{\alpha\beta}=-F^i_{\alpha\beta},
\qquad
\mathcal C^k_{\alpha i}=-\mathcal C^k_{i\alpha}=\Gamma^k_{\alpha i},
\qquad
\mathcal C^k_{ij}=-c^k_{ij},
\]
where $F^i_{\alpha\beta}$ are the components of the curvature $F_A$,
$\Gamma^k_{\alpha i}$ those of the induced connection on
$\widetilde{\mathfrak g}$ introduced in Section~\ref{rlpp}, and $c^k_{ij}$ the
structure constants of the fiberwise bracket of the adjoint bundle,
$[e_i,e_j]=c^k_{ij}e_k$, together with the values forced by
antisymmetry of the bracket, $\mathcal C^i_{\beta\alpha}=F^i_{\alpha\beta}$ and
$\mathcal C^k_{ji}=c^k_{ij}$. All structure functions not determined by these
vanish.
\end{proposition}

\dem
The anchor statements are immediate. Since $[X^h_\alpha,X^h_\beta]$ projects to
$[\partial_\alpha,\partial_\beta]=0$ it is vertical, and its vertical part is
$A\bigl([X^h_\alpha,X^h_\beta]\bigr)=-F_A(\partial_\alpha,\partial_\beta)$ by
Lemma~\ref{bracket_curvature}. (Readers comparing with \cite{GJ11} should note
the opposite curvature convention there: $F_\nabla(X,Y)=[\nabla X,\nabla Y]_E-
\nabla([X,Y])$ equals $A([X^h,Y^h])$, so that $F_\nabla=-F_A$ with the Cartan
structure equation $F_A=dA+\frac12[A,A]$ used here.) The bracket of a horizontal lift with a vertical
section is the covariant derivative of that section,
$[X^h_\alpha,X_\sigma]=X_{\nabla^A_{\partial_\alpha}\sigma}$, which one checks
in a local trivialization. The bracket of two invariant vertical fields is the
\emph{opposite} of the fiberwise bracket of $\widetilde{\mathfrak g}$: for a
right action the invariant vertical fields restrict on each fiber to
right-invariant vector fields, and right-invariant fields bracket
anti-homomorphically, so $[X_\sigma,X_\tau]=-X_{[\sigma,\tau]}$ and
$\mathcal C^k_{ij}=-c^k_{ij}$. (For a left action the invariant vertical
fields are the left-invariant ones and $\mathcal C^k_{ij}=+c^k_{ij}$,
consistently with the sign change of the coadjoint equation in
Remark~\ref{left_action_variational}.)
\findem

\begin{theorem}
\label{martinez_comparison}
Under the identification \eqref{algebroid_pontryagin}, and writing
$y=(p_z,\mu)$, the algebroid Hamiltonian equations \eqref{algebroid_hamilton} for
the Atiyah algebroid coincide with the reduced Lagrange--Pontryagin equations
\eqref{RHP}. In particular the curvature term of \eqref{RHP} is the structure
function $\mathcal C^i_{\alpha\beta}=-F^i_{\alpha\beta}$, the coadjoint term is
$\mathcal C^k_{ij}=-c^k_{ij}$, and the structure functions
$\mathcal C^k_{\alpha i}=\Gamma^k_{\alpha i}$ are absorbed by the covariant
derivative $D\mu/Dt$ and by the covariant partial derivative
$\partial h/\partial z$ of Section~\ref{rlpp}.
\end{theorem}

\dem
Substituting Proposition~\ref{atiyah_structure} into
\eqref{algebroid_hamilton} and writing $y_\alpha=p_\alpha$, $y_i=\mu_i$,
$\partial h/\partial\mu_i=\xi^i$, the first equation gives
$\dot z^\alpha=\partial h/\partial p_\alpha$. The equation for $p_\alpha$ reads
\[
\dot p_\alpha
=
-\frac{\partial h}{\partial z^\alpha}
-\mathcal C^i_{\alpha\beta}\mu_i\frac{\partial h}{\partial p_\beta}
-\mathcal C^k_{\alpha i}\mu_k\frac{\partial h}{\partial\mu_i}
=
-\frac{\partial h}{\partial z^\alpha}
+F^i_{\alpha\beta}\mu_i\dot z^\beta
-\Gamma^k_{\alpha i}\mu_k\xi^i ,
\]
in which the middle term is $-\langle\mu,F_A(\dot z,\partial_\alpha)\rangle$ and
the last is the connection term by which the naive partial derivative
$\partial h/\partial z^\alpha$, taken at fixed $(p_\alpha,\mu_i)$, differs from
the covariant partial derivative of Section~\ref{rlpp}: since the reduced
equations are independent of the auxiliary connection, we may take $\nabla$
flat in the chart, and then $(Dp_z/Dt)_\alpha=\dot p_\alpha$ while
$(\partial h/\partial z)_\alpha
=\partial h/\partial z^\alpha+\Gamma^k_{\alpha i}\mu_k\xi^i$; the displayed
equation is therefore the second equation of \eqref{RHP}. The equation for
$\mu_i$ reads, using $\mathcal C^k_{i\alpha}=-\Gamma^k_{\alpha i}$,
\[
\dot\mu_i
=
-\mathcal C^k_{ij}\mu_k\frac{\partial h}{\partial\mu_j}
-\mathcal C^k_{i\alpha}\mu_k\frac{\partial h}{\partial p_\alpha}
=
c^k_{ij}\mu_k\xi^j+\Gamma^k_{\alpha i}\mu_k\dot z^\alpha ,
\]
which is $D\mu/Dt=-\operatorname{ad}^*_\xi\mu$, since
$\bigl(\operatorname{ad}^*_\xi\mu\bigr)_i
=\langle\mu,[\xi,e_i]\rangle=-c^k_{ij}\mu_k\xi^j$. The stationarity condition
$\partial h/\partial u^a=0$ is common to both.
\findem

\begin{remark}
The reduction of the \emph{variations}, rather than only of the equations, has
also been studied in the algebroid setting; see Grabowska and Grabowski
\cite{GG08} for constrained variational calculus on general algebroids and the
attendant distinction between the vakonomic and nonholonomic principles.
\end{remark}

\begin{remark}
The comparison also explains the status of the two constructions. Mart\'inez's
formulation is intrinsic, the algebroid $TP/G$ requiring no choice of connection;
the price is that the equations are written in terms of structure functions
relative to a chosen local basis. The formulation of Section~\ref{pmpreduc} fixes
a connection once and for all, which converts the structure functions into the
geometric objects $F_A$ and $\operatorname{ad}^*$ and makes the vertical variable
$\xi$ available for reconstruction. Proposition~\ref{connection_change} is the
statement that changing the connection is a change of adapted basis, which is why
the resulting dynamics is the same.
\end{remark}

\begin{remark}
The Dirac formulation of Section~\ref{drec} does not have an evident algebroid
counterpart. The reduced object $[D_\Omega]_G$ is a Dirac subbundle of the
Courant algebroid $\bigl(TM\oplus T^*M\bigr)/G$, and the natural question of
whether it can be described as a Dirac structure on a suitable prolongation of
$TP/G$ appears to be open; we do not pursue it here.
\end{remark}

\section{Reduction of the Dirac structure of optimal control problems with symmetry}\label{drec}

\subsection{Lie--Dirac reduction of the Dirac structure}

Let $(P,C,\Gamma,L)$ be a $G$-invariant optimal control problem, and let
$D_\Omega\subset TM\oplus T^*M$ be the Dirac structure on the Pontryagin bundle
$M=T^*P\times_PC$ associated with the canonical presymplectic form $\Omega$.

Since the action of $G$ on $M$ lifts naturally to the tangent and cotangent bundles, it also induces an action on the Whitney sum $TM\oplus T^*M$. It is therefore natural to ask whether the Dirac structure descends to the quotient. The following result shows that this is indeed the case.

\begin{proposition}
\label{Dinvariant}
Let $(P,C,\Gamma,L)$ be a $G$-invariant optimal control problem. Then the canonical Dirac structure
$D_\Omega\subset TM\oplus T^*M$
is $G$-invariant.
\end{proposition}

\begin{proof}
If $\alpha=i_X\Omega$ then $g^{*-1}\alpha=i_{g_*X}\Omega$, since $g^*\Omega=\Omega$ by Proposition~\ref{invariance_prop}; hence $(g_*X,g^{*-1}\alpha)\in D_\Omega$.
\end{proof}

Since the action of $G$ on $TM\oplus T^*M$ is fiber preserving, the quotient
$(TM\oplus T^*M)/G$, which is canonically the Whitney sum of the quotient bundles $TM/G$ and $T^*M/G$,
is a vector bundle over the reduced Pontryagin bundle $M/G$, and the $G$-invariance of $D_\Omega$ implies that it descends to a vector subbundle of the quotient bundle.

\begin{definition}
The vector subbundle
$[D_\Omega]_G\subset TM/G\oplus T^*M/G$
is called the \emph{Lie--Dirac reduction} of the Dirac structure $D_\Omega$.
\end{definition}

From now on we assume, as holds in all the
examples of this paper, that the standard fiber $U$ of the control bundle is a
vector space on which $G$ acts linearly, so that $\widetilde U$ is a vector
bundle and $A$ induces linear connections on $\widetilde{\mathfrak g}$,
$\widetilde{\mathfrak g}^{\,*}$ and $\widetilde U$; in particular, $A$ identifies $TM/G$ with $T(M/G)\oplus\widetilde{\mathfrak g}$ and $T^*M/G$ with $T^*(M/G)\oplus\widetilde{\mathfrak g}^{\,*}$, as in Section~\ref{atiyah_bundle}.

Notice that, in general, $[D_\Omega]_G$ is not a Dirac structure on the reduced Pontryagin bundle $M/G$, since it is not a subbundle of
\[
T(M/G)\oplus T^*(M/G).
\]
Indeed, using the canonical identifications
$TM/G\cong T(M/G)\oplus\widetilde{\mathfrak g}$ and
$T^*M/G\cong T^*(M/G)\oplus\widetilde{\mathfrak g}^{\,*}$,
we obtain
\[
[D_\Omega]_G
\subset
T(M/G)\oplus T^*(M/G)
\oplus
\widetilde{\mathfrak g}
\oplus
\widetilde{\mathfrak g}^{\,*}.
\]

It is, however, a Dirac structure in a precise sense. The bundle
\[
E=\bigl(TM\oplus T^*M\bigr)/G\longrightarrow M/G
\]
carries the quotients of the symmetric pairing of Definition~\ref{Dirac_vector_bundle} and of the Courant bracket of Definition~\ref{integrability}, both of which are preserved by the tangent and cotangent lifts of the $G$-action; hence $E$ is a Courant algebroid over $M/G$ in the general sense of Liu, Weinstein and Xu \cite{LWX97}, although it is not of the form $T(M/G)\oplus T^*(M/G)$. Accordingly, if $D\subset TM\oplus T^*M$ is a $G$-invariant Dirac structure, then $D/G$ is a Dirac subbundle of $E$, and $D/G$ is integrable whenever $D$ is, since $G$ preserves the Courant bracket (cf. \cite[p.~400]{Yo07a}). In particular $\Omega=\mathrm{pr}_P^{\,*}\omega_P$ is closed, so that $D_\Omega$ is integrable and hence so is $[D_\Omega]_G$.

We shall keep the notation $[D_\Omega]_G$ and refer to it as the \emph{reduced Dirac structure}, understanding it as a Dirac subbundle of the Courant algebroid $E$; the pullback and pushforward introduced below are then the corresponding morphisms of Courant algebroids.

 \subsection{Reduced Dirac equations of an optimal control problem with symmetry}
\label{REOCP}

Once the reduced Dirac structure $[D_\Omega]_G$ has been constructed, the corresponding reduced dynamics is simply determined by the inclusion
\begin{equation}
\label{Dirac_reduced_equations}
([X],[dH])\in[D_\Omega]_G.
\end{equation}

Here $[X]$ denotes a section of the reduced tangent bundle
$TM/G$, while $[dH]$ is the corresponding section of $T^*M/G$ determined by the differential of the reduced Pontryagin Hamiltonian
$H_G:M/G\rightarrow\mathbb R$, induced by the $G$-invariant Hamiltonian $H_P$ (cf. \eqref{Hinvar}). Under the identifications of Section~\ref{atiyah_bundle}, the reduced vector field takes values in $T(M/G)\oplus\widetilde{\mathfrak g}$ and the reduced covector field in $T^*(M/G)\oplus\widetilde{\mathfrak g}^{\,*}$.

Although equation \eqref{Dirac_reduced_equations} provides the intrinsic reduced dynamics, its explicit computation is rather involved because of the structure of the reduced bundle. A considerable simplification follows from observing that the presymplectic form on the Pontryagin bundle is the pullback of the canonical symplectic form on $T^*P$, namely
$\Omega=\mathrm{pr}_P^{\,*}\omega_P$, and does not involve the control directions. Consequently, after projecting these out, the Lie--Dirac reduction of $D_\Omega$ is directly related to the Lie--Dirac reduction of the canonical Dirac structure on $T^*P$ associated with $\omega_P$.

To make this relation precise we recall the notions of pullback and pushforward of Dirac structures \cite{Yo09}, which for the quotients considered here are the corresponding morphisms of Courant algebroids.

\begin{definition}[Pushforward and pullback of Dirac structures]
\label{push_pull}

Let $\psi:N\rightarrow N'$ be a smooth map, and let $\mathfrak{Dir}(N)$ and $\mathfrak{Dir}(N')$ denote the sets of Dirac structures on $N$ and $N'$, respectively.

The \emph{forward map}
$\mathcal F\psi:\mathfrak{Dir}(N)\rightarrow\mathfrak{Dir}(N')$
is defined by
\[
\mathcal F\psi(D)
=
\{(\psi_*u,\alpha)\mid
(u,\psi^*\alpha)\in D\},
\]
where $u\in TN$ and $\alpha\in T^*N'$.

Similarly, the \emph{backward map}
$\mathcal B\psi:\mathfrak{Dir}(N')\rightarrow\mathfrak{Dir}(N)$
is defined by
\[
\mathcal B\psi(D')
=
\{(v,\psi^*\beta)\mid
(\psi_*v,\beta)\in D'\},
\]
where $v\in TN$ and $\beta\in T^*N'$.

The pushforward and pullback of Dirac structures are then defined by
\[
\psi_*D=\mathcal F\psi(D),
\qquad
\psi^*D'=\mathcal B\psi(D').
\]
\end{definition}

Both formulas are, in the first place, pointwise prescriptions, and for an
arbitrary smooth $\psi$ the subspaces they produce are maximally isotropic at
each point but need not have constant rank, hence need not be smooth
subbundles; a forward image must in addition be constant along the fibers of
$\psi$ in order to be well defined on $N'$. We therefore read
Definition~\ref{push_pull} as producing \emph{pointwise Dirac relations}, and
call $\psi_*D$ or $\psi^*D'$ a Dirac structure only under the usual
clean-intersection and constant-rank hypotheses (see \cite{Bl04} and
Remark~\ref{restriction_comparison} below). All the maps to which the
definition is applied in this paper are surjective submersions and all the
structures involved are graphs of pullbacks of two-forms, for which those
hypotheses are automatic: the backward image of the graph of $\widetilde\Omega$
along $\psi$ is the graph of $\psi^*\widetilde\Omega$, of constant rank
$\dim N$, and the forward image of the graph of $\psi^*\widetilde\Omega$ is
again the graph of $\widetilde\Omega$.

The Dirac structures $D_\Omega$ and $D_P$ are related by pullback and pushforward through the canonical projection $\mathrm{pr}_P:M\rightarrow T^*P$.

\begin{lemma}
\label{pullback_Dirac}
Let $D_P$ denote the canonical Dirac structure on $T^*P$ associated with the symplectic form $\omega_P$. Then
\[
\mathrm{pr}_P^{\,*}D_P=D_\Omega,
\qquad
\mathrm{pr}_{P\,*}D_\Omega=D_P.
\]
\end{lemma}
\begin{proof}
Since $\Omega=\mathrm{pr}_P^{\,*}\omega_P$ and $D_P$ is the graph of $\omega_P$, the condition $i_X\Omega=\mathrm{pr}_P^{\,*}\alpha$ is equivalent to $i_{\mathrm{pr}_{P*}X}\omega_P=\alpha$, which by Definition~\ref{push_pull} is both equalities at once.
\end{proof}

The previous result is a particular instance of a more general situation. Let $(\widetilde M,\widetilde\Omega)$ be a symplectic manifold and let
$p:M\rightarrow\widetilde M$
be a surjective submersion. If
$\Omega=p^*\widetilde\Omega$,
then exactly the same argument shows that $D_\Omega=p^*D_{\widetilde\Omega}$: the presymplectic Dirac structure $D_\Omega$ projects onto the symplectic Dirac structure $D_{\widetilde\Omega}$.

Assume now that a Lie group $G$ acts on both $M$ and $\widetilde M$ in such a way that the projection $p$ is equivariant,
\begin{equation}
\label{pequiv}
p(x\cdot g)=p(x)\cdot g,
\qquad
x\in M,\;g\in G,
\end{equation}
and that the action preserves the symplectic form,
$g^*\widetilde\Omega=\widetilde\Omega$.
Then both $D_{\widetilde\Omega}$ and $D_\Omega$ are $G$-invariant.

Consequently, the Lie--Dirac reduction of $D_{\widetilde\Omega}$ defines a Dirac subbundle of the corresponding Courant algebroid over the reduced space,
\[
[D_{\widetilde\Omega}]_G
\subset
T(\widetilde M/G)\oplus
T^*(\widetilde M/G)
\oplus
\widetilde{\mathfrak g}
\oplus
\widetilde{\mathfrak g}^{\,*}.
\]

Moreover, the equivariance condition \eqref{pequiv} induces a smooth map
$p_G:M/G\rightarrow\widetilde M/G$
satisfying
\begin{equation}
\label{equiv_quotient}
p_G([x])=[p(x)].
\end{equation}

Before stating the corresponding statement after reduction, we must say what
``pullback'' means there, since $[D_\Omega]_G$ and $[D_{\widetilde\Omega}]_G$ are
not Dirac structures on manifolds but subbundles of the quotient Courant
algebroids
\[
\mathcal E_{M}=(TM/G)\oplus(T^*M/G),
\qquad
\mathcal E_{\widetilde M}=(T\widetilde M/G)\oplus(T^*\widetilde M/G),
\]
over $M/G$ and $\widetilde M/G$ respectively. A class $[\alpha]\in T^*M/G$ is not
an ordinary covector on $M/G$ --- under a principal connection it is a pair
consisting of a covector on $M/G$ and a section of the coadjoint bundle --- so
$[\,p^*\beta\,]$ is not an ordinary cotangent pullback and no map
$M/G\to\widetilde M/G$ can by itself induce it. What does induce it is the
following pair of bundle maps, both consequences of the equivariance of $p$.

Since $Tp:TM\to T\widetilde M$ is $G$-equivariant, it descends to a vector bundle
morphism
\begin{equation}
\label{quotient_tangent_map}
Tp/G:\ TM/G\longrightarrow T\widetilde M/G,
\qquad
[v]\longmapsto[\,Tp(v)\,],
\end{equation}
covering $p_G$. Dually, if $\beta$ is a covector at $p(x)$ then $p^*\beta$ is a
covector at $x$, and $p^*(\beta\cdot g)=(p^*\beta)\cdot g$, so the assignment
\begin{equation}
\label{quotient_cotangent_map}
(Tp)^*/G:\ p_G^*\bigl(T^*\widetilde M/G\bigr)\longrightarrow T^*M/G,
\qquad
[\beta]\longmapsto[\,p^*\beta\,],
\end{equation}
is a well-defined vector bundle morphism over $M/G$: the representative
$\beta$ is taken at $p(x)$ for the same representative $x$ used on the left, and
changing $x$ to $xg$ changes both sides by the action of $g$. In the splitting
determined by a principal connection, \eqref{quotient_tangent_map} is
$(Tp_G,\mathrm{id}_{\widetilde{\mathfrak g}})$ and
\eqref{quotient_cotangent_map} is
$(p_G^*,\mathrm{id}_{\widetilde{\mathfrak g}^{\,*}})$, the adjoint and coadjoint
components being the identity because $p$ is equivariant and therefore matches
infinitesimal generators: $Tp(\zeta_M)=\zeta_{\widetilde M}$ for every
$\zeta\in\mathfrak g$. This is the precise content of the informal statement that
$p_G$ ``acts as the identity on the adjoint factor''.

The pair \eqref{quotient_tangent_map}--\eqref{quotient_cotangent_map} defines a
Courant relation
\begin{equation}
\label{quotient_courant_relation}
\mathcal R_p
=
\Bigl\{
\bigl(([v],[\alpha]),([w],[\beta])\bigr)
\ :\
[w]=(Tp/G)[v],
\quad
[\alpha]=\bigl((Tp)^*/G\bigr)[\beta]
\Bigr\}
\subset
\mathcal E_M\times_{M/G}p_G^*\,\mathcal E_{\widetilde M},
\end{equation}
and we write $\mathcal B\mathcal R_p$ for the associated backward image,
\[
\mathcal B\mathcal R_p(\mathcal D')
=
\bigl\{
([v],[\alpha])\in\mathcal E_M
\ :\
\exists\,[\beta],\
\bigl((Tp/G)[v],[\beta]\bigr)\in\mathcal D'
\text{ and }
[\alpha]=\bigl((Tp)^*/G\bigr)[\beta]
\bigr\}.
\]
This is the reduced counterpart of the backward map $\mathcal B\psi$ of
Definition~\ref{push_pull}, and it reduces to it when $G$ is trivial.

\begin{proposition}
\label{pull_quotient_canonical}
With the notation introduced above, the Lie--Dirac reduction of the presymplectic
Dirac structure $D_\Omega$ is the backward image, along the Courant relation
\eqref{quotient_courant_relation}, of the Lie--Dirac reduction of the symplectic
Dirac structure $D_{\widetilde\Omega}$,
\[
[D_\Omega]_G
=
\mathcal B\mathcal R_p\bigl([D_{\widetilde\Omega}]_G\bigr)
=:
\widetilde p_G^{\,*}[D_{\widetilde\Omega}]_G,
\]
where $\widetilde p_G^{\,*}$ denotes $\mathcal B\mathcal R_p$ and not a cotangent pullback. In particular, for the Pontryagin bundle $M=T^*P\times_PC$ of a $G$-invariant optimal control problem and $p=\mathrm{pr}_P$,
\[
[D_\Omega]_G
=
\mathcal B\mathcal R_{\mathrm{pr}_P}\bigl([D_P]_G\bigr)
=:
\mathrm{pr}_G^{\,*}[D_P]_G,
\]
where $\mathrm{pr}_G:M/G\longrightarrow T^*P/G$ is the map induced by the canonical projection $\mathrm{pr}_P$.
\end{proposition}

\begin{proof}
By Lemma~\ref{pullback_Dirac}, extended to $p$ as above, $(X,\alpha)\in D_\Omega$ if and only if $\alpha=p^*\beta$ for some $\beta$ with $(p_*X,\beta)\in D_{\widetilde\Omega}$. Passing to $G$-orbits, which is legitimate since $D_\Omega$ and $D_{\widetilde\Omega}$ are $G$-invariant and the action is free, and using $[p_*X]=(Tp/G)[X]$ and $[p^*\beta]=\bigl((Tp)^*/G\bigr)[\beta]$, which are the definitions \eqref{quotient_tangent_map} and \eqref{quotient_cotangent_map}, this becomes: $([X],[\alpha])\in[D_\Omega]_G$ if and only if $[\alpha]=\bigl((Tp)^*/G\bigr)[\beta]$ for some $[\beta]$ with $\bigl((Tp/G)[X],[\beta]\bigr)\in[D_{\widetilde\Omega}]_G$, which is the membership condition for $\mathcal B\mathcal R_p([D_{\widetilde\Omega}]_G)$. The two inclusions are this equivalence read in its two directions.
\end{proof}

\subsection{The reduced canonical Dirac structure}
\label{reduced_canonical}

We now derive the connection-dependent description of $[D_P]_G$ that will be
used in the sequel. Throughout, $A$ is a fixed principal connection on
$\rho:P\rightarrow B=P/G$, with curvature $F_A=dA+\tfrac12[A,A]$, and we use the
identification $T^*P/G\cong T^*(P/G)\oplus\widetilde{\mathfrak g}^{\,*}$ of
Section~\ref{weinstein_space}.

At a point $(z,p_z,\mu)$ we write an element of the reduced tangent bundle as
$v=(v_z,v_{p_z};\xi,v_\mu)$, where
\[
(v_z,v_{p_z})\in T_{(z,p_z)}T^*(P/G),
\qquad
\xi\in\widetilde{\mathfrak g}_z,
\qquad
v_\mu\in\widetilde{\mathfrak g}^{\,*}_z,
\]
and an element of the corresponding dual bundle as
$\alpha=(\pi_z,\pi_{p_z};\alpha_\xi,\alpha_{v_\mu})$. The pair notation
$(v_z,v_{p_z})$ for a tangent vector to $T^*(P/G)$ refers to the splitting
$T\,T^*(P/G)\cong TB\oplus T^*B$ determined by the auxiliary torsion-free
connection $\nabla$ fixed in Section~\ref{rlpp}; equivalently, all formulas
below may be read in cotangent-natural local coordinates. Because $\nabla$ is
torsion-free, the canonical symplectic form of $T^*(P/G)$ takes in this
splitting the familiar pair form appearing as the first two terms of
\eqref{Omega_red}. The duality is
\begin{equation}
\label{reduced_duality}
\langle\alpha,v\rangle
=
\langle\pi_z,v_z\rangle
+
\langle\pi_{p_z},v_{p_z}\rangle
+
\langle\alpha_\xi,\xi\rangle
+
\langle\alpha_{v_\mu},v_\mu\rangle.
\end{equation}

\begin{lemma}
\label{reduced_two_form}
Let $\Theta$ denote the canonical one-form on $T^*P$, so that
$\omega_P=-d\Theta$. Under the identification \eqref{cotangent_quotient}, the form
$\omega_P$ descends to the fiberwise skew-symmetric bilinear form
\begin{equation}
\label{Omega_red}
\begin{aligned}
\Omega_{\mathrm{red}}(v,w)
={}&
\langle w_{p_z},v_z\rangle
-
\langle v_{p_z},w_z\rangle
\\
&-
\langle v_\mu,\zeta\rangle
+
\langle w_\mu,\xi\rangle
-
\langle\mu,[\xi,\zeta]\rangle
-
\langle\mu,F_A(v_z,w_z)\rangle,
\end{aligned}
\end{equation}
where $w=(w_z,w_{p_z};\zeta,w_\mu)$.
\end{lemma}

\dem
For $\alpha_x\in T^*_xP$ and $v\in T_xP$ the connection splits
$v=\mathrm{hor}\,v+\mathrm{ver}\,v$, with $\mathrm{ver}\,v=(A(v))_P(x)$, so
that
\[
\langle\alpha_x,v\rangle
=
\langle p_z,T\rho(v)\rangle
+
\langle J(x,\alpha_x),A(v)\rangle,
\]
where $p_z\in T^*_zB$ is the horizontal part of $\alpha_x$ and $J$ is the
momentum map appearing in \eqref{cotangent_quotient}. Consequently
\[
\Theta=\Theta_B+\langle J,A\rangle
\]
in the obvious sense, the two terms being pulled back to $T^*P$ along the maps
of \eqref{cotangent_quotient}. Taking $-d$ and using the structure equation,
\[
-d\langle J,A\rangle
=
-\langle dJ\wedge A\rangle
-
\langle J,F_A\rangle
+
\tfrac12\langle J,[A,A]\rangle,
\]
Passing to the quotient, the second term produces the curvature term
$-\langle\mu,F_A(v_z,w_z)\rangle$, and the third produces
$+\langle\mu,[\xi,\zeta]\rangle$, since $[A,A](v,w)=2[\xi,\zeta]$. The first
term requires care: the fiber components $v_\mu,w_\mu$ of the splitting are
the \emph{covariant} ones, and along a curve through the point the naive
derivative of the momentum in the trivialization differs from them by the
coadjoint correction of the transport, $dJ(v)=v_\mu+\operatorname{ad}^*_\xi\mu$.
Hence
\[
-\langle dJ\wedge A\rangle(v,w)
=
-\langle v_\mu,\zeta\rangle+\langle w_\mu,\xi\rangle
-2\langle\mu,[\xi,\zeta]\rangle ,
\]
and the bracket contributions of the first and third terms combine into the
single term $-\langle\mu,[\xi,\zeta]\rangle$ of \eqref{Omega_red}, while
$-d\Theta_B$ gives the canonical part
$\langle w_{p_z},v_z\rangle-\langle v_{p_z},w_z\rangle$. Skew-symmetry of
\eqref{Omega_red} is immediate, the last two terms being skew by the
antisymmetry of $[\cdot,\cdot]$ and of $F_A$. For the detailed verification we
refer to \cite{Yo09}.
\findem

\begin{proposition}
\label{red_Dirac_structure}
The Lie--Dirac reduction $[D_P]_G$ of the canonical Dirac structure on $T^*P$
is the graph of \eqref{Omega_red}, that is,
\[
[D_P]_G
=
\bigl\{(v,\alpha)\ \big|\ \alpha=i_v\Omega_{\mathrm{red}}\bigr\},
\]
and $(v,\alpha)\in[D_P]_G$ if and only if
\begin{equation}
\label{reduced_canonical_dirac_explicit}
\left\{
\begin{aligned}
v_z
&=
\pi_{p_z},
\\[0.4em]
v_{p_z}+\pi_z
&=
-\left\langle
\mu,F_A(v_z,\cdot)
\right\rangle,
\\[0.4em]
\xi
&=
\alpha_{v_\mu},
\\[0.4em]
v_\mu+\alpha_\xi
&=
-\operatorname{ad}_\xi^*\mu.
\end{aligned}
\right.
\end{equation}
Moreover $[D_P]_G$ is
maximally isotropic with respect to the pairing induced by
\eqref{reduced_duality}.
\end{proposition}

\dem
Write $\langle\alpha,w\rangle=\Omega_{\mathrm{red}}(v,w)$ and evaluate on the
four families $w=(w_z,0;0,0)$, $(0,w_{p_z};0,0)$, $(0,0;\zeta,0)$ and
$(0,0;0,w_\mu)$ in turn. The second gives $\pi_{p_z}=v_z$. The first gives
$\langle\pi_z,w_z\rangle=-\langle v_{p_z},w_z\rangle-\langle\mu,F_A(v_z,w_z)\rangle$,
that is, $v_{p_z}+\pi_z=-\langle\mu,F_A(v_z,\cdot)\rangle$. The fourth gives
$\alpha_{v_\mu}=\xi$. Finally the third gives
\[
\langle\alpha_\xi,\zeta\rangle
=
-\langle v_\mu,\zeta\rangle
-
\langle\mu,\operatorname{ad}_\xi\zeta\rangle
=
-\bigl\langle v_\mu+\operatorname{ad}_\xi^*\mu,\zeta\bigr\rangle
\qquad
\text{for all }\zeta\in\widetilde{\mathfrak g}_z,
\]
that is, $v_\mu+\alpha_\xi=-\operatorname{ad}_\xi^*\mu$. These are precisely the
four relations \eqref{reduced_canonical_dirac_explicit}.

As for the last assertion, let $(v,\alpha)$ and $(w,\beta)$ belong to
$[D_P]_G$. Then
\[
\langle\!\langle(v,\alpha),(w,\beta)\rangle\!\rangle
=
\langle\alpha,w\rangle+\langle\beta,v\rangle
=
\Omega_{\mathrm{red}}(v,w)+\Omega_{\mathrm{red}}(w,v)
=
0
\]
by skew-symmetry, so that $[D_P]_G$ is isotropic; and it is maximally so,
being the graph of a bundle map and therefore of fiber rank equal to that of
the reduced tangent bundle.
\findem

\begin{remark}
Both of the signs that carry the content of the reduction, namely that of
$\operatorname{ad}_\xi^*\mu$ and that of the curvature coupling, originate in
the single term $-d\langle J,A\rangle$ of Lemma~\ref{reduced_two_form}, through
the structure equation $F_A=dA+\tfrac12[A,A]$. They are not determined by
isotropy alone: replacing $-\operatorname{ad}_\xi^*\mu$ by
$+\operatorname{ad}_\xi^*\mu$, or reversing the sign of the curvature term,
again yields a maximally isotropic subbundle of the same rank. This is why
\eqref{Omega_red} is derived here, rather than the relations
\eqref{reduced_canonical_dirac_explicit} being posited.
\end{remark}

\begin{remark}[the left-action convention]
\label{left_action_dirac}
With the conventions of Remark~\ref{left_action_variational}, that is, for a
left action of $G$ on $P$, the same computation gives the relations \eqref{reduced_canonical_dirac_explicit} with the last one replaced by
\begin{equation}
\label{reduced_canonical_dirac_left}
v_\mu+\alpha_\xi
=
+\operatorname{ad}_\xi^*\mu,
\end{equation}
which for $P=G$ is the relation of Yoshimura and Marsden
\cite[Prop.~4.1]{Yo07a}; substituting the reduced vector field and the reduced
differential then returns \eqref{RHP} with the coadjoint equation \eqref{RHP_left},
so that Theorem~\ref{red_Dirac_ham} below holds verbatim in this convention.
\end{remark}


\subsection{The Lie--Dirac reduction and the reduced Pontryagin Maximum Principle}
\label{ldrpmp}

In Section~\ref{pmpreduc} we derived the reduced Pontryagin Maximum Principle by reducing the Lagrange--Pontryagin variational principle. We now show that exactly the same equations are obtained from the Lie--Dirac reduction of the implicit Dirac formulation of the optimal control problem. By Proposition~\ref{pull_quotient_canonical}, the reduced presymplectic Dirac structure $[D_\Omega]_G$ is the pullback of the reduced canonical Dirac structure $[D_P]_G$, so that it is sufficient to describe the dynamics using $[D_P]_G$.

Let $H_G:M/G\longrightarrow\mathbb R$ be the reduced Pontryagin Hamiltonian. Since the optimal control problem is regular, the primary constraint manifold
$M_1=\{\partial H_P/\partial u^a=0\}$
is $G$-invariant and therefore descends to the quotient $[M_1]_G=M_1/G$, given by $\partial H_G/\partial u^a=0$.
Accordingly, the reduced implicit Dirac system is defined by
\begin{equation}
\label{ecimplicpontry}
\left(
[X]_G,
dh
\right)
\in
[D_P]_G,
\end{equation}
where both entries are taken along the map
\[
p_{1,G}:[M_1]_G\longrightarrow T^*P/G
\]
induced by $\mathrm{pr}_P$, and $h$ is the reduced Hamiltonian on $T^*P/G$
obtained by eliminating the stationary control. Some care is needed with the
types here. The differential $dH_G|_{T[M_1]_G}$ is a covector on $[M_1]_G$,
whereas the covector component of $[D_P]_G$ lives over $T^*P/G$; the two are
matched by $p_{1,G}$, which by Proposition~\ref{regular_points} is a local
diffeomorphism at every regular point of $M_1$; since the action is free and
proper and $\mathrm{pr}_P$ is equivariant, that property descends to
$p_{1,G}$. Hence
\[
h=\bigl(p_{1,G}^{\,*}\bigr)^{-1}\bigl(H_G|_{[M_1]_G}\bigr)
\]
is defined, and the vector component $[X]_G$ is likewise a vector field along
$p_{1,G}$ rather than on $[M_1]_G$. Since regularity gives only a local
diffeomorphism, $h$ is defined locally on each branch of the stationary control,
that is, of the solution set of $\partial H_G/\partial u^a=0$; it is defined
globally on $T^*P/G$ only under the
additional hypothesis that the stationary control is globally unique. All the
statements below are to be read on a fixed such branch.

We now show that this implicit formulation is completely equivalent to the reduced Pontryagin Maximum Principle obtained from the reduced variational principle.
\begin{theorem}
\label{red_Dirac_ham}
Let $(P,C,\Gamma,L)$ be a regular $G$-invariant optimal control problem. The equations of the reduced Pontryagin Maximum Principle obtained from the reduced Lagrange--Pontryagin variational principle coincide with the implicit equations determined by the Lie--Dirac reduction
\eqref{ecimplicpontry}.

\end{theorem}

\begin{proof}
Let $(z(t),p_z(t),\mu(t),u(t))$ be a curve in $[M_1]_G$, on a fixed branch of the stationary control. In the splitting of Section~\ref{reduced_canonical}, determined by $\nabla$ and by the connection induced by $A$, the velocity of the reduced curve and the differential of $h$ are represented by
\begin{equation}
\label{reduced_dirac_vector_field}
[X]_G
=
\left(
\dot z,
\frac{Dp_z}{Dt};
\xi,
\frac{D\mu}{Dt}
\right),
\qquad
\left(
\pi_z,\pi_{p_z};
\alpha_\xi,\alpha_{v_\mu}
\right)
=
\left(
\frac{\partial h}{\partial z},
\frac{\partial h}{\partial p_z};
0,
\frac{\partial h}{\partial\mu}
\right),
\end{equation}
where $Dp_z/Dt$, $D\mu/Dt$ and the covariant partial derivative $\partial h/\partial z$ are those fixed in Section~\ref{rlpp}, and the derivative with respect to the coadjoint variable $\mu$ is paired with the variation $v_\mu$ and therefore corresponds to $\alpha_{v_\mu}$, so that $\alpha_\xi=0$. Substituting \eqref{reduced_dirac_vector_field} into the four relations \eqref{reduced_canonical_dirac_explicit} gives, in turn, $\dot z=\partial h/\partial p_z$, $Dp_z/Dt=-\partial h/\partial z-\langle\mu,F_A(\dot z,\cdot)\rangle$, $\xi=\partial h/\partial\mu$ and, since $\alpha_\xi=0$, $D\mu/Dt=-\operatorname{ad}_\xi^*\mu$. Together with the primary constraint $\partial h/\partial u=0$ these are exactly the reduced Lagrange--Pontryagin equations \eqref{RHP}, the curvature term being written in the equivalent form noted after \eqref{RHP}.
\end{proof}

\section{Application: Optimal control of a rigid spacecraft with internal rotors}
\label{sec:spacecraft}

In this section, we illustrate the reduction procedure developed above by means
of the optimal control of a rigid spacecraft equipped with three internal
reaction wheels. The mechanical system is adapted from the classical model
considered in \cite{Bloch1992}. It provides a particularly convenient example
for our purposes, since the configuration manifold carries a natural
$SO(3)$-symmetry, while the torques applied to the internal rotors provide
natural control variables.

We first describe the mechanical system and its Hamiltonian formulation.
We then regard the resulting Hamiltonian phase space as the state manifold of
an optimal control problem. Thus, the configuration manifold of the mechanical
system will be denoted by $Q$, whereas the state manifold entering the optimal
control problem will be $P=T^*Q$. The reduction procedure developed in the
previous sections is therefore applied to the induced $SO(3)$-action on $P$.


\subsection{Mechanical model}

The configuration manifold of the spacecraft--rotor system is
\[
Q=SO(3)\times (S^1)^3,
\]
where $R\in SO(3)$ describes the attitude of the spacecraft and
$\theta=(\theta_1,\theta_2,\theta_3)\in(S^1)^3$ denotes the angular positions
of the three reaction wheels. The body angular velocity is
$\widehat\Omega=R^{-1}\dot R\in\mathfrak{so}(3)$, written simply as
$\Omega\in\mathbb R^3$ under the standard identification
$\mathfrak{so}(3)\simeq\mathbb R^3$, and the angular velocity of the rotors
relative to the spacecraft is $\Omega_r=\dot\theta\in\mathbb R^3$.

Let $I\in\mathbb R^{3\times3}$ be the inertia tensor of the spacecraft
excluding the axial inertia of the reaction wheels, and let
$J=\operatorname{diag}(J_1,J_2,J_3)$ be the inertia tensor of the three rotors
about their rotation axes. The mechanical Lagrangian is taken to be
\begin{equation}
\label{spacecraft_lagrangian}
\mathcal L(R,\theta,\Omega,\Omega_r)
=
\frac12\Omega^T I\Omega
+
\frac12(\Omega+\Omega_r)^T
J(\Omega+\Omega_r).
\end{equation}
In particular, $\mathcal L$ is independent of $R$ and $\theta$ when expressed
in body variables.

The natural symmetry of the rigid-body system is the left multiplication of
$SO(3)$ on the attitude variable, that is, the \emph{left} action
\begin{equation}
\label{spacecraft_action}
A\cdot(R,\theta)
=
(AR,\theta),
\qquad
A\in SO(3).
\end{equation}
This is a smooth, free and proper left action. The general reduction theory of
Sections~\ref{pmpreduc} and~\ref{drec} was formulated for right actions; the
corresponding statements for a left action are those of
Remarks~\ref{left_action_variational} and~\ref{left_action_dirac}, and it is
these that we use throughout the present section. With this choice the product
connection on $P\cong SO(3)\times B$ is normalized by $A(\zeta_P)=\zeta$ with
$\zeta_P(R)=\zeta R$, so that $A(\dot R)=\dot RR^{-1}$; trivializing the adjoint
bundle by the section $R=I$ then gives
$\xi=\operatorname{Ad}_{R^{-1}}(\dot RR^{-1})=R^{-1}\dot R=\Omega$ and, dually,
$\mu=\Pi$.

Indeed, under \eqref{spacecraft_action} one has $R\mapsto AR$ and
$\dot R\mapsto A\dot R$, so that $(AR)^{-1}(A\dot R)=R^{-1}\dot R$. The body
angular velocity $\Omega$ is therefore invariant under the action, as is
$\Omega_r=\dot\theta$, and consequently
$\mathcal L(AR,\theta,\Omega,\Omega_r)=\mathcal L(R,\theta,\Omega,\Omega_r)$:
the mechanical Lagrangian is $SO(3)$-invariant. The quotient of the
configuration manifold is $Q/SO(3)\cong(S^1)^3$.


\subsection{Hamiltonian formulation}

We now pass to the Hamiltonian description of the spacecraft. The momenta
conjugate to $\Omega$ and $\Omega_r$ are
\begin{equation}
\label{spacecraft_momenta}
m
=
\frac{\partial\mathcal L}{\partial\Omega}
=
I\Omega+J(\Omega+\Omega_r),
\end{equation}
and
\begin{equation}
\label{spacecraft_rotor_momentum}
l
=
\frac{\partial\mathcal L}{\partial\Omega_r}
=
J(\Omega+\Omega_r).
\end{equation}
It follows that
\begin{equation}
\label{spacecraft_inverse_legendre}
\Omega
=
I^{-1}(m-l),
\qquad
\Omega_r
=
J^{-1}l-\Omega.
\end{equation}

Assuming that $I$ and $J$ are positive definite, the Legendre transformation
is hyperregular. We may therefore work on the Hamiltonian phase space $P=T^*Q.$

Using the body trivialization of $T^*SO(3)$, we identify
\[
P
\simeq
SO(3)\times(S^1)^3
\times\mathfrak{so}(3)^*
\times\mathbb R^3,
\]
and write a point of $P$ as $x=(R,\theta,m,l)$.

The Hamiltonian is obtained by the Legendre transformation,
\[
H
=
\langle m,\Omega\rangle
+
\langle l,\Omega_r\rangle
-
\mathcal L.
\]
Using \eqref{spacecraft_inverse_legendre}, it becomes
\begin{equation}
\label{spacecraft_hamiltonian}
H(R,\theta,m,l)
=
\frac12(m-l)^T I^{-1}(m-l)
+
\frac12l^T J^{-1}l.
\end{equation}

The corresponding Hamiltonian equations are
\begin{equation}
\label{spacecraft_uncontrolled_hamiltonian}
\begin{aligned}
\dot R
&=
R\widehat{\Omega},
\\
\dot\theta
&=
\Omega_r,
\\
\dot m
&=
m\times\Omega,
\\
\dot l
&=
0,
\end{aligned}
\end{equation}
where
\begin{equation}
\label{spacecraft_velocities}
\Omega
=
I^{-1}(m-l),
\qquad
\Omega_r
=
J^{-1}l-I^{-1}(m-l).
\end{equation}


\subsection{The induced symmetry on the state manifold}

The $SO(3)$-action \eqref{spacecraft_action} on $Q$ induces, by cotangent
lift, a free and proper action on $P=T^*Q.$

In the body trivialization introduced above, the action is
\begin{equation}
\label{spacecraft_phase_action}
A\cdot(R,\theta,m,l)
=
(AR,\theta,m,l).
\end{equation}
Thus the reduced state manifold is
\begin{equation}
\label{spacecraft_reduced_state}
B=P/SO(3)
\cong
(S^1)^3
\times\mathfrak{so}(3)^*
\times\mathbb R^3.
\end{equation}
A point of $B$ is written $z=(\theta,m,l)$. Notice that it is $B=P/G$, and not
$Q/G$, that is the base manifold entering the reduction of the optimal control
problem developed in the previous sections.

The principal bundle $\rho_P:P\rightarrow B$ is trivial, and we choose the
principal connection associated with the product decomposition
$P\simeq SO(3)\times B$. The adjoint and coadjoint bundles are then trivial,
\[
\widetilde{\mathfrak{so}(3)}
\cong
B\times\mathfrak{so}(3),
\qquad
\widetilde{\mathfrak{so}(3)}^{\,*}
\cong
B\times\mathfrak{so}(3)^*,
\]
and the connection is flat, $F_A=0$.

Accordingly, the general identifications obtained in the previous sections
become
\begin{equation}
\label{spacecraft_atiyah}
TP/SO(3)
\cong
TB\oplus\widetilde{\mathfrak{so}(3)}
\end{equation}
and
\begin{equation}
\label{spacecraft_reduced_cotangent}
T^*P/SO(3)
\cong
T^*B\oplus\widetilde{\mathfrak{so}(3)}^{\,*}.
\end{equation}

These are precisely the reduced tangent and cotangent bundles entering the
reduction procedure developed above.


\subsection{Optimal control extension}

We now introduce controls through torques applied to the three internal
reaction wheels. The control space is $U=\mathbb R^3$ with the trivial
$SO(3)$-action, so the control bundle is $C=P\times U\rightarrow P$ and a
control is written $u=(u_1,u_2,u_3)\in\mathbb R^3$. The applied torques modify
the rotor momentum equation to $\dot l=u$.
Thus the controlled dynamics is described by the bundle map
$\Gamma:C=P\times U\rightarrow TP$ given by
\begin{equation}
\label{spacecraft_controlled_vector_field}
\Gamma(R,\theta,m,l,u)
=
\left(
R\widehat{\Omega},
\Omega_r,
m\times\Omega,
u
\right),
\end{equation}
where $\Omega=I^{-1}(m-l)$ and $\Omega_r=J^{-1}l-\Omega$. Equivalently, the
controlled equations are
\begin{equation}
\label{spacecraft_controlled_dynamics}
\begin{aligned}
\dot R
&=
R\widehat{\Omega},
\\
\dot\theta
&=
J^{-1}l-\Omega,
\\
\dot m
&=
m\times\Omega,
\\
\dot l
&=
u.
\end{aligned}
\end{equation}

Introducing the input vector fields $Y_i=\partial/\partial l_i$ for
$i=1,2,3$, the controlled dynamics is equivalently
$\Gamma(x,u)=X_H(x)+\sum_{i=1}^3u_iY_i(x)$, which is control-affine with drift
$X_H$.

\begin{proposition}
The bundle map $\Gamma:P\times U\rightarrow TP$ is $SO(3)$-equivariant.
\end{proposition}

\begin{proof}
The Hamiltonian \eqref{spacecraft_hamiltonian} is invariant under the
$SO(3)$-action \eqref{spacecraft_phase_action}, so its Hamiltonian vector field
$X_H$ is equivariant. The input vector fields $Y_i=\partial/\partial l_i$ are
invariant, since the lifted action leaves the rotor momentum variables $l_i$
unchanged. The action on $U$ is trivial, so
$T\Phi_A^P\circ\Gamma=\Gamma\circ\Psi_A$ for every $A\in SO(3)$, with
$\Psi_A(x,u)=(A\cdot x,u)$. Hence $\Gamma$ is $SO(3)$-equivariant.
\end{proof}

We consider the quadratic running cost
\begin{equation}
\label{spacecraft_running_cost}
L:P\times U\longrightarrow\mathbb R,
\qquad
L(x,u)
=
\frac12|u|^2.
\end{equation}
Since the $SO(3)$-action on $U$ is trivial, the running cost is
$SO(3)$-invariant.

We have therefore obtained an $SO(3)$-invariant optimal control problem
\[
(P,C,\Gamma,L),
\qquad
P=T^*Q,
\qquad
C=P\times\mathbb R^3.
\]


\subsection{Pontryagin formulation}

We now construct the Pontryagin bundle associated with this optimal control
problem. Since the state manifold is $P=T^*Q$, the Pontryagin bundle is
\begin{equation}
\label{spacecraft_pontryagin_bundle}
M
=
T^*P\times_PC
=
T^*(T^*Q)\times_{T^*Q}
\bigl(T^*Q\times\mathbb R^3\bigr).
\end{equation}

A costate $\lambda\in T_x^*P$ is written, in the trivialization introduced
above, as $\lambda=(\Pi,p_\theta,p_m,p_l)$, where
\[
\Pi\in\mathfrak{so}(3)^*,
\qquad
p_\theta,p_m,p_l\in\mathbb R^3.
\]

The Pontryagin Hamiltonian is
\begin{equation}
\label{spacecraft_pontryagin_hamiltonian}
H_P(x,\lambda,u)
=
\langle\lambda,\Gamma(x,u)\rangle
-
\frac12|u|^2.
\end{equation}
Using \eqref{spacecraft_controlled_vector_field}, this becomes
\begin{equation}
\label{spacecraft_pontryagin_hamiltonian_explicit}
\begin{aligned}
H_P
={}&
\langle\Pi,\Omega\rangle
+
\left\langle
p_\theta,
J^{-1}l-\Omega
\right\rangle
+
\left\langle
p_m,
m\times\Omega
\right\rangle
\\
&+
\langle p_l,u\rangle
-
\frac12|u|^2,
\end{aligned}
\end{equation}
where $\Omega=I^{-1}(m-l)$. The primary constraint $\partial H_P/\partial u=0$
gives
\begin{equation}
\label{spacecraft_stationarity}
u=p_l,
\end{equation}
and $\partial^2H_P/\partial u^2=-\operatorname{Id}_{\mathbb R^3}$ is
nondegenerate, so the optimal control problem is regular.


\subsection{Reduction of the optimal control problem}

Since the action on $U=\mathbb R^3$ is trivial, the reduced control bundle is
$\widetilde U=C/SO(3)\cong B\times\mathbb R^3$, a vector bundle, and an element
of it is written $(z,u)=(\theta,m,l,u)$. The quotient of the Pontryagin bundle
is, by the general construction,
\begin{equation}
\label{spacecraft_reduced_pontryagin}
M/SO(3)
\cong
T^*B
\oplus
\widetilde{\mathfrak{so}(3)}^{\,*}
\oplus
\widetilde U.
\end{equation}
All the bundles appearing here are trivial in the present example, so that
$M/SO(3)\cong T^*B\oplus(B\times\mathfrak{so}(3)^*)\oplus(B\times\mathbb R^3)$.

The equivariant controlled vector field $\Gamma$ induces a reduced bundle map
$[\Gamma]:\widetilde U\rightarrow TP/SO(3)$, which the identification
\eqref{spacecraft_atiyah} decomposes as
$[\Gamma](z,u)=(\widetilde\Gamma(z,u),\widetilde\gamma(z,u))$.

For $z=(\theta,m,l)$, the reduced base dynamics is
\begin{equation}
\label{spacecraft_reduced_horizontal_dynamics}
\widetilde\Gamma(\theta,m,l,u)
=
\left(
J^{-1}l-\Omega,
m\times\Omega,
u
\right),
\end{equation}
whereas the vertical component associated with the eliminated attitude
variable is
\begin{equation}
\label{spacecraft_reduced_vertical_dynamics}
\widetilde\gamma(\theta,m,l,u)
=
\Omega
=
I^{-1}(m-l).
\end{equation}


\subsection{Reduced variational formulation}

\subsubsection{Reduced Lagrange--Pontryagin principle}

We now perform explicitly the reduction of the Lagrange--Pontryagin
variational principle for the spacecraft optimal control problem.

Recall that the unreduced Lagrange--Pontryagin functional is
\begin{equation}
\label{spacecraft_unreduced_LP}
S
=
\int_0^T
\left[
L(x,u)
+
\left\langle
\lambda,\dot x-\Gamma(x,u)
\right\rangle
\right]dt,
\end{equation}
with $x=(R,\theta,m,l)\in P$, $\lambda\in T_x^*P$ and $u\in\mathbb R^3$.
Equivalently, in Hamilton--Pontryagin form,
\begin{equation}
\label{spacecraft_unreduced_HP}
S
=
\int_0^T
\left[
\langle\lambda,\dot x\rangle
-
H_P(x,\lambda,u)
\right]dt.
\end{equation}

The functional is reduced using the decomposition
$TP/SO(3)\cong TB\oplus\widetilde{\mathfrak{so}(3)}$ determined by the
principal connection introduced above. A curve
$x(t)=(R(t),\theta(t),m(t),l(t))$ in $P$ has reduced base curve
$z(t)=(\theta(t),m(t),l(t))\in B$, and its velocity $\dot x$ determines the
reduced pair $(\dot z,\xi)\in TB\oplus\widetilde{\mathfrak{so}(3)}$, where
$\xi$ is the vertical component carried by the attitude variable. Dually, the
costate is represented by $(p_z,\mu)\in
T^*B\oplus\widetilde{\mathfrak{so}(3)}^{\,*}$ with $p_z=(p_\theta,p_m,p_l)$.

For the spacecraft the two components of the reduced controlled vector field
are
\[
\widetilde\gamma(z,u)=\Omega=I^{-1}(m-l),
\qquad
\widetilde\Gamma(z,u)=\bigl(J^{-1}l-\Omega,\ m\times\Omega,\ u\bigr).
\]
Under these identifications the pairing splits,
\[
\langle\lambda,\dot x-\Gamma(x,u)\rangle
=
\left\langle p_z,\dot z-\widetilde\Gamma(z,u)\right\rangle
+
\left\langle \mu,\xi-\widetilde\gamma(z,u)\right\rangle ,
\]
so that \eqref{spacecraft_unreduced_LP} reduces to the density \eqref{red_lag}
of the general theory, whose running-cost term is here $\tfrac12|u|^2$. (In
\eqref{red_lag} and \eqref{reduced_pontryagin_hamiltonian} the symbol $l$
denotes the reduced running cost; in the present section $l$ is the rotor
momentum, a component of $z=(\theta,m,l)$.) Written out,
\begin{equation}
\label{spacecraft_reduced_LP_explicit}
\begin{aligned}
S_G
=
\int_0^T
\Big[
&
\frac12|u|^2
+
\left\langle
p_\theta,
\dot\theta-J^{-1}l+\Omega
\right\rangle
\\
&
+
\left\langle
p_m,
\dot m-m\times\Omega
\right\rangle
+
\left\langle
p_l,
\dot l-u
\right\rangle
\\
&
+
\left\langle
\mu,
\xi-\Omega
\right\rangle
\Big]dt,
\end{aligned}
\end{equation}
where $\Omega=I^{-1}(m-l)$ throughout. Equivalently, the reduced Pontryagin
Hamiltonian \eqref{reduced_pontryagin_hamiltonian} is here
\begin{equation}
\label{spacecraft_reduced_hamiltonian_explicit}
\begin{aligned}
h
={}&
\left\langle
p_\theta,
J^{-1}l-\Omega
\right\rangle
+
\left\langle
p_m,
m\times\Omega
\right\rangle
+
\langle p_l,u\rangle
\\
&+
\langle\mu,\Omega\rangle
-
\frac12|u|^2,
\end{aligned}
\end{equation}
and hence
\begin{equation}
\label{spacecraft_reduced_HP_functional}
S_G
=
\int_0^T
\left[
\langle p_z,\dot z\rangle
+
\langle\mu,\xi\rangle
-
h(z,p_z,\mu,u)
\right]dt.
\end{equation}


\subsubsection{Reduced variations}

We now derive the reduced optimality equations directly from
\eqref{spacecraft_reduced_HP_functional}, rather than by quoting the general
result of Section~\ref{pmpreduc}.

The variations $\delta z=(\delta\theta,\delta m,\delta l)$ of the base
variables are arbitrary and vanish at the endpoints,
$\delta z(0)=\delta z(T)=0$; the variations $\delta p_\theta$, $\delta p_m$,
$\delta p_l$, $\delta\mu$ and $\delta u$ are arbitrary and unconstrained.

The variation of the vertical variable $\xi$, however, is constrained by
the reduction. With the convention adopted in the general reduction theory,
it is of the form
\begin{equation}
\label{spacecraft_reduced_vertical_variation}
\delta\xi
=
\frac{D\eta}{Dt}
+
[\xi,\eta]
+
F_A(\delta z,\dot z),
\end{equation}
where $\eta(t)\in\mathfrak{so}(3)$ is arbitrary and vanishes at the
endpoints. The product connection chosen in the present example is flat,
$F_A=0$, and therefore
\begin{equation}
\label{spacecraft_vertical_variation_flat}
\delta\xi
=
\dot\eta+[\xi,\eta].
\end{equation}

Taking the variation of
\eqref{spacecraft_reduced_HP_functional} gives
\begin{equation}
\label{spacecraft_delta_SG}
\begin{aligned}
\delta S_G
=
\int_0^T
\Big[
&
\langle\delta p_z,\dot z\rangle
+
\langle p_z,\delta\dot z\rangle
+
\langle\delta\mu,\xi\rangle
+
\langle\mu,\delta\xi\rangle
\\
&
-
\left\langle
\frac{\partial h}{\partial z},
\delta z
\right\rangle
-
\left\langle
\frac{\partial h}{\partial p_z},
\delta p_z
\right\rangle
-
\left\langle
\frac{\partial h}{\partial\mu},
\delta\mu
\right\rangle
-
\left\langle
\frac{\partial h}{\partial u},
\delta u
\right\rangle
\Big]dt.
\end{aligned}
\end{equation}

Two integrations by parts remove the derivatives from the variations. The
$\delta\dot z$ term gives
$\int_0^T\langle p_z,\delta\dot z\rangle\,dt
=-\int_0^T\langle\dot p_z,\delta z\rangle\,dt$,
the boundary term vanishing with $\delta z$. For the vertical term, substitute
\eqref{spacecraft_vertical_variation_flat}, integrate the $\dot\eta$ term by
parts using $\eta(0)=\eta(T)=0$, and use the coadjoint convention of this
paper, $\langle\mu,[\xi,\eta]\rangle=\langle\operatorname{ad}_\xi^*\mu,\eta\rangle$:
\[
\int_0^T
\langle\mu,\delta\xi\rangle\,dt
=
-
\int_0^T
\left\langle
\frac{D\mu}{Dt}
-
\operatorname{ad}_\xi^*\mu,
\eta
\right\rangle dt.
\]

Since all the remaining variations are arbitrary, the condition $\delta S_G=0$
yields
\begin{equation}
\label{spacecraft_reduced_variational_equations}
\left\{
\begin{aligned}
\dot z
&=
\frac{\partial h}{\partial p_z},
\\[0.4em]
\dot p_z
&=
-\frac{\partial h}{\partial z},
\\[0.4em]
\frac{D\mu}{Dt}
&=
\operatorname{ad}_\xi^*\mu,
\\[0.4em]
\xi
&=
\frac{\partial h}{\partial\mu},
\\[0.4em]
0
&=
\frac{\partial h}{\partial u}.
\end{aligned}
\right.
\end{equation}
This is \eqref{RHP} with $F_A=0$ and with the sign of the coadjoint equation reversed as in \eqref{RHP_left}, obtained here directly from the reduced
variational principle rather than by quoting the general reduced Pontryagin
Maximum Principle. The agreement is the first of the two independent checks
this section performs on the general theory; the second is in
Section~\ref{spacecraft_dirac_section}.


\subsubsection{Reduced optimality equations}
\label{spacecraft_reduced_eqs}

We now compute the equations in
\eqref{spacecraft_reduced_variational_equations} explicitly.

Variations with respect to the costates $p_\theta$, $p_m$ and $p_l$ give
$\dot z=\partial h/\partial p_z$, that is,
\begin{equation}
\label{spacecraft_reduced_state_equations}
\left\{
\begin{aligned}
\dot\theta
&=
J^{-1}l-\Omega,
\\
\dot m
&=
m\times\Omega,
\\
\dot l
&=
u,
\end{aligned}
\right.
\qquad
\Omega=I^{-1}(m-l).
\end{equation}

Variation with respect to $\mu$ gives $\xi=\partial h/\partial\mu$, that is,
\begin{equation}
\label{spacecraft_xi_equation}
\xi=\Omega=I^{-1}(m-l),
\end{equation}
and variation with respect to the control gives
$0=\partial h/\partial u=p_l-u$, so that the optimal control is
\begin{equation}
\label{spacecraft_optimal_control}
u=p_l.
\end{equation}
Here $\partial^2h/\partial u^2=-\operatorname{Id}_{\mathbb R^3}$, so the
problem is regular and the stationarity condition determines the control
uniquely.

The vertical variation gives the reduced coadjoint equation
\begin{equation}
\label{spacecraft_reduced_coadjoint}
\frac{D\mu}{Dt}
=
\operatorname{ad}_{\Omega}^{*}\mu,
\qquad\text{equivalently}\qquad
\dot\mu=\mu\times\Omega.
\end{equation}

It remains to determine the equations for the base costates. Since
\[
h
=
\left\langle
p_\theta,J^{-1}l-\Omega
\right\rangle
+
\left\langle
p_m,m\times\Omega
\right\rangle
+
\langle p_l,u\rangle
+
\langle\mu,\Omega\rangle
-
\frac12|u|^2,
\]
and $h$ is independent of $\theta$, we immediately obtain
\begin{equation}
\label{spacecraft_ptheta_equation}
\dot p_\theta=0.
\end{equation}

To compute the remaining derivatives, observe that
\begin{equation}
\label{spacecraft_delta_Omega}
\delta\Omega
=
I^{-1}(\delta m-\delta l).
\end{equation}
Throughout, $\mathfrak{so}(3)^*\simeq\mathbb R^3$ by the Euclidean scalar
product, and the inertia matrices $I$ and $J$ are symmetric.

For variations with respect to $m$ one has $\delta_m\Omega=I^{-1}\delta m$,
and hence
\begin{align*}
\delta_m h
={}&
-\left\langle
p_\theta,I^{-1}\delta m
\right\rangle
+
\left\langle
p_m,\delta m\times\Omega
\right\rangle
\\
&+
\left\langle
p_m,m\times I^{-1}\delta m
\right\rangle
+
\left\langle
\mu,I^{-1}\delta m
\right\rangle.
\end{align*}
The two cross-product terms are rearranged by the identities
$\langle p_m,\delta m\times\Omega\rangle=\langle\Omega\times p_m,\delta m\rangle$
and
$\langle p_m,m\times I^{-1}\delta m\rangle=\langle I^{-1}(p_m\times m),\delta m\rangle$,
the second using the symmetry of $I$. This gives
\[
\frac{\partial h}{\partial m}
=
-I^{-1}p_\theta
+
\Omega\times p_m
+
I^{-1}(p_m\times m)
+
I^{-1}\mu.
\]
Therefore,
\begin{equation}
\label{spacecraft_pm_equation}
\dot p_m
=
I^{-1}p_\theta
-
\Omega\times p_m
-
I^{-1}(p_m\times m)
-
I^{-1}\mu.
\end{equation}

For variations with respect to $l$ one has $\delta_l\Omega=-I^{-1}\delta l$,
and therefore
\begin{align*}
\delta_l h
={}&
\left\langle
p_\theta,
J^{-1}\delta l+I^{-1}\delta l
\right\rangle
\\
&-
\left\langle
p_m,
m\times I^{-1}\delta l
\right\rangle
-
\left\langle
\mu,I^{-1}\delta l
\right\rangle.
\end{align*}
It follows that
\[
\frac{\partial h}{\partial l}
=
J^{-1}p_\theta
+
I^{-1}p_\theta
-
I^{-1}(p_m\times m)
-
I^{-1}\mu,
\]
and hence
\begin{equation}
\label{spacecraft_pl_equation}
\dot p_l
=
-J^{-1}p_\theta
-
I^{-1}p_\theta
+
I^{-1}(p_m\times m)
+
I^{-1}\mu.
\end{equation}

Substituting the optimal control $u=p_l$ and the reconstruction variable
$\xi=\Omega$ into the preceding equations gives the complete reduced
optimality system
\begin{equation}
\label{spacecraft_complete_reduced_PMP}
\left\{
\begin{aligned}
\dot\theta
&=
J^{-1}l-\Omega,
\\[0.3em]
\dot m
&=
m\times\Omega,
\\[0.3em]
\dot l
&=
p_l,
\\[0.5em]
\dot p_\theta
&=
0,
\\[0.3em]
\dot p_m
&=
I^{-1}p_\theta
-
\Omega\times p_m
-
I^{-1}(p_m\times m)
-
I^{-1}\mu,
\\[0.3em]
\dot p_l
&=
-J^{-1}p_\theta
-
I^{-1}p_\theta
+
I^{-1}(p_m\times m)
+
I^{-1}\mu,
\\[0.5em]
\frac{D\mu}{Dt}
&=
\operatorname{ad}_{\Omega}^{*}\mu,
\\[0.3em]
u
&=
p_l,
\end{aligned}
\right.
\qquad
\Omega=I^{-1}(m-l).
\end{equation}

Under the standard identification $\mathfrak{so}(3)^*\simeq\mathbb R^3$ the
coadjoint equation reads $\dot\mu=\mu\times\Omega$. Two consistency checks are
worth recording. First, $\mu$ and the state variable $m$ are both body
quantities in $\mathfrak{so}(3)^*$ attached to the same $SO(3)$ symmetry, and
\eqref{spacecraft_complete_reduced_PMP} gives them the same Lie--Poisson form,
$\dot m=\operatorname{ad}_\Omega^*m$ and
$\dot\mu=\operatorname{ad}_\Omega^*\mu$. Second, since $H_P$ is
$SO(3)$-invariant, the spatial costate
$\nu=\operatorname{Ad}^*_{R^{-1}}\mu=R\mu$ is conserved, and
$\dot\nu=R(\dot\mu+\Omega\times\mu)$ vanishes precisely when
$\dot\mu=\mu\times\Omega$.

The reduced optimality system has therefore been obtained directly by reducing
the Lagrange--Pontryagin variational principle. The attitude variable $R$ no
longer appears among the reduced state variables; its contribution is carried
by the vertical variable $\xi=\Omega$ and the coadjoint variable
$\mu\in\mathfrak{so}(3)^*$. The next subsection recovers the same system
independently, from the Lie--Dirac reduction of the implicit formulation.

\subsection{Reduced Dirac formulation}
\label{spacecraft_dirac_section}
\subsubsection{The reduced Dirac structure}

The canonical Dirac structure $D_P$ induces a reduced Dirac structure
$[D_P]_{SO(3)}$ on
$T^*P/SO(3)\cong T^*B\oplus\widetilde{\mathfrak{so}(3)}^{\,*}$. At a point
$(z,p_z,\mu)$, an element of the reduced tangent bundle is written
$(v_z,v_{p_z};\xi,v_\mu)$ and an element of the dual bundle is written
$(\pi_z,\pi_{p_z};\alpha_\xi,\alpha_{v_\mu})$.

According to the reduced canonical Dirac structure
\eqref{reduced_canonical_dirac_explicit} in the left-action
convention of Remark~\ref{left_action_dirac}, that is, with the coadjoint relation \eqref{reduced_canonical_dirac_left}, these variables satisfy
\begin{equation}
\label{spacecraft_reduced_dirac_structure_general}
\left\{
\begin{aligned}
v_z
&=
\pi_{p_z},
\\
v_{p_z}+\pi_z
&=
-\left\langle
\mu,F_A(v_z,\cdot)
\right\rangle,
\\
\xi
&=
\alpha_{v_\mu},
\\
v_\mu+\alpha_\xi
&=
\operatorname{ad}_\xi^*\mu.
\end{aligned}
\right.
\end{equation}

The product connection chosen in the spacecraft example is flat, $F_A=0$, so
the reduced Dirac structure simplifies to
\begin{equation}
\label{spacecraft_reduced_dirac_structure}
\left\{
\begin{aligned}
v_z
&=
\pi_{p_z},
\\
v_{p_z}
&=
-\pi_z,
\\
\xi
&=
\alpha_{v_\mu},
\\
v_\mu+\alpha_\xi
&=
\operatorname{ad}_\xi^*\mu.
\end{aligned}
\right.
\end{equation}


\subsubsection{Reduced implicit Dirac system}

The reduced Pontryagin Hamiltonian is the one already obtained in
\eqref{spacecraft_reduced_hamiltonian_explicit}, with
$\Omega=I^{-1}(m-l)$. The reduced implicit Dirac system is
\begin{equation}
\label{spacecraft_reduced_implicit_dirac}
\left(
[X]_{SO(3)},
dh
\right)
\in
[D_P]_{SO(3)},
\end{equation}
together with the reduced primary constraint $\partial h/\partial u=0$. The
problem is regular, so that constraint determines the control locally as a
function of the remaining reduced variables, and on the reduced primary
constraint manifold
$[M_1]_{SO(3)}=\{(z,p_z,\mu,u):\partial h/\partial u=0\}$ the reduced
Pontryagin Hamiltonian determines a Hamiltonian on
$T^*B\oplus\widetilde{\mathfrak{so}(3)}^{\,*}$, as in
Section~\ref{ldrpmp}. Along a curve in $[M_1]_{SO(3)}$ the reduced tangent
component of the dynamics is represented by
\begin{equation}
\label{spacecraft_reduced_dirac_vector}
[X]_{SO(3)}
=
\left(
\dot z,
\dot p_z;
\xi,
\frac{D\mu}{Dt}
\right).
\end{equation}

The differential of the reduced Hamiltonian is represented, in the dual
variables introduced above, by
\begin{equation}
\label{spacecraft_reduced_dirac_dh}
\left(
\pi_z,\pi_{p_z};
\alpha_\xi,\alpha_{v_\mu}
\right)
=
\left(
\frac{\partial h}{\partial z},
\frac{\partial h}{\partial p_z};
0,
\frac{\partial h}{\partial\mu}
\right).
\end{equation}
In particular $\alpha_\xi=0$ and $\alpha_{v_\mu}=\partial h/\partial\mu$: the
derivative with respect to the coadjoint variable $\mu$ is paired with the
variation $v_\mu$, and so corresponds to $\alpha_{v_\mu}$ and not to
$\alpha_\xi$.

Substituting
\eqref{spacecraft_reduced_dirac_vector} and
\eqref{spacecraft_reduced_dirac_dh}
into the defining relations
\eqref{spacecraft_reduced_dirac_structure}, we obtain
\begin{equation}
\label{spacecraft_dirac_equations_abstract}
\left\{
\begin{aligned}
\dot z
&=
\frac{\partial h}{\partial p_z},
\\[0.4em]
\dot p_z
&=
-\frac{\partial h}{\partial z},
\\[0.4em]
\xi
&=
\frac{\partial h}{\partial\mu},
\\[0.4em]
\frac{D\mu}{Dt}
&=
\operatorname{ad}_\xi^*\mu,
\\[0.4em]
0
&=
\frac{\partial h}{\partial u}.
\end{aligned}
\right.
\end{equation}

The four relations of \eqref{spacecraft_reduced_dirac_structure} give the
first four equations in order, the fourth using $\alpha_\xi=0$, and the last
equation is the reduced primary constraint.

Thus, before any particular derivative is computed, the reduced Dirac system
reproduces the equations obtained from the reduced Lagrange--Pontryagin
principle, namely \eqref{spacecraft_reduced_variational_equations}. The flat
product connection is what removes the curvature term from the reduced costate
equation, and the vertical part of the reduced Dirac structure is what
produces the coadjoint equation. It remains to write these equations out
explicitly for the spacecraft.

\subsubsection{Explicit Dirac equations}
The five equations of \eqref{spacecraft_dirac_equations_abstract} are now
evaluated. The reduced Hamiltonian $h$ entering them is
\eqref{spacecraft_reduced_hamiltonian_explicit}, which is the same function of
$(z,p_z,\mu,u)$ that appeared in the variational route, so its partial
derivatives have already been computed in Section~\ref{spacecraft_reduced_eqs}
and need not be recomputed here.

The state equations $\dot z=\partial h/\partial p_z$ return
\eqref{spacecraft_reduced_state_equations}. The vertical equation
$\xi=\partial h/\partial\mu$ returns $\xi=\Omega$, and the primary constraint
$0=\partial h/\partial u=p_l-u$ returns $u=p_l$. The vertical component of the
reduced Dirac relation, the fourth line of
\eqref{spacecraft_reduced_dirac_structure} with $\alpha_\xi=0$, gives the
coadjoint equation $D\mu/Dt=\operatorname{ad}^*_\Omega\mu$, that is,
$\dot\mu=\mu\times\Omega$. Finally $\dot p_z=-\partial h/\partial z$ gives
$\dot p_\theta=0$, since $h$ is independent of $\theta$, together with
\eqref{spacecraft_pm_equation} and \eqref{spacecraft_pl_equation} from
$\partial h/\partial m$ and $\partial h/\partial l$.

Substituting $u=p_l$ and $\xi=\Omega$, the complete reduced Dirac system is
therefore
\begin{equation}
\label{spacecraft_complete_reduced_dirac_system}
\left\{
\begin{aligned}
\dot\theta
&=
J^{-1}l-\Omega,
\\[0.3em]
\dot m
&=
m\times\Omega,
\\[0.3em]
\dot l
&=
p_l,
\\[0.5em]
\dot p_\theta
&=
0,
\\[0.3em]
\dot p_m
&=
I^{-1}p_\theta
-
\Omega\times p_m
-
I^{-1}(p_m\times m)
-
I^{-1}\mu,
\\[0.3em]
\dot p_l
&=
-J^{-1}p_\theta
-
I^{-1}p_\theta
+
I^{-1}(p_m\times m)
+
I^{-1}\mu,
\\[0.5em]
\frac{D\mu}{Dt}
&=
\operatorname{ad}_{\Omega}^{*}\mu,
\\[0.3em]
u
&=
p_l,
\end{aligned}
\right.
\qquad
\Omega=I^{-1}(m-l).
\end{equation}

Comparing
\eqref{spacecraft_complete_reduced_dirac_system}
with
\eqref{spacecraft_complete_reduced_PMP},
the two systems are term by term the same. Thus, the reduction of
the implicit Dirac formulation and the reduction of the
Lagrange--Pontryagin variational principle produce exactly the same reduced
optimality equations for the spacecraft with internal rotors.

This provides an explicit realization, for the spacecraft optimal control
problem, of the equivalence between the reduced variational and reduced
Dirac descriptions established in the general theory.

\section{Application: the sub-Riemannian geodesic problem on the Heisenberg
group}
\label{heisenberg}

The spacecraft of Section~\ref{sec:spacecraft} was reduced by a nonabelian group,
but with respect to a flat product connection, so that the curvature term in
the reduced equations \eqref{RHP} and in the reduced Dirac structure
\eqref{reduced_canonical_dirac_explicit} played no role. We now treat an
example of the opposite kind. The symmetry group is the center of the
Heisenberg group, hence abelian, so that the coadjoint term drops out; but the
natural connection is not flat, its curvature being nowhere zero. (The bundle
itself is trivial: $\sigma(x,y)=(x,y,0)$ is a global section, and
$\sigma(x,y)\cdot s=(x,y,s)$ is a global equivariant trivialization. It is the
connection, not the topology of the bundle, that distinguishes this example
from the previous one; nonzero curvature is perfectly compatible with a trivial
bundle.) Taken together, the two examples exercise both of the terms that
distinguish the reduced Pontryagin Maximum Principle from its unreduced
counterpart.

A further feature of the present example is that the reduction is not merely a
convenience: the vertical variable is precisely the one eliminated by passing
to the sub-Riemannian picture, and the reconstruction equation turns out to be
the horizontality condition itself.


\subsection{The Heisenberg group and its horizontal distribution}

Let $H$ denote the three-dimensional Heisenberg group, that is, $\mathbb R^3$
with coordinates $(x,y,w)$ and the multiplication
\begin{equation}
\label{heis_group_law}
(x_1,y_1,w_1)\cdot(x_2,y_2,w_2)
=
\Bigl(
x_1+x_2,\;
y_1+y_2,\;
w_1+w_2+\tfrac12(x_1y_2-y_1x_2)
\Bigr).
\end{equation}
A basis of left-invariant vector fields is
\begin{equation}
\label{heis_frame}
X=\frac{\partial}{\partial x}-\frac{y}{2}\frac{\partial}{\partial w},
\qquad
Y=\frac{\partial}{\partial y}+\frac{x}{2}\frac{\partial}{\partial w},
\qquad
W=\frac{\partial}{\partial w},
\end{equation}
and the only nonvanishing bracket is
\begin{equation}
\label{heis_bracket}
[X,Y]=W.
\end{equation}
The distribution
\[
\mathcal H=\mathrm{span}\{X,Y\}\subset TH
\]
is bracket generating of step two, and the sub-Riemannian structure is obtained
by declaring $X,Y$ to be orthonormal.


\subsection{The optimal control problem}

We take as state manifold
\[
P=H\cong\mathbb R^3,
\]
as control bundle the trivial bundle
\[
C=P\times U,
\qquad
U=\mathbb R^2,
\qquad
u=(u_1,u_2),
\]
and as controlled dynamics the bundle map
\begin{equation}
\label{heis_control_vector_field}
\Gamma(x,y,w,u)=u_1X+u_2Y,
\end{equation}
that is,
\begin{equation}
\label{heis_controlled_dynamics}
\dot x=u_1,
\qquad
\dot y=u_2,
\qquad
\dot w=\tfrac12(xu_2-yu_1).
\end{equation}
The running cost is the sub-Riemannian energy
\begin{equation}
\label{heis_cost}
L(x,y,w,u)=\tfrac12|u|^2=\tfrac12(u_1^2+u_2^2).
\end{equation}
Minimizing \eqref{heis_cost} subject to \eqref{heis_controlled_dynamics} is the
sub-Riemannian geodesic problem on $H$; see \cite{Mo02}. Since the admissible
curves are exactly the horizontal ones, the quadruple $(P,C,\Gamma,L)$ is an
optimal control problem in the sense of Definition~\ref{ocp}.


\subsection{The symmetry group and the reduced bundles}

Let
\[
G=Z(H)=\{(0,0,s)\ :\ s\in\mathbb R\}\cong\mathbb R
\]
be the center of $H$, acting on $P$ by right translation,
\begin{equation}
\label{heis_action}
(x,y,w)\cdot s=(x,y,w+s).
\end{equation}
Because $G$ is central, this coincides with the corresponding left action, and
the left and right conventions of Remarks~\ref{left_action_variational}
and~\ref{left_action_dirac} agree; we may therefore use the equations of
Sections~\ref{pmpreduc} and~\ref{drec} exactly as they stand.

The action is free and proper, $G$ acts trivially on $U$, and neither $\Gamma$
nor $L$ depends on $w$, so that $(P,C,\Gamma,L)$ is $G$-invariant in the sense
of Definition~\ref{ocp_symmetry}. The base of the reduction is
\begin{equation}
\label{heis_base}
B=P/G\cong\mathbb R^2,
\qquad
z=(x,y),
\end{equation}
and $\rho:H\rightarrow\mathbb R^2$, $\rho(x,y,w)=(x,y)$, is a principal
$\mathbb R$-bundle.

Since $\mathfrak g=\mathbb R$ is abelian, the adjoint and coadjoint bundles are
trivial,
\begin{equation}
\label{heis_adjoint_bundles}
\widetilde{\mathfrak g}\cong B\times\mathbb R,
\qquad
\widetilde{\mathfrak g}^{\,*}\cong B\times\mathbb R,
\qquad
\operatorname{ad}_\xi^*\mu=0
\quad\text{for all }\xi,\mu,
\end{equation}
and the induced covariant derivative on $\widetilde{\mathfrak g}$ and
$\widetilde{\mathfrak g}^{\,*}$ reduces to the ordinary time derivative. The
reduced control bundle is $\widetilde U\cong B\times\mathbb R^2$, so that the
reduced Pontryagin bundle is
\begin{equation}
\label{heis_reduced_pontryagin}
M/G
\cong
T^*\mathbb R^2\oplus\widetilde{\mathfrak g}^{\,*}\oplus\widetilde U,
\end{equation}
with coordinates $(x,y,p_x,p_y,\mu,u)$.

\subsubsection*{The connection and its curvature}

We take as principal connection the $\mathbb R$-valued one-form
\begin{equation}
\label{heis_connection}
A=dw-\tfrac12(x\,dy-y\,dx).
\end{equation}
It is indeed a principal connection: the fundamental vector field of
$\zeta\in\mathfrak g=\mathbb R$ is $\zeta_P=\zeta\,\partial_w$, whence
$A(\zeta_P)=\zeta$, and $A$ is invariant under \eqref{heis_action} because its
coefficients do not involve $w$. Moreover
\[
A(X)=\frac{y}{2}-\frac{y}{2}=0,
\qquad
A(Y)=-\frac{x}{2}+\frac{x}{2}=0,
\]
so that the horizontal subspace of $A$ is exactly the sub-Riemannian
distribution,
\begin{equation}
\label{heis_horizontal}
\ker A=\mathcal H=\mathrm{span}\{X,Y\}.
\end{equation}
Since $\mathfrak g$ is abelian, $[A,A]=0$ and the curvature is
\begin{equation}
\label{heis_curvature}
F_A=dA=-\,dx\wedge dy,
\end{equation}
which is a nowhere vanishing basic two-form on $B=\mathbb R^2$. This is the
essential difference with the example of Section~\ref{sec:spacecraft}.

\subsubsection*{The reduced dynamics}

The decomposition $[\Gamma](z,u)=(\widetilde\Gamma(z,u),\widetilde\gamma(z,u))$
determined by $A$ is
\begin{equation}
\label{heis_reduced_dynamics}
\widetilde\Gamma(z,u)=u_1\frac{\partial}{\partial x}+u_2\frac{\partial}{\partial y},
\qquad
\widetilde\gamma(z,u)=A\bigl(\Gamma(x,y,w,u)\bigr)=0,
\end{equation}
the vertical part vanishing identically by \eqref{heis_horizontal}: the
controlled vector field is horizontal, which is precisely the statement that
admissible curves of the sub-Riemannian problem are horizontal curves.


\subsection{Pontryagin formulation}

Writing a costate as $\lambda=(\lambda_x,\lambda_y,\lambda_w)\in T^*_{(x,y,w)}P$,
the Pontryagin Hamiltonian $H_P=\langle\lambda,\Gamma\rangle-L$ is
\begin{equation}
\label{heis_pontryagin_hamiltonian}
H_P
=
\lambda_xu_1+\lambda_yu_2
+
\frac{\lambda_w}{2}(xu_2-yu_1)
-
\frac12|u|^2 .
\end{equation}
The primary constraints $\partial H_P/\partial u^a=0$ read
\begin{equation}
\label{heis_primary_constraints}
u_1=\lambda_x-\frac{y}{2}\lambda_w,
\qquad
u_2=\lambda_y+\frac{x}{2}\lambda_w,
\end{equation}
and
\[
\frac{\partial^2H_P}{\partial u^a\partial u^b}=-\operatorname{Id}_{\mathbb R^2},
\]
so that the problem is regular at every point, and the standing assumptions of
Section~\ref{ocpsec} are met. (Regularity of the normal problem does not by
itself exclude abnormal extremals, which live on the leaf $p_0=0$ where the
Hessian in the controls vanishes; that the Heisenberg problem has no
nontrivial abnormal extremals is proved in
Corollary~\ref{heisenberg_no_abnormal}.)
The unreduced first-order conditions are \eqref{heis_controlled_dynamics}
together with
\begin{equation}
\label{heis_unreduced_costate}
\dot\lambda_x=-\frac{\lambda_w}{2}u_2,
\qquad
\dot\lambda_y=\frac{\lambda_w}{2}u_1,
\qquad
\dot\lambda_w=0 .
\end{equation}

Under the identification \eqref{cotangent_quotient} determined by $A$, the
momentum map is $J(x,\lambda)=\langle\lambda,\partial_w\rangle=\lambda_w$ and
the horizontal part of $\lambda$ is obtained by evaluating on the frame
\eqref{heis_frame}. Hence the reduced variables are
\begin{equation}
\label{heis_reduced_variables}
p_x=\lambda_x-\frac{y}{2}\lambda_w,
\qquad
p_y=\lambda_y+\frac{x}{2}\lambda_w,
\qquad
\mu=\lambda_w .
\end{equation}


\subsection{Reduced variational formulation}

By \eqref{heis_reduced_dynamics} the reduced Pontryagin Hamiltonian
\eqref{reduced_pontryagin_hamiltonian} is
\begin{equation}
\label{heis_reduced_hamiltonian}
h(z,p_z,\mu,u)
=
\bigl\langle p_z,\widetilde\Gamma(z,u)\bigr\rangle
+
\bigl\langle\mu,\widetilde\gamma(z,u)\bigr\rangle
-
l(z,u)
=
p_xu_1+p_yu_2-\frac12|u|^2 .
\end{equation}
Substituting \eqref{heis_reduced_variables} shows that $h$ coincides with
\eqref{heis_pontryagin_hamiltonian}, as it must. Note that $h$ is independent
of both $z$ and $\mu$: the entire dependence on the eliminated variable is
carried by the curvature.

The reduced Lagrange--Pontryagin functional \eqref{red_funct_hamiltonian} is
\begin{equation}
\label{heis_reduced_HP_functional}
S_G
=
\int_0^T
\Bigl[
\langle p_z,\dot z\rangle+\langle\mu,\xi\rangle-h(z,p_z,\mu,u)
\Bigr]dt,
\end{equation}
and, since $\mathfrak g$ is abelian, the constrained variations
\eqref{reduced_vertical_variation} reduce to
\begin{equation}
\label{heis_constrained_variation}
\delta\xi=\dot\eta+F_A(\delta z,\dot z)
=
\dot\eta-\bigl(\delta x\,\dot y-\delta y\,\dot x\bigr),
\end{equation}
with $\eta(0)=\eta(T)=0$. The bracket term $[\eta,\xi]$, which produced the
coadjoint equation in Section~\ref{sec:spacecraft}, is absent here; what survives
is exactly the curvature term.

Taking variations and integrating by parts as in Section~\ref{pmpreduc},
\[
\int_0^T\langle\mu,\delta\xi\rangle\,dt
=
-\int_0^T\langle\dot\mu,\eta\rangle\,dt
+
\int_0^T\langle\mu,F_A(\delta z,\dot z)\rangle\,dt,
\]
so that stationarity of \eqref{heis_reduced_HP_functional} yields
\begin{equation}
\label{heis_reduced_equations}
\left\{
\begin{aligned}
\dot z
&=
\frac{\partial h}{\partial p_z},
\\[0.4em]
\dot p_z
&=
-\frac{\partial h}{\partial z}
-
\bigl\langle\mu,F_A(\dot z,\,\cdot\,)\bigr\rangle,
\\[0.4em]
\dot\mu
&=
0,
\\[0.4em]
\xi
&=
\frac{\partial h}{\partial\mu},
\\[0.4em]
0
&=
\frac{\partial h}{\partial u},
\end{aligned}
\right.
\end{equation}
which is \eqref{RHP} with $\operatorname{ad}^*_\xi\mu=0$.

\subsubsection*{The reduced optimality equations}

We compute \eqref{heis_reduced_equations} explicitly. The stationarity
condition in the control gives
\begin{equation}
\label{heis_optimal_control}
0=\frac{\partial h}{\partial u}=p_z-u,
\qquad\text{that is,}\qquad
u=(p_x,p_y).
\end{equation}
Since $h$ is independent of $\mu$,
\begin{equation}
\label{heis_xi}
\xi=\frac{\partial h}{\partial\mu}=0 .
\end{equation}
For the costate equation, $\partial h/\partial z=0$ and, by
\eqref{heis_curvature},
\[
F_A(\dot z,\,\cdot\,)
=
-\bigl(\dot x\,dy-\dot y\,dx\bigr)
=
u_2\,dx-u_1\,dy ,
\]
whence
\begin{equation}
\label{heis_costate}
\dot p_x=-\mu u_2,
\qquad
\dot p_y=\mu u_1 .
\end{equation}
Collecting \eqref{heis_optimal_control}--\eqref{heis_costate}, the complete
reduced optimality system is
\begin{equation}
\label{heis_complete_reduced_PMP}
\left\{
\begin{aligned}
\dot x&=p_x,
\\[0.2em]
\dot y&=p_y,
\\[0.4em]
\dot p_x&=-\mu\,p_y,
\\[0.2em]
\dot p_y&=\mu\,p_x,
\\[0.4em]
\dot\mu&=0,
\\[0.2em]
\xi&=0,
\\[0.2em]
u&=(p_x,p_y).
\end{aligned}
\right.
\end{equation}

\begin{remark}
Equations \eqref{heis_complete_reduced_PMP} may be checked directly against the
unreduced conditions \eqref{heis_unreduced_costate}. Differentiating
\eqref{heis_reduced_variables} along an unreduced extremal and using
$\dot\lambda_w=0$ gives
\[
\dot p_x
=
\dot\lambda_x-\frac{\dot y}{2}\lambda_w
=
-\frac{\lambda_w}{2}u_2-\frac{u_2}{2}\lambda_w
=
-\mu u_2,
\qquad
\dot p_y
=
\dot\lambda_y+\frac{\dot x}{2}\lambda_w
=
\mu u_1,
\]
in agreement with \eqref{heis_costate}. This provides an independent
verification of the sign of the curvature term in \eqref{RHP} and in
\eqref{reduced_canonical_dirac_explicit}, which the example of
Section~\ref{sec:spacecraft} could not supply.
\end{remark}


\subsection{Reduced Dirac formulation}

We now recover \eqref{heis_complete_reduced_PMP} from the Lie--Dirac reduction.
By Proposition~\ref{red_Dirac_structure}, and using
$\operatorname{ad}_\xi^*\mu=0$ together with \eqref{heis_curvature}, the
reduced canonical Dirac structure on
$T^*\mathbb R^2\oplus\widetilde{\mathfrak g}^{\,*}$ is
\begin{equation}
\label{heis_reduced_dirac_structure}
\left\{
\begin{aligned}
v_z
&=
\pi_{p_z},
\\
v_{p_z}+\pi_z
&=
-\bigl\langle\mu,F_A(v_z,\cdot)\bigr\rangle,
\\
\xi
&=
\alpha_{v_\mu},
\\
v_\mu+\alpha_\xi
&=
0 .
\end{aligned}
\right.
\end{equation}
Note that, in contrast with \eqref{spacecraft_reduced_dirac_structure}, the
second relation is not the canonical one: the vertical variable $\mu$ enters
the horizontal block through the curvature.

Along a curve in the reduced primary constraint manifold
\[
[M_1]_G=\left\{(z,p_z,\mu,u)\in M/G\ \Big|\ \frac{\partial h}{\partial u}=0\right\},
\]
the reduced tangent component of the dynamics is
\begin{equation}
\label{heis_reduced_dirac_vector}
[X]_G=\bigl(\dot z,\dot p_z;\xi,\dot\mu\bigr),
\end{equation}
while the differential of the reduced Hamiltonian is represented by
\begin{equation}
\label{heis_reduced_dirac_dh}
\bigl(\pi_z,\pi_{p_z};\alpha_\xi,\alpha_{v_\mu}\bigr)
=
\left(
\frac{\partial h}{\partial z},
\frac{\partial h}{\partial p_z};
0,
\frac{\partial h}{\partial\mu}
\right).
\end{equation}
Substituting \eqref{heis_reduced_dirac_vector} and
\eqref{heis_reduced_dirac_dh} into \eqref{heis_reduced_dirac_structure}, the
first relation gives $\dot z=\partial h/\partial p_z$; the second gives
\[
\dot p_z
=
-\frac{\partial h}{\partial z}
-
\bigl\langle\mu,F_A(\dot z,\,\cdot\,)\bigr\rangle;
\]
the third gives $\xi=\partial h/\partial\mu$; and the fourth, since
$\alpha_\xi=0$, gives $\dot\mu=0$. Together with the reduced primary constraint
$\partial h/\partial u=0$ these are precisely
\eqref{heis_reduced_equations}, hence \eqref{heis_complete_reduced_PMP}.

Thus, as in Section~\ref{sec:spacecraft}, the reduction of the implicit Dirac
formulation and the reduction of the Lagrange--Pontryagin variational principle
produce the same reduced optimality equations; here, however, the agreement is
tested on the curvature term rather than on the coadjoint term.


\subsection{Reconstruction and explicit solution}

The reduced system \eqref{heis_complete_reduced_PMP} is readily integrated.
Since $\mu$ is constant and
\[
\frac{d}{dt}(p_x+ip_y)=i\mu\,(p_x+ip_y),
\]
we obtain, for $\mu\neq0$,
\begin{equation}
\label{heis_momentum_solution}
p_x(t)+ip_y(t)=r\,e^{i(\mu t+\varphi)},
\qquad
r=|p_z(0)|,
\end{equation}
so that $|u|=r$ is constant and the extremal is parametrized proportionally to
arclength. Integrating $\dot z=p_z$,
\begin{equation}
\label{heis_base_solution}
x(t)+iy(t)
=
c+\frac{r}{i\mu}\,e^{i(\mu t+\varphi)},
\qquad
c=\bigl(x(0)+iy(0)\bigr)+\frac{ir}{\mu}e^{i\varphi},
\end{equation}
that is, the reduced base curve is a circle of center $c$ and radius
$r/|\mu|$, traversed with angular velocity $\mu$. For $\mu=0$ the momentum
$p_z$ is constant and the base curve is a straight line.

The reconstruction equation is $\xi=0$, that is, $A(\dot x)=0$: the extremal in
$P$ is the \emph{horizontal lift} of the reduced base curve. Explicitly,
\begin{equation}
\label{heis_reconstruction}
\dot w=\tfrac12\bigl(x\dot y-y\dot x\bigr),
\qquad\text{so that}\qquad
w(T)-w(0)=\frac12\int_0^T\bigl(x\,dy-y\,dx\bigr),
\end{equation}
the signed area swept by the segment joining the origin to the moving base
point. One recovers in this way the classical description of Heisenberg
geodesics: they project to circles and straight lines in the plane, the vertical
increment being the swept area \cite{Mo02}. For an open projected path this
swept area is not an intrinsic enclosed area: it depends on the choice of
origin, changing by $\tfrac12(a\,\Delta y-b\,\Delta x)$ under the translation of
the origin by $(a,b)$, in step with the gauge dependence of
$A=dw-\tfrac12(x\,dy-y\,dx)$. It is the enclosed signed area exactly when the
projection is closed, as for a full circle.

\begin{remark}
The reduced equations \eqref{heis_complete_reduced_PMP} are those of a particle
of unit mass and unit charge moving in the plane under the constant magnetic
two-form
\[
-\mu F_A=\mu\,dx\wedge dy ,
\qquad\text{that is,}\qquad
\ddot x=-\mu\dot y,
\quad
\ddot y=\mu\dot x .
\]
The sign is that of the curvature term $\langle\mu,F_A(\cdot,\dot z)\rangle
=-C_{\dot z}\mu$ of \eqref{RHP}, and it is the same sign that appears in Wong's
equation \eqref{normal_wong} in Section~\ref{strictness}. Only the product of
charge and field strength is physical, and we fix the convention once: unit
charge in the effective field $-\mu F_A$. Equivalently, one may say charge $\mu$
in the normalized field $-F_A=dx\wedge dy$; what one may not do is combine the
two readings, since the vertical momentum $\mu$ --- the value of the momentum
map, conserved by Noether's theorem --- would then be counted twice. With
either convention the effective field strength seen by the particle is $\mu$.
This is the standard mechanical interpretation of the curvature term in
\eqref{RHP}: after reduction, the eliminated variable reappears as a magnetic
coupling on the base.
\end{remark}

\begin{remark}
The two applications are complementary. In Section~\ref{sec:spacecraft} the group
$SO(3)$ is nonabelian and the connection is flat, so that the reduced equations
test the coadjoint term $\operatorname{ad}^*_\xi\mu$ and not the curvature. Here
the group is abelian and the coadjoint term vanishes identically, while the
curvature is nowhere zero. Since every proper \emph{connected} Lie subgroup of
$H$ is abelian, an example exhibiting both terms simultaneously requires a
different state manifold; the reduction of a mechanical system on a principal
bundle with nonabelian structure group and non-flat connection, such as a rigid
body carrying a rotor constrained to a nonholonomic distribution, would provide
one. The connectedness is not decorative: proper \emph{discrete} nonabelian
subgroups exist, for instance the lattice
$\Lambda=\mathbb Z\times\mathbb Z\times\tfrac12\mathbb Z$, in which
$(1,0,0)\cdot(0,1,0)=(1,1,\tfrac12)$ and
$(0,1,0)\cdot(1,0,0)=(1,1,-\tfrac12)$ by \eqref{heis_group_law}. What is true,
and what is used here, is the infinitesimal statement: every proper Lie
subalgebra $\mathfrak a\subsetneq\mathfrak h$ is abelian. Indeed, writing
$\pi(aX+bY+cW)=(a,b)$, one has $[v,v']=\omega(\pi v,\pi v')W$ with $\omega$ the
standard area form, so $\dim\pi(\mathfrak a)=2$ would force
$W\notin\mathfrak a$ while $[\mathfrak a,\mathfrak a]=\mathbb RW$, which is
impossible; hence $\dim\pi(\mathfrak a)\leq1$ and $\mathfrak a$ is abelian.
A discrete subgroup has trivial Lie algebra and so cannot revive the
$\operatorname{ad}^*$ term either.
\end{remark}


\section{Singular and abnormal extremals}
\label{singular}

Throughout Sections~\ref{ocpsec}--\ref{drec} the optimal control problem was
assumed regular, so that the Hessian $W_{ab}=\partial^2H_P/\partial u^a\partial
u^b$ was invertible, the control could be eliminated, and $p_1:M_1\rightarrow
T^*P$ was a local diffeomorphism. Under that hypothesis the primary constraint
submanifold is \emph{locally} symplectomorphic to $T^*P$, on each branch of the
stationary control, and the Pontryagin Maximum Principle is there an ordinary
Hamiltonian system; the Dirac formulation is then a convenient language, but not
more than that. Regularity by itself gives no global statement: several
nondegenerate solutions of $\partial H_P/\partial u=0$ may lie over the same
$(x,p)$, as for $P=\mathbb R^2$, $\Gamma(x,u)=u$ and
$L(u)=e^{u_1}\cos u_2$, where $\det\operatorname{Hess}L=-e^{2u_1}$ never
vanishes while $p_1(x,u)=(x,\nabla L(u))$ is $2\pi$-periodic in $u_2$ and misses
$p=0$.

Dropping that hypothesis is not generality for its own sake. Already in
\cite{Hs92} one finds a bracket-generating rank-two system, the Hilbert--Cartan
system on $\mathbb R^5$, carrying a five-parameter family of integral curves that
admit no compactly supported variations at all; such curves are trivially
stationary for every functional, and are in general \emph{not} solutions of the
associated Euler--Lagrange system. They are counterexamples to the assertion,
still occasionally made, that a minimizing curve in a sub-Riemannian manifold
satisfies the geodesic equations \cite[\S2]{Hs92}. The regular theory of
Sections~\ref{ocpsec}--\ref{drec} cannot see them.

In this section we drop the regularity assumption. Two things change. First, the
control can no longer be eliminated, and the constraint algorithm of Gotay,
Nester and Hinds \cite{Go79} must be run; for optimal control this was carried
out in the presymplectic setting by Barbero-Li\~n\'an and Mu\~noz-Lecanda
\cite{Ba08}. Second, and more importantly for the present paper, the reduction
of Section~\ref{drec} cannot be simplified as in Section~\ref{ldrpmp}: the
passage from $[D_\Omega]_G$ to $[D_P]_G$ used the regularity of the problem, so
in the singular case one must work with the reduced Dirac structure
$[D_\Omega]_G$ on $M/G$ itself. The object constructed in Section~\ref{drec} and
then set aside is precisely the one that is needed here.

The main result of this section is that reduction commutes with the constraint
algorithm (Theorem~\ref{reduction_commutes}): the algorithm may be run on the
reduced space, which is smaller, and the answer is the quotient of the answer
obtained upstairs. We then show that abnormal extremals occupy a distinguished
leaf of the reduced Pontryagin bundle, and that for control-linear systems the
first secondary constraint is exactly the vanishing of the curvature pairing
$\langle\mu,F_A\rangle$ appearing in the reduced costate equation \eqref{RHP}
(Theorem~\ref{abnormal_curvature}). The term that distinguishes the reduced
Pontryagin Maximum Principle from its unreduced counterpart is therefore also
the term that obstructs abnormality.


\subsection{Dirac structures induced by a distribution}
\label{induced_dirac}

We first recall the construction, due in this form to Yoshimura and Marsden
\cite[\S7]{Yo07a}, by which a two-form and a distribution together determine a
Dirac structure in the vector-bundle sense of
Definition~\ref{Dirac_vector_bundle}. It is this construction, rather than the
graph of a nondegenerate form, that will carry the constraints.

\begin{definition}
\label{def_induced_dirac}
Let $N$ be a manifold, $\Omega$ a two-form on $N$ and $\Delta\subset TN$ a
regular distribution. The \emph{Dirac structure induced by $(\Omega,\Delta)$} is
\begin{equation}
\label{induced_dirac_structure}
D_{\Omega,\Delta}(n)
=
\Bigl\{
(v,\alpha)\in T_nN\oplus T^*_nN
\ \Big|\
v\in\Delta(n),\quad
\alpha|_{\Delta(n)}=\Omega(n)(v,\cdot)\big|_{\Delta(n)}
\Bigr\}.
\end{equation}
\end{definition}

Equivalently, $(v,\alpha)\in D_{\Omega,\Delta}$ if and only if $v\in\Delta$ and
\begin{equation}
\label{annihilator_form}
\alpha-i_v\Omega\in\Delta^\circ ,
\end{equation}
which is the form in which the constraints will appear: the failure of the
unconstrained equation $\alpha=i_v\Omega$ is required to lie in the annihilator
of $\Delta$, exactly as a reaction force does in nonholonomic mechanics.

\begin{proposition}
\label{induced_is_dirac}
$D_{\Omega,\Delta}$ is a Dirac structure on the vector bundle $TN$, in the
sense of Definition~\ref{Dirac_vector_bundle}, for every two-form $\Omega$,
possibly degenerate, and every regular distribution $\Delta$. Moreover
\begin{equation}
\label{kernel_of_induced}
G_{D_{\Omega,\Delta}}(n)
:=
\bigl\{v\in T_nN\ :\ (v,0)\in D_{\Omega,\Delta}(n)\bigr\}
=
\Delta(n)\cap\Delta(n)^{\perp\Omega},
\end{equation}
where $\Delta^{\perp\Omega}$ denotes the $\Omega$-orthogonal complement.
Taking $\Delta=TN$ recovers the graph $D_\Omega$, with
$G_{D_\Omega}=\ker\Omega$.
\end{proposition}

\dem
Isotropy: if $(v,\alpha)$ and $(w,\beta)$ lie in $D_{\Omega,\Delta}$ then
$w\in\Delta$ and $v\in\Delta$, so
\[
\langle\!\langle(v,\alpha),(w,\beta)\rangle\!\rangle
=
\langle\alpha,w\rangle+\langle\beta,v\rangle
=
\Omega(v,w)+\Omega(w,v)
=
0 .
\]
Maximality: for fixed $v\in\Delta$ the admissible $\alpha$ form an affine space
over $\Delta^\circ$, so
$\dim D_{\Omega,\Delta}=\dim\Delta+\dim\Delta^\circ=\dim N$,
which is the rank of a maximally isotropic subbundle of $TN\oplus T^*N$.
The kernel: $(v,0)\in D_{\Omega,\Delta}$ means $v\in\Delta$ and
$\Omega(v,\cdot)|_\Delta=0$, that is, $v\in\Delta\cap\Delta^{\perp\Omega}$.
\findem

Only the pointwise linear algebra of $D_{\Omega,\Delta}$ enters the constraint
algorithm of the next subsection, so no integrability hypothesis will be
needed there; a routine computation with the Courant bracket shows that
$D_{\Omega,\Delta}$ is integrable precisely when $\Delta$ is involutive and
$d\Omega$ vanishes on $\Delta$, but this will not be used. We shall use Definition~\ref{def_induced_dirac} with $\Delta=TN$ for a
submanifold $N\subset M$, in which case \eqref{annihilator_form} reads
$\alpha-i_v\Omega\in(TN)^\circ$ along $N$.

\begin{remark}
\label{restriction_comparison}
For a general Dirac structure $D$ on $M$ and a submanifold $N$, the restriction
$D_N(x)=\sigma(x)\bigl(D(x)\cap(T_xN\times T^*_xM)\bigr)$, with $\sigma(x)(v,v^*)=(v,v^*|_{T_xN})$, defined by Courant \cite{Co90} is a Dirac
structure on $N$ provided $\dim\bigl(D(x)\cap(T_xN\times T^*_xM)\bigr)$ is
constant along $N$ \cite[Prop.~2]{Bl04}. For $D=D_{\Omega,\Delta}$ that dimension is $\dim(\Delta\cap TN)+\operatorname{codim}\Delta$, which is constant exactly when $\Delta\cap TN$ has constant rank (constancy of the ranks of $\Delta$ and of $TN$ separately does not suffice: on $\mathbb R^2$, $\Delta=\operatorname{span}\{\partial_x\}$ meets the tangent bundle of the parabola $y=x^2$ in rank one at the origin and rank zero elsewhere), and is automatic for $\Delta=TM$, in which case $\sigma(D_{\Omega,TN})$ is the graph of $i_N^*\Omega$. We work
with $D_{\Omega,\Delta}$ rather than with $\sigma(D_{\Omega,\Delta})$ because the
annihilator in \eqref{annihilator_form}, which carries the constraint forces,
remains visible.
\end{remark}


\subsection{The constraint algorithm in Dirac form}
\label{algorithm}

Let $D$ be a Dirac structure on a manifold $N$ and $H\in C^\infty(N)$. The
implicit Hamiltonian system determined by the pair is $(X,dH)\in D$. Writing, as
in \eqref{kernel_of_induced},
\[
G_D(n)=\{v\in T_nN\ :\ (v,0)\in D(n)\},
\qquad
\rho^*(D)(n)=\{\alpha\in T^*_nN\ :\ \exists v,\ (v,\alpha)\in D(n)\},
\]
both understood \emph{pointwise}, maximal isotropy of $D(n)$ in
$T_nN\oplus T^*_nN$ gives $\rho^*(D)(n)=G_D(n)^\circ$, so that the system admits
a solution at $n$ if and only if $dH(n)\in G_D(n)^\circ$. This is the intrinsic
form of the primary constraint.

\begin{remark}
\label{pointwise_caveat}
The pointwise reading is essential. Blankenstein and Ratiu
\cite[Rem.~3]{Bl04} observe that the distribution generated by the local sections
$X$ with $(X,0)\in D$ may be strictly smaller than the pointwise subspace
$G_D(n)$, and that the two agree precisely when $\rho^*(D)$ has constant
dimension. Their example is $N=\mathbb R^2$ with
$\Omega=\|x\|^2\,dx_1\wedge dx_2$ and $D=D_\Omega$: the sections generate the
zero distribution, whereas $G_D(0)=T_0N$. It is the pointwise subspace that
governs solvability at a point, and it is therefore the pointwise subspace that
enters \eqref{algorithm_sequence}; the price is that $G_D$ need not be a smooth
distribution, which is one source of the regularity hypotheses in
Definition~\ref{def_algorithm}.
\end{remark}

\begin{definition}
\label{def_algorithm}
Let $\Omega$ be a two-form on $N$, let $D=D_\Omega$ be its graph, and let
$H\in C^\infty(N)$. The \emph{constraint algorithm} for $(N,D_\Omega,H)$ is the
sequence
\begin{equation}
\label{algorithm_sequence}
N_0=N,
\qquad
N_{k+1}
=
\bigl\{n\in N_k\ :\ dH(n)\big|_{G_{D_k}(n)}=0\bigr\},
\qquad
D_k=D_{\Omega,TN_k},
\end{equation}
where $D_k$ is the structure induced on $N$ along $N_k$ in the sense of
Definition~\ref{def_induced_dirac}; by Proposition~\ref{induced_is_dirac} it is
a Dirac structure in the vector-bundle sense, which is all that the pointwise
conditions of \eqref{algorithm_sequence} require. We assume throughout that
every $N_k$ is an embedded submanifold and that the sequence stabilizes,
$N_{k+1}=N_k=:N_\infty$ for $k$ large.
\end{definition}

Two features of this definition should be noted, because they are what make it
usable and what limit its scope. First, it is stated for a \emph{presymplectic
graph} $D=D_\Omega$ and not for an arbitrary Dirac structure $D$ on $N$: the
recipe $D_k=D_{\Omega,TN_k}$ of Definition~\ref{def_induced_dirac} takes a
two-form and a distribution as its data, and for a general $D$ there is no such
$\Omega$. For a general Dirac structure the corresponding object is the clean
backward image $i_k^!D$ along the inclusion $i_k:N_k\hookrightarrow N$, which
requires the constant-rank hypotheses of Remark~\ref{restriction_comparison};
we do not need that generality, since every structure to which the algorithm is
applied in this paper is of the form $D_\Omega$. Second, $D_k$ is an
\emph{ambient} object: it is a maximally isotropic subbundle of
$(TN\oplus T^*N)|_{N_k}$, whose covectors are still covectors of $N$, and not
an intrinsic Dirac structure on $N_k$. This is deliberate, and is what keeps the
constraint forces visible in the annihilator, as explained in
Remark~\ref{restriction_comparison}; the intrinsic structure on $N_k$ is its
image $\sigma(D_k)$, the graph of $i_k^*\Omega$.

Applied to the Pontryagin bundle this reproduces the constraints of
Section~\ref{ocpsec} with no regularity hypothesis.

\begin{proposition}
\label{primary_from_dirac}
Let $M=T^*P\times_PC$ carry $\Omega=\mathrm{pr}_P^{\,*}\omega_P$ and the
Pontryagin Hamiltonian $H_P$. Then $G_{D_\Omega}=\ker\Omega=K$ is the vertical
bundle of $\mathrm{pr}_P$, and the first step of \eqref{algorithm_sequence} is
\[
M_1
=
\left\{
(x,p,u)\in M
\ :\
\frac{\partial H_P}{\partial u^a}=0,\ a=1,\dots,r
\right\},
\]
which is \eqref{constr}. At the next step,
$G_{D_1}=TM_1\cap(TM_1)^{\perp\Omega}=\ker\Omega_1$ and
\[
M_2=\bigl\{m\in M_1\ :\ dH_P(m)|_{\ker\Omega_1(m)}=0\bigr\},
\]
which is the Gotay--Nester secondary constraint. If the problem is regular then
$\ker\Omega_1=0$ and the algorithm stops at $M_1$.
\end{proposition}

\dem
The first assertion is the computation of $\ker\Omega$ in Section~\ref{ocpsec},
and $dH_P|_K=0$ is $\partial H_P/\partial u^a=0$. The second follows from
\eqref{kernel_of_induced} with $\Delta=TM_1$, since
$\Omega_1=i_1^*\Omega$ and hence
$\ker\Omega_1=TM_1\cap(TM_1)^{\perp\Omega}$. The last assertion is
Proposition~\ref{regular_points}(iii).
\findem

Along $M_k$ the equation $(X,dH_P)\in D_k$ reads, by \eqref{annihilator_form},
\begin{equation}
\label{reaction_form}
X\in TM_k,
\qquad
i_X\Omega-dH_P\in(TM_k)^\circ ,
\end{equation}
so the constraints enter twice, as a restriction on the admissible velocities
and as a relaxation of the covector equation into the annihilator. In the
regular case $(TM_1)^\circ$ contributes nothing beyond the elimination of the
control, and one recovers the formulation of Section~\ref{ocpsec}.

\begin{remark}
The presymplectic algorithm for the Pontryagin Maximum Principle, including the
identification of the secondary constraints with the classical Goh and
generalized Legendre--Clebsch conditions, is due to Barbero-Li\~n\'an and
Mu\~noz-Lecanda \cite{Ba08}; what is new here is the Dirac formulation
\eqref{reaction_form} of each stage and, in the next subsection, its behavior
under symmetry reduction.
\end{remark}


\subsection{Reduction commutes with the constraint algorithm}
\label{commutation}

We now come to the main point. The reduced object $[D_\Omega]_G$ is a Dirac
subbundle of the Courant algebroid $E=(TM\oplus T^*M)/G$ over $M/G$, not a
Dirac structure on $T(M/G)\oplus T^*(M/G)$, so Definition~\ref{def_algorithm}
does not apply verbatim to the reduced data $(M/G,[D_\Omega]_G,H_G)$: the
kernel
\[
G_{[D]_G}([m])
=
\bigl\{[v]\in(TM/G)_{[m]}\ :\ ([v],0)\in[D]_G([m])\bigr\}
\]
is a subspace of the fiber $(TM/G)_{[m]}$, whereas $dH_G([m])$ is a covector
on $M/G$, and the two cannot be paired directly. The pairing is supplied by
the anchor map of $E$, which is the anchor of the Atiyah algebroid $TM/G$
(cf.\ Section~\ref{algebroid}, where the case $TP/G$ is spelled out),
\[
a\colon TM/G\longrightarrow T(M/G),
\qquad
a([v])=T\pi_M(v),
\]
well defined by the equivariance of $T\pi_M$. The constraint algorithm for the
reduced data is then the sequence
\begin{equation}
\label{reduced_algorithm_sequence}
(M/G)_0=M/G,
\qquad
(M/G)_{k+1}
=
\bigl\{[m]\in(M/G)_k\ :\
dH_G([m])\big|_{a(\overline G_k([m]))}=0\bigr\},
\end{equation}
where $\overline G_k=G_{\overline D_k}$ and $\overline D_k$ is the structure
induced along $(M/G)_k$, read fiberwise in $E$ as in
Definition~\ref{def_induced_dirac}: $\overline D_k
=D_{[\Omega]_G,\,\Delta_k}$ with two-form
$[\Omega]_G([v],[w])=\Omega(v,w)$ on $TM/G$, well defined by the invariance of
$\Omega$, and distribution $\Delta_k=a^{-1}\bigl(T(M/G)_k\bigr)\subset TM/G$.
Write $(M/G)_k$ for the resulting sequence.

\begin{lemma}
\label{kernel_descends}
Let $G$ act freely and properly on $M$ and let $D\subset TM\oplus T^*M$ be a
$G$-invariant Dirac structure. Then
\[
G_{[D]_G}=\bigl[G_D\bigr]_G .
\]
\end{lemma}

\dem
Both sides are understood pointwise, as in Remark~\ref{pointwise_caveat}. The
action of $G$ on $TM$ and on $T^*M$ is by vector bundle automorphisms, so
$[v]=0$ if and only if $v=0$, and likewise for covectors. Since the action is
free, a class $([v],[\alpha])$ at $[m]$ has, for each choice of representative
$m$, a unique representative pair $(v,\alpha)\in T_mM\oplus T^*_mM$, and
$([v],[\alpha])\in[D]_G$ if and only if $(v,\alpha)\in D(m)$. Hence
$([v],0)\in[D]_G$ if and only if $(v,0)\in D$, that is, $v\in G_D$.
\findem

\begin{theorem}
\label{reduction_commutes}
Let $(P,C,\Gamma,L)$ be a $G$-invariant optimal control problem, not assumed
regular, with $G$ acting freely and properly. Assume that the constraint
algorithm for $(M,D_\Omega,H_P)$ produces embedded submanifolds $M_k$ and
stabilizes. Then

\begin{enumerate}
\item[(i)] every $M_k$ is $G$-invariant, and the induced Dirac structures $D_k$
are $G$-invariant;
\item[(ii)] $(M/G)_k=M_k/G$ for every $k$, and $\overline D_k=[D_k]_G$: the
Dirac structure carried by the $k$-th stage of the reduced algorithm is the
reduction of $D_k$;
\item[(iii)] consequently $(M/G)_\infty=M_\infty/G$, and the reduced dynamics on
$(M/G)_\infty$ is the reduction of the dynamics on $M_\infty$.
\end{enumerate}
\end{theorem}

\dem
We argue by induction on the statement: $M_k$ is $G$-invariant,
$(M/G)_k=M_k/G$, and $\overline D_k=[D_k]_G$. For $k=0$ all three assertions
hold trivially.

Suppose they hold at stage $k$. Then $TM_k$ is a $G$-invariant distribution
and, $\Omega$ being $G$-invariant by Proposition~\ref{invariance_prop}, the
induced structure $D_k=D_{\Omega,TM_k}$ is $G$-invariant; this proves the
second half of (i) at stage $k$. By Lemma~\ref{kernel_descends},
$G_{[D_k]_G}=[G_{D_k}]_G$, so that
\[
a\bigl(\overline G_k([m])\bigr)
=
a\bigl(G_{[D_k]_G}([m])\bigr)
=
T_m\pi_M\bigl(G_{D_k}(m)\bigr).
\]
Since $H_P=\pi_M^*H_G$, one has
$\langle dH_P(m),v\rangle=\langle dH_G([m]),T_m\pi_M(v)\rangle$ for every
$v\in T_mM$. Hence, for $m\in M_k$,
\[
dH_G([m])\big|_{a(\overline G_k([m]))}=0
\iff
dH_P(m)\big|_{G_{D_k}(m)}=0 ,
\]
so that $M_{k+1}$ is $G$-invariant and $(M/G)_{k+1}=M_{k+1}/G$. Finally,
$M_{k+1}$ is $G$-invariant and contains the $G$-orbits of its points, so the
vertical spaces lie in $TM_{k+1}$ and
$TM_{k+1}=(T\pi_M)^{-1}\bigl(T(M_{k+1}/G)\bigr)$ along $M_{k+1}$; hence
$\Delta_{k+1}=a^{-1}\bigl(T(M/G)_{k+1}\bigr)=TM_{k+1}/G$, and the two
pointwise conditions of \eqref{induced_dirac_structure} pass to the quotient
classwise, exactly as in the proof of Lemma~\ref{kernel_descends}, giving
$\overline D_{k+1}=[D_{k+1}]_G$. This is (i) and (ii) at stage $k+1$.

For (iii), stabilization upstairs gives stabilization downstairs by (ii), and
the reduced dynamics is characterized by
$([X]_G,[dH_P]_G)\in[D_\infty]_G$, with $[dH_P]_G=a^*dH_G$ the section of
$T^*M/G$ of \eqref{Dirac_reduced_equations}, which by (ii) and the
$G$-invariance of $D_\infty$ is the reduction of $(X,dH_P)\in D_\infty$.
\findem

\begin{remark}
Theorem~\ref{reduction_commutes} is of practical as well as structural interest:
the algorithm may be run on $M/G$, whose dimension is smaller by $\dim G$, and
the constraints obtained upstairs are then recovered by pulling back. In the
regular case the theorem is vacuous, both algorithms stopping at the first
stage; it has content precisely when the problem is singular.
\end{remark}

\begin{remark}
The proof uses only that $[D_\Omega]_G$ is a Dirac subbundle of the Courant
algebroid $E$, the anchor map $a$ of $E$, and that $H_G$ is a function on
$M/G$. It therefore does not
require the identification $T^*P/G\cong T^*(P/G)\oplus\widetilde{\mathfrak
g}^{\,*}$, and in particular no principal connection is chosen. A connection is
needed only to write the resulting equations in the explicit form \eqref{RHP}.
\end{remark}

\begin{remark}
\label{BR_comparison}
It is worth contrasting the reduced object used here with the one of
Blankenstein and van der Schaft \cite{Bl01} and Blankenstein and Ratiu
\cite{Bl04}. Given a symmetry group acting regularly, they project $D$ to an
honest Dirac structure $\widehat D$ \emph{on} $M/G$, whose local sections are the
pairs $(\widehat X,\widehat\alpha)$ admitting a lift $(X,\alpha)\in D$ with
$\alpha=\pi^*\widehat\alpha$; that is, they retain only those covectors which are
pullbacks from the quotient \cite[Prop.~5]{Bl04}. This is the appropriate object
for a Lie--Poisson type description of the reduced dynamics, and it is a Dirac
structure in the ordinary sense.

It is not, however, the object required here. Discarding the covectors that are
not basic discards precisely the $\widetilde{\mathfrak g}^{\,*}$ component, and
with it the variables $\xi$ and $\mu$ through which the reduced equations
\eqref{RHP} are expressed; no reconstruction equation survives the passage. The
quotient $[D_\Omega]_G$ used in Section~\ref{drec}, a Dirac subbundle of the
Courant algebroid $E=(TM\oplus T^*M)/G$ rather than a Dirac structure on $M/G$,
retains that information, which is why it is the one carried through the
constraint algorithm above. The two constructions are complementary rather than
competing: $\widehat D$ answers what the reduced dynamics is on $M/G$, and
$[D_\Omega]_G$ answers what it is on $M/G$ \emph{together with} the data needed
to reconstruct the unreduced extremal.

A further caveat from the same source deserves mention. Restriction of an
implicit Hamiltonian system to an invariant submanifold carries solutions in one
direction only, there being in general no bijection between the solutions of
$(M,D,H)$ contained in $N$ and those of $(N,D_N,H_N)$ \cite[Prop.~4]{Bl04}; and
not every solution of a reduced implicit Hamiltonian system need be the
projection of a solution upstairs, although every solution is when the
constraints have index one \cite[Prop.~6]{Bl04}. Theorem~\ref{reduction_commutes}
is a statement about the constraint manifolds and the Dirac structures carried on
them, and is unaffected by these subtleties; the correspondence of individual
solutions requires them.
\end{remark}


\subsection{The abnormal multiplier as a momentum map}
\label{cost_extension}

The Pontryagin Hamiltonian $H_P=\langle p,\Gamma\rangle-L$ used throughout
corresponds to the normalization $p_0=-1$ of the cost multiplier, and therefore
describes normal extremals only. Abnormal extremals are recovered by adjoining
the running cost to the state, a device introduced in this context by Ohsawa
\cite{Oh15} and developed by J\'o\'zwikowski and Respondek \cite{JR16}, who work
on the extended space $Q\times\mathbb R$ of configurations and costs and pass to
the projectivization of $T^*(Q\times\mathbb R)$, thereby treating normal and
abnormal extremals uniformly by contact-geometric means. What we add is only the
observation that, read through the reduction theory of this paper, the extension
is itself an instance of that theory: translation in the cost variable is a
symmetry whose momentum map is the multiplier $p_0$, so that the normalization is
a choice of level set rather than a convention.

Let
\[
\widehat P=P\times\mathbb R,
\qquad
\widehat C=C\times\mathbb R\longrightarrow\widehat P,
\]
with the extra coordinate $y$ recording the accumulated cost, and let
\begin{equation}
\label{cost_extended_dynamics}
\widehat\Gamma(x,y,u)=\bigl(\Gamma(x,u),\,L(x,u)\bigr),
\qquad
\widehat L=0 .
\end{equation}
Minimizing the terminal objective $\Phi(x,y)=y$ evaluated at $t=T$, subject to
\eqref{cost_extended_dynamics} with $y(0)=0$ and $x(0)=x_0$, $x(T)=x_T$, is the
Mayer form of the original problem. Its Pontryagin Hamiltonian, with costate
$(p,p_0)$, is
\begin{equation}
\label{extended_hamiltonian}
\widehat H(x,y,p,p_0,u)=\langle p,\Gamma(x,u)\rangle+p_0L(x,u),
\end{equation}
which does not depend on $y$.

\begin{proposition}
\label{p0_momentum}
The group $\widehat G=\mathbb R$ acting on $\widehat P$ by $y\mapsto y+s$, and
trivially on controls, is a symmetry group of the extended optimal control
problem in the sense of Definition~\ref{ocp_symmetry}. Its momentum map is
$J=p_0$, the bundle $\widehat P\rightarrow\widehat P/\widehat G\cong P$ is
trivial with flat product connection, and $\widehat{\mathfrak g}$ is abelian.
Consequently the reduced equations \eqref{RHP} for the extended problem read
\begin{equation}
\label{extended_reduced}
\dot x=\frac{\partial\widehat H}{\partial p},
\qquad
\dot p=-\frac{\partial\widehat H}{\partial x},
\qquad
\dot p_0=0,
\qquad
\xi=\frac{\partial\widehat H}{\partial p_0}=L,
\qquad
0=\frac{\partial\widehat H}{\partial u},
\end{equation}
and the framework of Sections~\ref{ocpsec}--\ref{drec} is the restriction of
\eqref{extended_reduced} to the leaf $p_0=-1$.
\end{proposition}

One point about the extension deserves to be made explicitly, because it
concerns the boundary conditions and not the interior equations. Translation in
$y$ is a symmetry of the dynamics \eqref{cost_extended_dynamics} and of
$\widehat L=0$, hence of the interior variational problem; it is \emph{not} a
symmetry of the Mayer boundary problem, whose terminal objective $\Phi=y$ and
whose initial condition $y(0)=0$ both single out a value of $y$. That is as it
should be: the PMP transversality condition for the terminal cost $\Phi$,
namely $p_0(T)=-\partial\Phi/\partial y=-1$ in the normal case, is precisely
what selects the leaf $p_0=-1$ and fixes the sign and normalization of the
multiplier, and no transversality condition is imposed on $p(T)$ because $x(T)$
is prescribed. The abnormal case $p_0=0$ is the alternative permitted when the
nontriviality requirement $(p,p_0)\not\equiv0$ is met by $p$ alone. Conservation
of $p_0$ along the flow, which is the content of \eqref{extended_reduced}, is
a consequence of the interior symmetry and is unaffected by this remark.

\dem
Invariance is immediate from \eqref{cost_extended_dynamics}, since neither
$\Gamma$ nor $L$ depends on $y$ and $\widehat L=0$. The fundamental vector field
of $\zeta\in\widehat{\mathfrak g}=\mathbb R$ is $\zeta\,\partial_y$, so
$\langle J,\zeta\rangle=\langle(p,p_0),\zeta\partial_y\rangle=\zeta p_0$. Both
the curvature and the coadjoint terms in \eqref{RHP} vanish, the first because
the product connection on $P\times\mathbb R$ is flat and the second because
$\widehat{\mathfrak g}$ is abelian, which gives \eqref{extended_reduced}.
Restricting to $p_0=-1$ turns \eqref{extended_hamiltonian} into
$\langle p,\Gamma\rangle-L=H_P$.
\findem

Thus reducing the Mayer problem by the translation symmetry of the cost variable
is what converts it to Lagrange form; the reconstruction variable $\xi=L$ is the
running cost, and reconstruction recovers the accumulated cost. Because the
reduction used here is of Poisson rather than symplectic type, all levels of
$p_0$ are described simultaneously, the level appearing as a coordinate on the
reduced space. In particular:

\begin{definition}
\label{def_abnormal}
The \emph{abnormal extremals} of $(P,C,\Gamma,L)$ are the solutions of
\eqref{extended_reduced} lying on the leaf $p_0=0$, with $p\not\equiv0$.
\end{definition}

On that leaf $\widehat H=\langle p,\Gamma\rangle$ carries no cost term, the
Hessian in the controls is $\langle p,\partial^2\Gamma/\partial u^2\rangle$, and
this vanishes identically whenever $\Gamma$ is affine in $u$. \emph{For
control-affine systems}, abnormal extremals are therefore always singular in the
sense of Section~\ref{ocpsec}, and it is exactly here that the algorithm of
Section~\ref{algorithm} is required. The restriction is necessary: for a
$\Gamma$ depending nonlinearly on $u$ an abnormal extremal may be regular. With
$P=\mathbb R$ and $\Gamma(x,u)=(1+u^2)\partial_x$ one has
$\widehat H=p(1+u^2)$ on $p_0=0$, so $\widehat H_u=2pu$ vanishes at $u=0$ while
$\widehat H_{uu}=2p\neq0$; the curve $\dot x=1$ with $p$ a nonzero constant is
an abnormal extremal through a regular point of $M_1$. All the systems treated
below, namely \eqref{horizontal_system}, are control-linear, so the hypothesis
is in force throughout the rest of this section.


\subsection{Control-linear systems: abnormality and curvature}
\label{abnormal_curvature_sec}

We now specialize to the situation in which the admissible velocities form a
distribution, and show that the constraint algorithm is governed by the
curvature of the connection.

Let $\rho:P\rightarrow B=P/G$ be a principal $G$-bundle with connection $A$ and
curvature $F_A$, let
\[
\mathcal D=\ker A\subset TP
\]
be the horizontal distribution, and consider the driftless control-linear system
obtained by taking $C=\rho^*\widetilde U$ with $U\cong\mathbb R^{\dim B}$ and
\begin{equation}
\label{horizontal_system}
\Gamma(x,u)=\sum_a u^a X_a^h(x),
\end{equation}
where $X_a^h$ are the horizontal lifts of a local frame $X_a$ of $TB$. The
sub-Riemannian problem of Section~\ref{heisenberg} is of this form, and so is
the Martinet problem below.

\begin{convention}
\label{immersed_convention}
Unless the contrary is stated, admissible curves of \eqref{horizontal_system}
are assumed \emph{immersed}, that is, $u(t)\neq0$ for every $t$; whenever a
metric is present they are assumed parametrized by arclength. Constant curves,
for which every covector in $\mathcal D^\circ$ is trivially an abnormal
multiplier, are thereby excluded, and so are reparametrizations that stop. The
convention is used at three places below --- the constraint chain of
Section~\ref{martinet}, Corollary~\ref{martinet_elliptic} and
Proposition~\ref{martinet_nonstrict} --- and we recall it at each. Theorem
\ref{abnormal_curvature}, by contrast, is deliberately stated \emph{without} it,
so that the reader can see exactly where the hypothesis is needed.
\end{convention}

\begin{proposition}
\label{abnormal_leaf}
For the system \eqref{horizontal_system} the abnormal primary constraint
manifold is the annihilator of the horizontal distribution,
\[
M_1\big|_{p_0=0}
=
\mathcal D^\circ\times_P C ,
\]
and under the identification \eqref{cotangent_quotient} determined by $A$ this
is precisely the leaf
\[
p_z=0
\]
of the reduced Pontryagin bundle: the base costate of an abnormal extremal
vanishes identically, and the extremal is described by the coadjoint variable
$\mu$ alone.
\end{proposition}

\dem
On $p_0=0$ we have $\widehat H=\sum_au^a\langle p,X_a^h\rangle$, so the primary
constraints $\partial\widehat H/\partial u^a=0$ read $\langle p,X_a^h\rangle=0$
for all $a$, that is, $p\in\mathcal D^\circ$. In the splitting
\eqref{cotangent_quotient} the component $p_z$ is the restriction of $p$ to the
horizontal subspace, so $p\in\mathcal D^\circ$ is exactly $p_z=0$, and $\mu$,
which is determined by $p=A_x^*\mu$, that is, by
$\langle p,v\rangle=\langle\mu,A_x(v)\rangle$ for $v\in T_xP$, is unconstrained.
(Note the direction: $A_x:T_xP\rightarrow\mathfrak g$, so
$A_x^*:\mathfrak g^*\rightarrow T_x^*P$, and it is $p$ that is the pullback of
$\mu$, not the reverse. Equivalently $\mu=J(p)$ is the value of the momentum
map.) The identification uses $\mathcal D=\ker A$, so that a covector
annihilating $\mathcal D$ is determined by its values on the vertical
subspace.
\findem

\begin{lemma}
\label{bracket_curvature}
For horizontal lifts, $A\bigl([X^h,Y^h]\bigr)=-F_A(X,Y)$. Consequently, if
$p\in\mathcal D^\circ$ corresponds to $\mu$ under \eqref{cotangent_quotient},
\begin{equation}
\label{goh_curvature}
\bigl\langle p,[X^h,Y^h]\bigr\rangle=-\bigl\langle\mu,F_A(X,Y)\bigr\rangle .
\end{equation}
\end{lemma}

\dem
Since $A(X^h)=A(Y^h)=0$,
\[
F_A(X^h,Y^h)
=
dA(X^h,Y^h)+[A(X^h),A(Y^h)]
=
-A\bigl([X^h,Y^h]\bigr),
\]
and $F_A$ is basic, so $F_A(X^h,Y^h)=F_A(X,Y)$. For $p\in\mathcal D^\circ$ one
has $\langle p,v\rangle=\langle\mu,A(v)\rangle$ for every $v$, whence
\eqref{goh_curvature}.
\findem

\begin{theorem}
\label{abnormal_curvature}
For the system \eqref{horizontal_system}, write
\begin{equation}
\label{curvature_form}
C_\mu(z)
=
\bigl\langle\mu,F_A(\,\cdot\,,\,\cdot\,)\bigr\rangle,
\end{equation}
a skew-symmetric bilinear form on $T_zB$. The first secondary constraint of the
algorithm on the abnormal leaf is that the control take values in the kernel of
$C_\mu$:
\begin{equation}
\label{goh_condition}
u(t)\in\ker C_{\mu(t)}\bigl(z(t)\bigr),
\qquad\text{that is,}\qquad
C_{\mu(t)}\bigl(u(t),\cdot\bigr)=0 .
\end{equation}
At any time at which $u(t)\neq0$, and only then, \eqref{goh_condition} forces
$C_{\mu(t)}$ to be \emph{degenerate}. In particular:
\begin{enumerate}
\item[(i)] if $C_\mu$ is nondegenerate for every $\mu\neq0$ at a point, then no
abnormal extremal with nonvanishing control --- equivalently, with nonvanishing
speed --- passes through that point;
\item[(ii)] if $\dim B=2$ then, along an immersed curve,
\eqref{goh_condition} reduces to
$\langle\mu,F_A\rangle=0$, which is the classical Goh condition; if moreover
$\dim\mathfrak g=1$, every immersed abnormal extremal projects into the locus
$\{F_A=0\}$ where the curvature vanishes.
\end{enumerate}
\end{theorem}

\dem
On the abnormal leaf $\widehat H=\sum_au^ah_a$ with $h_a=\langle p,X^h_a\rangle$,
and $M_1=\{h_a=0\}$ by Proposition~\ref{abnormal_leaf}. Tangency of the dynamics
to $M_1$ requires $\dot h_a=\sum_bu^b\{h_a,h_b\}=0$ for each $a$, and
$\{h_a,h_b\}=-\langle p,[X^h_a,X^h_b]\rangle$, which by
Lemma~\ref{bracket_curvature} equals $\langle\mu,F_A(X_a,X_b)\rangle$. Hence
tangency reads $C_\mu(u,\cdot)=0$, that is, \eqref{goh_condition}; a
\emph{nonvanishing} control satisfying it exists only where $C_\mu$ is
degenerate, which is Assertion (i) in contrapositive form. Nothing is asserted
at instants where $u(t)=0$, and indeed a constant curve with $u\equiv0$ and any
$0\neq p\in\mathcal D^\circ$ satisfies \eqref{extended_reduced} on $p_0=0$
vacuously. For (ii), a skew form on a two-dimensional space is a
multiple of the area form, so it is degenerate if and only if it vanishes; and
if in addition $\mathfrak g=\mathbb R$ then $C_\mu=\mu F_A$ with $\mu\neq0$ on a
nontrivial abnormal extremal, since $\mu=0$ together with $p_z=0$ would give
$p=0$.
\findem

\begin{remark}
Four comments place Theorem~\ref{abnormal_curvature} in context. It concerns abnormal extremals, not
\emph{strictly} abnormal ones in the sense of
Definition~\ref{strictly_abnormal}: an abnormal extremal may carry a normal
lift as well, as in the Martinet problem below, and deciding which abnormal extremals are strict requires more than the curvature and is taken
up in Section~\ref{strictness}. Condition (i) is the classical criterion of \emph{fatness}, that a distribution for
which $C_\mu$ is nondegenerate for all $\mu\neq0$ admits no nontrivial abnormal
extremals \cite{AgrachevSachkov,Mo02}, recovered here as the nondegeneracy of the curvature pairing appearing in \eqref{RHP}. Intrinsic criteria valid for an arbitrary distribution are available, notably that of J\'o\'zwikowski
and Respondek \cite[Thm.~5, Cor.~1]{JR16}, which characterizes the abnormal
sub-Riemannian extremals as the curves along which the smallest flow-invariant distribution containing $\mathcal D$ has rank less than $\dim P$, a condition on the iterated brackets $\operatorname{ad}^k_X(\Gamma(\mathcal D))$ for smooth fields; for $\mathcal D=\ker A$ the iterated brackets are generated by the
curvature and the criterion becomes the degeneracy of $C_\mu$, so the content here is
not a new criterion but the identification of the obstruction with the
curvature term of \eqref{RHP}. Finally, \eqref{goh_condition} is the form taken along the extremal by the
classical Goh condition \cite{AgrachevSachkov,Ba08}, with which it coincides
when $\dim B=2$; in general the Goh condition requires $C_\mu$ to vanish on all of $T_zB$, of which \eqref{goh_condition} retains
the part enforced by tangency, namely the vanishing of the term
$\langle\mu,F_A(\dot z,\cdot)\rangle$ to which the reduced costate
equation \eqref{RHP} reduces on the abnormal leaf, where $p_z\equiv0$ and $h\equiv0$. The curvature term of the reduced Pontryagin Maximum
Principle is thus simultaneously the term that distinguishes it from the
unreduced principle and the term that obstructs abnormality.
\end{remark}

\begin{corollary}
\label{heisenberg_no_abnormal}
The sub-Riemannian problem on the Heisenberg group of Section~\ref{heisenberg}
has no nontrivial abnormal extremals.
\end{corollary}

\dem
There $\mathfrak g=\mathbb R$ and $F_A=-dx\wedge dy$ is nowhere zero, so
Theorem~\ref{abnormal_curvature}(i) applies with empty locus.
\findem


\subsection{Example: the Martinet problem}
\label{martinet}

The complementary example is the Martinet distribution, for which the curvature
degenerates on a hypersurface. Let $P=\mathbb R^3$ with coordinates $(x,y,z)$,
let
\begin{equation}
\label{martinet_connection}
A=dz-\frac{x^2}{2}\,dy,
\qquad
F_A=dA=-\,x\,dx\wedge dy ,
\end{equation}
and let $G=\mathbb R$ act by translation in $z$, so that $B=\mathbb R^2$ and
$\mu=p_z$. The horizontal lifts of $\partial_x,\partial_y$ are
\[
X=\frac{\partial}{\partial x},
\qquad
Y=\frac{\partial}{\partial y}+\frac{x^2}{2}\frac{\partial}{\partial z},
\qquad
[X,Y]=x\frac{\partial}{\partial z},
\]
and $A([X,Y])=x=-F_A(\partial_x,\partial_y)$, in agreement with
Lemma~\ref{bracket_curvature}. The corresponding optimal control problem, with
$\Gamma=u_1X+u_2Y$ and $L=\frac12|u|^2$, is the Martinet sub-Riemannian problem
\cite{Mo02,ABCK97}.

\begin{remark}
\label{martinet_convention}
Conventions differ. We place the Martinet surface at $\{x=0\}$, whereas
\cite{ABCK97} use $\omega=dz-\tfrac12y^2\,dx$ and place it at $\{y=0\}$. The
transposition $\sigma(x,y,z)=(y,x,z)$ satisfies $\sigma^*\omega=A$ and is an
isometry of the flat metric $dx^2+dy^2$, so the two models are the same problem;
being orientation-reversing on the base it changes the sign of $F_A$, which
accounts for the sign discrepancies between the formulas below and those of
\cite{ABCK97}.
\end{remark}

Unlike the Heisenberg case, $F_A$ vanishes on the \emph{Martinet surface}
$\{x=0\}$, and Theorem~\ref{abnormal_curvature} predicts that abnormal
extremals are confined to it. The algorithm confirms this. On the abnormal leaf
$p_0=0$,
\[
h_X=p_x,
\qquad
h_Y=p_y+\frac{x^2}{2}p_z,
\qquad
\widehat H=u_1h_X+u_2h_Y ,
\]
so that
\[
M_1=\{h_X=h_Y=0\},
\qquad
\{h_X,h_Y\}=-x\,p_z ,
\]
and tangency gives $\dot h_X=-u_2\,x\,p_z$, $\dot h_Y=u_1\,x\,p_z$. The
secondary constraint is therefore
\[
M_2=\{\,u_1\,x\,p_z=u_2\,x\,p_z=0\,\}\cap M_1 ,
\]
which is \eqref{goh_condition} for this connection. As a subset of $M$ this is
the union $\{x\,p_z=0\}\cup\{u=0\}$, which is not a submanifold; by
Convention~\ref{immersed_convention} we work on the immersed branch $u\neq0$,
where it is
\[
M_2=\{x\,p_z=0\}\cap M_1 .
\]
On $M_1$ the vanishing of
$p_z$ would force $p_y=0$ and hence $p=0$, so $p_z\neq0$ and
\[
x\equiv0 .
\]
Tangency of $\{x=0\}$ requires $\dot x=u_1=0$, after which
$\dot p_x=-u_2xp_z=0$ and $\dot z=u_2x^2/2=0$, and the algorithm stabilizes. Note
that the control constraint $u_1=0$ is part of the final constraint manifold,
which lives in $M=T^*P\times_PC$ and therefore carries $u$ among its
coordinates. The abnormal extremal is
\[
x\equiv0,
\qquad
p_x=p_y\equiv0,
\qquad
p_z=\text{const}\neq0,
\qquad
z=\text{const},
\]
with $y$ traversed at the arbitrary rate $u_2$: the Martinet line, recovered
here as the final constraint manifold of the algorithm.

Since all the data are invariant under translation in $z$, and the whole chain
was computed from $\mu=p_z$ and $F_A$ alone, the computation may equally be
performed on $M/G$; by Theorem~\ref{reduction_commutes} the two answers agree,
the reduced final constraint manifold being
$\{x=0,\ p_x=p_y=0,\ u_1=0\}\subset M/G$, with $\mu$ a nonzero parameter and
$u_2$ a free nonzero parameter, by Convention~\ref{immersed_convention}.


\subsection{The Martinet normal geodesics}
\label{martinet_normal}

The abnormal leaf occupied us above; we now treat the normal one, $p_0=-1$. This
is the first place in the paper where the curvature term of \eqref{RHP} is
neither zero nor constant, and it is what makes the reduced equations
interesting: they turn out to be those of a pendulum.

Since the controlled vector field is horizontal, $\widetilde\gamma=0$ and the
reduced Pontryagin Hamiltonian is
\begin{equation}
\label{martinet_normal_hamiltonian}
h(z,p_z,\mu,u)=\langle p_z,u\rangle-\tfrac12|u|^2,
\end{equation}
independent of $z$ and of $\mu$, exactly as in Section~\ref{heisenberg}. The
stationarity condition gives $u=p_z$, and \eqref{RHP} becomes, using
$F_A=-x\,dx\wedge dy$,
\begin{equation}
\label{martinet_normal_equations}
\left\{
\begin{aligned}
\dot x&=p_x, & \dot p_x&=-\mu\,x\,p_y,\\
\dot y&=p_y, & \dot p_y&=+\mu\,x\,p_x,
\end{aligned}
\right.
\qquad
\dot\mu=0,
\qquad
\xi=0 .
\end{equation}
The contrast with the Heisenberg equations \eqref{heis_complete_reduced_PMP} is
the whole point: there the momentum $p_z$ rotated at the constant rate $\mu$,
here at the rate $\mu x$. The curvature is no longer a constant multiple of the
area form, and the reduced system ceases to be linear.

\begin{proposition}
\label{martinet_conserved}
Along solutions of \eqref{martinet_normal_equations} both
\[
|p_z|=|u|
\qquad\text{and}\qquad
p_y-\tfrac12\mu x^2
\]
are constant. The first states that normal extremals are parametrized
proportionally to arclength. The second is the unreduced costate $\lambda_y$,
conserved because the Martinet structure is invariant under translation in $y$ as
well as in $z$.
\end{proposition}

\dem
$\frac{d}{dt}(p_x^2+p_y^2)=2p_x(-\mu xp_y)+2p_y(\mu xp_x)=0$, and
$\frac{d}{dt}\bigl(p_y-\tfrac12\mu x^2\bigr)=\mu x p_x-\mu x\dot x=0$. For the
identification, recall from Section~\ref{martinet} that $p_y=\lambda_y+\tfrac12
x^2\lambda_z$ and $\mu=\lambda_z$, whence $p_y-\tfrac12\mu x^2=\lambda_y$.
\findem

The second conserved quantity is worth a comment. In the variables $(z,p_z,\mu)$
adapted to the $\mathbb R_z$-reduction it is not manifest; the term
$\tfrac12\mu x^2$ is precisely the bookkeeping of the splitting
\eqref{cotangent_quotient}. A symmetry that has not been quotiented is present in
the reduced equations, but only in a form that the connection makes visible.

\begin{theorem}
\label{martinet_pendulum}
Write $p_z=r(\cos\theta,\sin\theta)$ with $r=|u|$ constant. Then along solutions
of \eqref{martinet_normal_equations}
\begin{equation}
\label{martinet_theta}
\dot\theta=\mu\,x,
\qquad\text{and hence}\qquad
\ddot\theta=\mu r\cos\theta .
\end{equation}
Equivalently, setting $\varphi=\theta-\pi/2$ and $\omega^2=\mu r$,
\begin{equation}
\label{martinet_pendulum_eq}
\ddot\varphi=-\omega^2\sin\varphi ,
\end{equation}
the equation of a simple pendulum. The base curve is recovered from
$\dot x=r\cos\theta$, $\dot y=r\sin\theta$, and the fiber coordinate by the
quadrature $\dot z=\tfrac12x^2\dot y$ furnished by $\xi=0$.
\end{theorem}

\dem
Differentiating $p_x=r\cos\theta$ gives $-r\sin\theta\,\dot\theta=-\mu xp_y=-\mu
xr\sin\theta$, so $\dot\theta=\mu x$; the same follows from the equation for
$p_y$. Then $\ddot\theta=\mu\dot x=\mu p_x=\mu r\cos\theta$. The substitution
$\varphi=\theta-\pi/2$ turns $\cos\theta$ into $-\sin\varphi$. The reconstruction
equation is $A(\dot x)=\xi=0$, that is, $\dot z=\tfrac12x^2\dot y$.
\findem

\begin{corollary}
\label{martinet_elliptic}
Let $r>0$ and $\mu>0$. For a librating solution of
\eqref{martinet_pendulum_eq} with modulus $k<1$,
\begin{equation}
\label{martinet_cn}
x(t)=2k\sqrt{\frac{r}{\mu}}\;\operatorname{cn}\bigl(\sqrt{\mu r}\,(t-t_0),\,k\bigr),
\end{equation}
periodic with period $4K(k)/\sqrt{\mu r}$, where $K$ is the complete elliptic
integral of the first kind. One recovers in this way the classical description of
the Martinet geodesics by elliptic functions \textup{\cite{Mo02,ABCK97}}: with
the arclength normalization $r=1$ the amplitude $2k\sqrt{r/\mu}$ and the period
$4K(k)/\sqrt{\mu r}$ are those of \textup{\cite[\S4.5]{ABCK97}}, where the same
solution is written as $-2k\lambda^{-1/2}\operatorname{cn}(K(k)+t\sqrt\lambda,k)$
with $\lambda=p_z$, the difference in sign and phase being the transposition of
Remark~\textup{\ref{martinet_convention}} together with a choice of origin.

The two excluded signs of $\mu$ behave differently. For $\mu<0$ the roles of
the two equilibria of \eqref{martinet_pendulum_eq} are exchanged: the libration
takes place about $\varphi=\pi$, and \eqref{martinet_cn} is replaced by
\[
x(t)=-2k\sqrt{\frac{r}{|\mu|}}\;
\operatorname{cn}\bigl(\sqrt{|\mu|r}\,(t-t_0),\,k\bigr),
\]
of period $4K(k)/\sqrt{|\mu|r}$; the expression \eqref{martinet_cn} itself is
complex there. For $\mu=0$ the pendulum degenerates, since $\omega^2=\mu r=0$
and the substitution $x=\dot\theta/\mu$ used in the proof is unavailable: the
direction $\theta$ of the momentum is constant, and the base curve is a straight
line. The case $r=0$ is excluded by
Convention~\ref{immersed_convention}, since $r=|u|$.
\end{corollary}

\dem
Since $x=\dot\theta/\mu=\dot\varphi/\mu$ and the librating pendulum has
$\dot\varphi=2k\omega\operatorname{cn}(\omega(t-t_0),k)$ with $\omega=\sqrt{\mu
r}$, the amplitude is $2k\omega/\mu=2k\sqrt{r/\mu}$ and the period is that of
$\operatorname{cn}$.
\findem

\subsubsection*{The abnormal line seen from the normal leaf}

The set $\{x=0,\ p_x=0\}$ is invariant under
\eqref{martinet_normal_equations}, since $\dot x=p_x$ and $\dot p_x=-\mu xp_y$
both vanish on it. On that set $p_z=(0,\pm r)$ and the base curve is the
$y$-axis: it is the Martinet line of Section~\ref{martinet}.

\begin{proposition}
\label{martinet_nonstrict}
The abnormal extremal of Section~\ref{martinet}, traversed at constant speed, is
also a normal extremal; the Martinet line is therefore not strictly abnormal as
an unparametrized curve. The parametrization qualifier is needed: the abnormal
computation of Section~\ref{martinet} left the rate $u_2(t)$ along the Martinet
line entirely arbitrary, whereas a normal lift of the quadratic cost has
$u=p_z$ with $|p_z|$ constant by
Proposition~\ref{martinet_conserved}, so an abnormal traversal at
non-constant speed carries no normal lift. This is the arclength convention of
Convention~\ref{immersed_convention} and of
Theorem~\ref{strictness_criterion}. In the pendulum
\eqref{martinet_pendulum_eq} it appears as the two equilibria $\varphi=0$ and
$\varphi=\pi$, corresponding to the line traversed in the directions $+y$ and
$-y$ respectively. For fixed $\mu>0$ the first is a center and the second a
saddle, so the two directions along the Martinet line are not equivalent.
\end{proposition}

\dem
Both equilibria lie in the invariant set $\{x=0,p_x=0\}$ and satisfy
\eqref{martinet_normal_equations} with $p_0=-1$, so the corresponding curves are
normal extremals; by Section~\ref{martinet} the same base curve carries an
abnormal extremal. Linearizing \eqref{martinet_pendulum_eq} gives
$\ddot{\delta\varphi}=-\omega^2\delta\varphi$ at $\varphi=0$ and
$\ddot{\delta\varphi}=+\omega^2\delta\varphi$ at $\varphi=\pi$.
\findem

The first assertion is Lemma~2.7 of \cite{ABCK97}, obtained there by exhibiting
the solution directly; here it falls out of the reduced equations, as the
statement that the invariant set $\{x=0,\ p_x=0\}$ of
\eqref{martinet_normal_equations} projects onto the final constraint manifold of
the algorithm of Section~\ref{martinet}. The reduction is what makes the two
computations comparable: both take place on $M/G$, with $\mu$ a parameter.

The asymmetry is not an artifact: the reflection $y\mapsto-y$ preserves the
Martinet distribution only together with $z\mapsto-z$, hence together with
$\mu\mapsto-\mu$, so for a fixed value of the reduced coadjoint variable the two
directions along the abnormal line are genuinely different. The separatrix of the
saddle is the family of normal geodesics asymptotic to the abnormal line, and it
is this structure that is responsible for the delicate behavior of the Martinet
sphere. Two of its consequences are recorded in \cite{ABCK97}. First, the period
$4K(k)/\sqrt{\mu r}$ of Corollary~\ref{martinet_elliptic} diverges
logarithmically as $k\to1^-$, that is, as the librating geodesics approach the
separatrix cycle joining the saddle to its translate by $2\pi$; the resulting
logarithmic term is what makes the Martinet sphere and the sub-Riemannian
distance fail to be sub-analytic \cite[Thm.~1.2]{ABCK97}. Second, the cut locus
acquires there a singularity like that of a Poincar\'e--Dulac return mapping
computed for a section to a separatrix connecting two saddles \cite{ABCK97}. Both are
statements about the phase portrait of \eqref{martinet_pendulum_eq}, which is to
say about the reduced system; the parameter $\mu$ that organizes the family is
the reduced coadjoint variable.

\begin{remark}
\label{abck_deformation}
The curvature criterion is stable under deformation of the model, and the
deformation of \cite{ABCK97} is instructive about what the criterion does and
does not see. Agrachev, Bonnard, Chyba and Kupka replace the flat structure by
\[
D_\varepsilon=\operatorname{span}\Bigl\{(1+\varepsilon y)\frac{\partial}{\partial
x}+\frac{y^2}{2}\frac{\partial}{\partial z},\ \frac{\partial}{\partial
y}\Bigr\},
\]
the two fields being declared orthonormal. Distribution and cost are again
invariant under translation in $z$, so Definition~\ref{ocp_symmetry} applies,
with
\[
A_\varepsilon=dz-\frac{y^2}{2(1+\varepsilon y)}\,dx,
\qquad
F_{A_\varepsilon}=\frac{y\,(2+\varepsilon y)}{2(1+\varepsilon y)^2}\,dx\wedge dy .
\]
All of this is local in $y$, and the correct domain is
\begin{equation}
\label{abck_domain}
|\varepsilon y|<1 ,
\end{equation}
or any connected neighborhood of $\{y=0\}$ avoiding the hypersurface
$\{1+\varepsilon y=0\}$, on which $A_\varepsilon$, $F_{A_\varepsilon}$ and the
reduced metric \eqref{abck_metric} below are all singular and the first
generator of $D_\varepsilon$ degenerates to $\tfrac12y^2\partial_z$, which is
tangent to the $G$-orbits and therefore no longer horizontal. On
\eqref{abck_domain} the numerator $y(2+\varepsilon y)$ vanishes only at $y=0$,
since the other root $y=-2/\varepsilon$ is excluded, so the curvature vanishes
exactly on $\{y=0\}$ and
Theorem~\ref{abnormal_curvature} confines the abnormal extremals to that
hypersurface; they are the lines $y=0$, $z=z_0$ found in \cite{ABCK97}. (The
domain $|y|<2/|\varepsilon|$, which one might be tempted to read off from the
numerator alone, will not do: it contains $y=-1/\varepsilon$.) For
$\varepsilon=0$ the condition \eqref{abck_domain} is empty and the domain is all
of $\mathbb R^3$, as it should be. What
$\varepsilon$ changes is not the location of the abnormal leaf but its position
relative to the normal one: for $\varepsilon=0$ the abnormal extremal is also
normal (Proposition~\ref{martinet_nonstrict}), while for $\varepsilon\neq0$ it
is strictly abnormal \cite[Thm.~1.3]{ABCK97}. Strictness is a statement about the
intersection of the two leaves, and is invisible to a criterion that sees only
the abnormal one.
\end{remark}

\subsubsection*{Reducing by the full symmetry group}

The Martinet structure is invariant under translations in $y$ as well as in $z$,
so one may instead take $G=\mathbb R^2$, acting by
$(y,z)\mapsto(y+s_1,z+s_2)$. The hypotheses of
Definition~\ref{ocp_symmetry} hold as before, the base is now $B=\mathbb R$ with
coordinate $x$, and the natural connection is
\[
A=(dy,\,dz),
\]
which is \emph{flat}: $F_A=0$. The whole of the dynamics moves into the vertical
part,
\[
\widetilde\Gamma(x,u)=u_1\frac{\partial}{\partial x},
\qquad
\widetilde\gamma(x,u)=\bigl(u_2,\ \tfrac12u_2x^2\bigr),
\]
so that $h=p_xu_1+\mu_1u_2+\tfrac12\mu_2u_2x^2-\tfrac12|u|^2$. Stationarity gives
$u_1=p_x$ and $u_2=\mu_1+\tfrac12\mu_2x^2$, the coadjoint variables are constant
because $\mathfrak g=\mathbb R^2$ is abelian, and \eqref{RHP} reduces to the
single second-order equation
\begin{equation}
\label{martinet_duffing}
\ddot x=-\mu_2\,x\left(\mu_1+\tfrac12\mu_2x^2\right),
\end{equation}
with first integral
\begin{equation}
\label{martinet_energy}
\tfrac12\dot x^2+\tfrac12\left(\mu_1+\tfrac12\mu_2x^2\right)^2=\tfrac12r^2 .
\end{equation}
Here $\mu_1=\lambda_y$ and $\mu_2=\lambda_z$, so that
\eqref{martinet_duffing} is exactly what
\eqref{martinet_normal_equations} gives after eliminating $p_y$ by
Proposition~\ref{martinet_conserved}; the two reductions produce the same base
curve. Equation \eqref{martinet_duffing} is an oscillator in the quartic
potential $V(x)=\tfrac12\mu_1\mu_2x^2+\tfrac18\mu_2^2x^4$, again integrable by
elliptic functions.

\begin{remark}
\label{curvature_vs_vertical}
The two reductions of the same problem are instructive about the structure of
\eqref{RHP}. Reducing by $\mathbb R_z$ leaves a curved connection and a
horizontal controlled field, so that $\widetilde\gamma=0$ and the entire
nonlinearity is carried by the curvature term. Reducing by $\mathbb R^2$ leaves a
flat connection, so that the curvature term disappears and the entire
nonlinearity is carried by $\widetilde\gamma$. The reduced dynamics is the same.
The division of the reduced equations into a curvature term and a vertical term
is therefore a feature of the chosen splitting rather than of the problem, in the
same way that Theorem~\ref{connection_independent} shows the division to be
unaffected by a change of connection for a fixed group. A systematic treatment of
the relation between the two reductions belongs to the theory of reduction by
stages, which we do not develop here.
\end{remark}


\subsection{Strict and non-strict abnormality}
\label{strictness}

Theorem~\ref{abnormal_curvature} locates the abnormal extremals using the
curvature alone, and Proposition~\ref{martinet_nonstrict} showed that in the
Martinet problem the curve so located also carries a normal extremal. That
cannot be read off from the curvature: the abnormal leaf does not depend on the
cost at all, whereas the normal leaf does. This subsection makes the difference
precise, and reduces the question to a linear problem in the fiber
$\widetilde{\mathfrak g}^*$.

\begin{definition}
\label{strictly_abnormal}
An admissible curve is \emph{abnormal} if it is the projection of an extremal
with $p_0=0$ and nonvanishing multiplier, \emph{normal} if it is the projection
of an extremal with $p_0=-1$, and \emph{strictly abnormal} if it is abnormal and
not normal.
\end{definition}

The terminology is that of \cite{ABCK97,MR20}. It is worth a warning that the
adjective \emph{strict} is used by other authors for a different property,
namely that the abnormal multiplier be unique up to scale
\cite[Def.~3.2]{BT01}; the two notions are independent.

\subsubsection*{The two leaves share an adjoint equation}

Fix an admissible reduced curve, that is, a curve $t\mapsto(z(t),\xi(t))$
together with a control $u(t)$ satisfying the first and fourth equations of
\eqref{RHP}. A \emph{lift} of it is a pair $(p_z,\mu)$ along the curve
satisfying the remaining equations of \eqref{RHP} for the given $u$.

\begin{lemma}
\label{common_adjoint}
Suppose the controlled field is linear in the control,
$[\Gamma](z,u)=\sum_au^a\bigl(\widetilde\Gamma_a(z),\widetilde\gamma_a(z)\bigr)$,
and the cost is $L=\tfrac12|u|^2$ with $|\cdot|$ a fixed $G$-invariant norm
on $U$, so that $L$ does not depend on the base point. Then along a fixed
admissible reduced curve the equations for $(p_z,\mu)$ in \eqref{RHP} form one
and the same linear homogeneous system on both leaves; the leaves differ only in
the algebraic condition imposed by $\partial h/\partial u=0$, namely
\[
h_a=0
\quad (p_0=0),
\qquad\text{versus}\qquad
h_a=u^a
\quad (p_0=-1),
\]
where $h_a=\langle p_z,\widetilde\Gamma_a\rangle+\langle\mu,\widetilde\gamma_a\rangle$.
\end{lemma}

\dem
$h=p_0L(u)+\sum_au^ah_a$ with $h_a$ linear in $(p_z,\mu)$ and $L$ independent of
$z$. Hence $\partial h/\partial z=\sum_au^a\partial h_a/\partial z$ carries no
$p_0$, and neither do the covariant derivative terms, which come from the
splitting alone. The condition $\partial h/\partial u^a=0$ reads
$p_0u^a+h_a=0$.
\findem

\begin{proposition}
\label{normal_affine}
For a fixed admissible reduced curve, the set of normal lifts is either empty or
an affine space over the vector space of abnormal lifts. In particular an
abnormal curve is non-strictly abnormal precisely when the inhomogeneous system
$h_a=u^a$ is consistent within the solution space of the common adjoint
equation.
\end{proposition}

\dem
Immediate from Lemma~\ref{common_adjoint}: the solutions of the linear
homogeneous system form a vector space $\mathcal V$, the map
$\mathrm{ev}\colon\mathcal V\to C^\infty([0,T],\mathbb R^m)$,
$\lambda\mapsto(h_a(\lambda(\cdot)))_a$, is linear, the abnormal lifts are
$\ker\mathrm{ev}$ and the normal lifts are $\mathrm{ev}^{-1}(u)$.
\findem

The unreduced counterpart of Proposition~\ref{normal_affine} is
\cite[Rem.~2]{MR20}. What the reduction contributes is that $\mathcal V$ and
$\mathrm{ev}$ are computable from reduced data, which we now do.

\subsubsection*{The criterion}

Take the horizontal system \eqref{horizontal_system} with
$U\cong\mathbb R^{\dim B}$ and $L=\tfrac12|u|^2$, so that the $X_a$ are the
horizontal lifts of a frame $X_a$ of $TB$ that we now take orthonormal for a
Riemannian metric $g$ on $B$; equivalently, $\mathcal D=\ker A$ carries the
pullback of $g$. Let $\nabla$ be the Levi-Civita connection of $g$, let
$\nabla^A$ be the connection induced by $A$ on $\widetilde{\mathfrak g}^*$, and
write $\sharp$ and $\flat$ for the musical isomorphisms of $g$. Along a curve
$z(\cdot)$ in $B$ define the \emph{curvature pairing}
\begin{equation}
\label{curvature_pairing}
C_{\dot z}\colon\widetilde{\mathfrak g}^*_z\longrightarrow T_z^*B,
\qquad
\bigl\langle C_{\dot z}\mu,v\bigr\rangle
=
\bigl\langle\mu,F_A(\dot z,v)\bigr\rangle ,
\end{equation}
and let
\[
K_{\dot z}
=
\{v\in T_zB:\ F_A(\dot z,v)=0\}
\supseteq
\langle\dot z\rangle ,
\qquad
\operatorname{im}C_{\dot z}=K_{\dot z}^{\circ}.
\]

\begin{lemma}
\label{wong}
For the system \eqref{horizontal_system} with the above cost, $\xi\equiv0$ and
\eqref{RHP} becomes
\begin{align}
\label{normal_wong}
p_0=-1:&\qquad
\nabla_{\dot z}\dot z=-\bigl(C_{\dot z}\mu\bigr)^{\sharp},
\qquad
\nabla^A_{\dot z}\mu=0,
\qquad
p_z=\dot z^{\flat};
\\[0.4em]
\label{abnormal_wong}
p_0=\phantom{-}0:&\qquad
p_z\equiv0,
\qquad
C_{\dot z}\mu=0,
\qquad
\nabla^A_{\dot z}\mu=0,
\qquad
\mu\neq0 .
\end{align}
\end{lemma}

\dem
Since the controlled field is horizontal, $\widetilde\gamma=0$ and
$\xi=\partial h/\partial\mu=0$, so the third equation of \eqref{RHP} gives
$\nabla^A_{\dot z}\mu=0$ in both cases. On $p_0=-1$ stationarity gives $u^a=h_a$,
so $p_z=\dot z^\flat$, and the second equation of \eqref{RHP}, whose curvature
term is $\langle\mu,F_A(\cdot,\dot z)\rangle=-C_{\dot z}\mu$, becomes
\eqref{normal_wong}. The last step deserves its computation, since it is what
produces Wong's equation. The auxiliary connection of Section~\ref{rlpp} is
here the Levi-Civita connection of $g$, so $Dp_z/Dt=\nabla_{\dot z}p_z$ by
definition; and on the constraint manifold $h=\tfrac12|p_z|^2_{g^*}$, whose
covariant partial derivative in $z$ vanishes: parallel-translating $p_z$,
\[
\Bigl(\frac{\partial h}{\partial z}\Bigr)_{\!k}
=\tfrac12\bigl(\partial_kg^{ij}\bigr)p_ip_j
+g^{ij}\Gamma^l_{ki}\,p_lp_j
=\bigl(-g^{lj}\Gamma^i_{kl}+g^{lj}\Gamma^i_{kl}\bigr)p_ip_j
=0 ,
\]
the first term having been rewritten with
$\partial_kg^{ij}=-\Gamma^i_{kl}g^{lj}-\Gamma^j_{kl}g^{il}$, which is
$\nabla g^{-1}=0$, and the second by relabeling. The second equation of
\eqref{RHP} therefore reads $\nabla_{\dot z}p_z=-C_{\dot z}\mu$, which with
$p_z=\dot z^\flat$ and $\nabla g=0$ is \eqref{normal_wong}. On $p_0=0$,
Proposition~\ref{abnormal_leaf} gives $p_z\equiv0$, and the same equation reduces
to $C_{\dot z}\mu=0$, which is Theorem~\ref{abnormal_curvature} evaluated on the
velocity.
\findem

Equation \eqref{normal_wong} is Wong's equation for a classical particle of
charge $\mu$ in the Yang--Mills field $F_A$; that reduced normal sub-Riemannian
geodesics on a principal bundle are magnetic geodesics on the base is classical
\cite{Mo02}. The point of writing the two leaves side by side is the following
statement, which is immediate from them.

\begin{theorem}
\label{strictness_criterion}
Let $z(\cdot)$ be the base projection of an abnormal extremal of
\eqref{horizontal_system}, parametrized by arclength. Then $z(\cdot)$ is
\emph{non}-strictly abnormal if and only if there exists a
$\nabla^A$-parallel section $\mu^{\mathrm n}$ of $\widetilde{\mathfrak g}^*$ along
$z(\cdot)$ with
\begin{equation}
\label{strictness_equation}
\nabla_{\dot z}\dot z
=
-\bigl(C_{\dot z}\mu^{\mathrm n}\bigr)^{\sharp} .
\end{equation}
The set of such $\mu^{\mathrm n}$ is then an affine space over the parallel
sections of $\ker C_{\dot z}$, which contains the abnormal multiplier
$\mu^{\mathrm{ab}}$.
\end{theorem}

\dem
By Lemma~\ref{wong} a normal lift of the same curve exists if and only if
\eqref{normal_wong} holds for some parallel $\mu^{\mathrm n}$, and its $p_z$
component is then forced to be $\dot z^\flat$, which is the same curve. The
affine structure is Proposition~\ref{normal_affine}, the abnormal lifts being
the parallel sections of $\ker C_{\dot z}$ by \eqref{abnormal_wong}.
\findem

Two consequences are worth isolating. The first is pointwise and costs nothing;
the second is the case in which it decides everything.

\begin{corollary}
\label{strictness_pointwise}
A necessary condition for non-strictness is
\begin{equation}
\label{first_order_test}
g\bigl(\nabla_{\dot z}\dot z,\,v\bigr)=0
\qquad\text{for every }v\in K_{\dot z} .
\end{equation}
Since $\dot z\in K_{\dot z}$ and $g(\nabla_{\dot z}\dot z,\dot z)=0$ for an
arclength parametrization, the obstruction \eqref{first_order_test} takes values
in $\bigl(K_{\dot z}/\langle\dot z\rangle\bigr)^*$.
\end{corollary}

\dem
$\operatorname{im}C_{\dot z}=K_{\dot z}^\circ$, so \eqref{strictness_equation} is
solvable pointwise for $\mu^{\mathrm n}$ if and only if
$(\nabla_{\dot z}\dot z)^\flat$ annihilates $K_{\dot z}$.
\findem

\begin{corollary}
\label{strictness_rank_one}
Suppose $\dim\mathfrak g=1$. Then $C_{\dot z}\equiv0$ along an abnormal extremal,
and $z(\cdot)$ is non-strictly abnormal if and only if it is a geodesic of
$(B,g)$. In that case every parallel $\mu^{\mathrm n}$ is admissible, so the
normal lifts form a full line.
\end{corollary}

\dem
By \eqref{abnormal_wong} there is $\mu^{\mathrm{ab}}\neq0$ with
$C_{\dot z}\mu^{\mathrm{ab}}=0$; as $\widetilde{\mathfrak g}^*$ is a line bundle this
forces $C_{\dot z}=0$, whence $K_{\dot z}=T_zB$ and
\eqref{strictness_equation} reads $\nabla_{\dot z}\dot z=0$, with no condition on
$\mu^{\mathrm n}$; parallel sections of a line bundle along a curve exist and are
unique up to scale.
\findem

Corollary~\ref{strictness_rank_one} is not new in the abelian rank-one case: for
$B=\mathbb R^2$ with the flat metric and $G=\mathbb R$ it is essentially
Theorem~3.1 of \cite{Mo94}, which states that the horizontal lift of the zero
locus of a magnetic field fails to be a normal geodesic as soon as its planar
curvature is nonzero, and for Martinet-type distributions in $\mathbb R^3$ it is
Proposition~2.4 of \cite{PS25}, where the same conclusion is reached in a normal
form. What Theorem~\ref{strictness_criterion} adds is the general
principal-bundle statement, with $\mathfrak g$ of arbitrary dimension and possibly
nonabelian and with $\mu$ a $\nabla^A$-parallel section rather than a constant;
in that generality Corollary~\ref{strictness_rank_one} fails, as the
rank-two example below shows. A criterion of a comparable
kind for Engel structures, phrased through a structure function of the canonical
frame, is Theorem~4 of \cite{BM18}.

\begin{remark}
\label{why_curvature_blind}
The reason Theorem~\ref{abnormal_curvature} cannot decide strictness is now
visible. Abnormal extremals are determined by $\mathcal D=\ker A$, hence by $A$
and $F_A$, and are independent of the cost; strictness compares them with the
normal extremals, which depend on the cost through $\nabla$. Every object in
\eqref{first_order_test} and \eqref{strictness_equation} except
$\nabla_{\dot z}\dot z$ is metric-free, and $\nabla_{\dot z}\dot z$ carries the
entire dependence on the cost. Generic perturbation of the metric therefore
makes abnormal extremals strictly abnormal without moving them
\cite{CJT08}.
\end{remark}

\subsubsection*{The parallelism obstruction}

Condition \eqref{first_order_test} is only necessary: a pointwise solution
$\mu^{\mathrm n}(t)$ of \eqref{strictness_equation} need not be parallel. The
residual obstruction is again resolved by a constraint algorithm, this time run
in the fiber $\widetilde{\mathfrak g}^*$ rather than on $M$.

Trivialize $\widetilde{\mathfrak g}^*$ along $z(\cdot)$ by $\nabla^A$-parallel
transport and $T^*B$ by $\nabla$-parallel transport, and write $\Phi(t)$ for
$C_{\dot z(t)}$ and $b(t)$ for $-(\nabla_{\dot z}\dot z)^\flat$ in these
trivializations. Then \eqref{strictness_equation} together with parallelism reads
\begin{equation}
\label{strictness_linear}
\Phi(t)\,\mu=b(t)
\quad\text{for all }t,
\qquad
\mu\ \text{constant} .
\end{equation}
Differentiating repeatedly, set
\[
\mathcal W_k(t)
=
\bigl\{\mu:\ \Phi^{(j)}(t)\mu=b^{(j)}(t),\ j=0,\dots,k\bigr\},
\]
a decreasing sequence of affine subspaces of $\mathfrak g^*$.

\begin{proposition}
\label{strictness_algorithm}
For each fixed $t$ the decreasing sequence $\mathcal W_k(t)$ of affine
subspaces is eventually constant. If the data are real-analytic, then
$z(\cdot)$ is non-strictly abnormal if and only if
$\bigcap_{k\geq0}\mathcal W_k(t_0)\neq\emptyset$ for one, equivalently every,
$t_0$; the set of normal multipliers is then that intersection, which is
independent of $t_0$. Moreover, there is a discrete exceptional set
$\Sigma\subset[0,T]$ such that for $t_0\notin\Sigma$ the intersection is
already attained at order $N=\dim\mathfrak g$,
\[
\bigcap_{k\geq0}\mathcal W_k(t_0)=\mathcal W_{N}(t_0),
\]
so that, writing
$\Phi_{[N]}=(\Phi;\Phi';\dots;\Phi^{(N)})$ and
$b_{[N]}=(b;b';\dots;b^{(N)})$, non-strictness holds if and
only if
\begin{equation}
\label{rank_test}
\operatorname{rank}\Phi_{[N]}(t_0)
=
\operatorname{rank}\bigl(\Phi_{[N]}(t_0)\mid b_{[N]}(t_0)\bigr),
\qquad
t_0\notin\Sigma .
\end{equation}
\end{proposition}

\dem
For fixed $t$, each strict inclusion
$\mathcal W_{k+1}(t)\subsetneq\mathcal W_k(t)$ lowers the dimension by at
least one or produces the empty set, so at most $\dim\mathfrak g+1$ of the
inclusions are strict and the sequence is eventually constant. The index at
which it stabilizes is, however, not bounded by $\dim\mathfrak g$ at every
point: with $\dim\mathfrak g=1$, $\Phi\equiv0$ and $b(t)=(t-t_0)^2$ one has
$\mathcal W_0(t_0)=\mathcal W_1(t_0)=\mathfrak g^*$ while
$\mathcal W_2(t_0)=\emptyset$. This is the reason for the exceptional set
$\Sigma$.

If $\mu\in\mathcal W_k(t_0)$ for every $k$ then all derivatives of
$t\mapsto\Phi(t)\mu-b(t)$ vanish at $t_0$, so by analyticity the function
vanishes identically and $\mu$ is a normal multiplier; the converse is
differentiation. Hence
$\bigcap_k\mathcal W_k(t_0)=\{\mu:\Phi\mu\equiv b\}$, which proves the second
assertion and the independence of $t_0$.

For the last part, let $\mathcal N_k$ denote the span, over the field of
functions meromorphic on $[0,T]$, of the rows of the augmented matrices
$(\Phi^{(j)}\mid b^{(j)})$ with $j\leq k$. If $\mathcal N_{k+1}=\mathcal N_k$,
then every row of order $k+1$ is a meromorphic combination of rows of order at
most $k$; differentiating the combination expresses the rows of order $k+2$
through rows of order at most $k+1$, so the flag is constant from $k$ on.
Since $\dim\mathcal N_k\leq\dim\mathfrak g+1$ and each strict step increases
the dimension by at least one, $\mathcal N_k=\mathcal N_N$ for every
$k\geq N=\dim\mathfrak g$; iterating the relations, every row of
$(\Phi^{(k)}\mid b^{(k)})$ with $k>N$ is a meromorphic combination of the rows
of $(\Phi_{[N]}\mid b_{[N]})$, with coefficients whose poles form a discrete
set. Let $\Sigma$ consist of those poles together with the (discrete) zero
sets of the finitely many minors realizing the generic ranks of $\Phi_{[N]}$
and of $(\Phi_{[N]}\mid b_{[N]})$. For $t_0\notin\Sigma$ the equations
defining $\mathcal W_k(t_0)$ with $k>N$ are consequences, by evaluation at
$t_0$, of those defining $\mathcal W_N(t_0)$, so
$\mathcal W_k(t_0)=\mathcal W_N(t_0)$ for every $k\geq N$; and
\eqref{rank_test} is the consistency criterion for the linear system
$\Phi_{[N]}\mu=b_{[N]}$.
\findem

The primary constraint of this algorithm, that $\mathcal W_0(t)$ be nonempty,
is exactly \eqref{first_order_test}; the later stages are the
parallelism. Note that $\Phi_{[N]}(t_0)\mu^{\mathrm{ab}}=0$, so
$\Phi_{[N]}$ is never injective on an abnormal, in agreement with
Theorem~\ref{strictness_criterion}. When $\dim\mathfrak g=1$ one has
$\Phi\equiv0$ along an abnormal extremal, so the test reduces to the
identical vanishing of $b$, that is, to the geodesic condition of
Corollary~\ref{strictness_rank_one}; generically in $t_0$ the algorithm halts
at $k=0$.

\subsubsection*{The three rank-one examples}

For the Heisenberg problem of Section~\ref{heisenberg} there is nothing to test:
$F_A$ is nowhere zero, so by Corollary~\ref{heisenberg_no_abnormal} there are no
abnormal extremals.

For the flat Martinet problem of Section~\ref{martinet}, $B=\mathbb R^2$ carries
the Euclidean metric and the abnormal extremal projects to the line $\{x=0\}$,
which is a geodesic. Corollary~\ref{strictness_rank_one} therefore says it is
non-strictly abnormal, with a line of normal lifts parametrized by $\mu$. This is
Proposition~\ref{martinet_nonstrict}, obtained there from the explicit pendulum
and here from the metric alone.

The deformation of Remark~\ref{abck_deformation} is the interesting case, and the
criterion says exactly what the deformation does. The reduced metric is read off
from the declared orthonormal frame: the projections of
$(1+\varepsilon y)\partial_x+\tfrac12y^2\partial_z$ and $\partial_y$ are
$(1+\varepsilon y)\partial_x$ and $\partial_y$, so
\begin{equation}
\label{abck_metric}
g_\varepsilon
=
\frac{dx^2}{(1+\varepsilon y)^2}+dy^2 ,
\qquad
\Gamma^{y}_{xx}
=
\frac{\varepsilon}{(1+\varepsilon y)^3} ,
\end{equation}
both on the domain \eqref{abck_domain}, off which they are singular.
The abnormal locus is $\{y=0\}$ for every $\varepsilon$, and the abnormal curve
is the $x$-axis, which is arclength-parametrized by $\dot x=1$ there. Hence
\begin{equation}
\label{abck_acceleration}
\bigl(\nabla_{\dot z}\dot z\bigr)^y\Big|_{y=0}
=
\Gamma^y_{xx}\big|_{y=0}
=
\varepsilon ,
\end{equation}
and Corollary~\ref{strictness_rank_one} gives:

\begin{proposition}
\label{abck_strictness}
The abnormal extremal of the deformed Martinet problem of
Remark~\ref{abck_deformation} is non-strictly abnormal if and only if
$\varepsilon=0$.
\end{proposition}

This is the content of \cite[Thm.~1.3]{ABCK97} in its first assertion, recovered
from \eqref{abck_metric} by one Christoffel symbol. The reading it supplies is
that the deformation does not act on the abnormal extremals at all: on the domain
\eqref{abck_domain} the curvature still vanishes exactly on $\{y=0\}$ and the
abnormal curves are still the lines $y=0$, $z=z_0$. What the deformation destroys is the property of those lines of
being geodesics of the reduced metric, and with it the normal lift.

\subsubsection*{A higher-rank phenomenon}

When $\dim\mathfrak g\geq2$ the first-order test \eqref{first_order_test} may be
vacuous while strictness still fails, so that Corollary~\ref{strictness_rank_one}
has no analogue and the algorithm of Proposition~\ref{strictness_algorithm} is
needed. The following example is the smallest we know.

Let $B=\mathbb R^2$ with the Euclidean metric, $G=\mathbb R^2$, $P=B\times G$ with
coordinates $(x,y,z_1,z_2)$ and
\begin{equation}
\label{rank_two_connection}
A
=
\Bigl(
dz_1+a_1(x,y)\,dy,\ \
dz_2+\bigl(2x+\tfrac12x^2\bigr)dy
\Bigr),
\qquad
F_A
=
\bigl((x^2+y^2-1)^2,\ 2+x\bigr)\,dx\wedge dy ,
\end{equation}
where
\[
a_1(x,y)=\tfrac15x^5+\tfrac23x^3(y^2-1)+x(y^2-1)^2,
\qquad
\frac{\partial a_1}{\partial x}=(x^2+y^2-1)^2 ,
\]
so that $F_A=dA$, the group being abelian.
Consider the unit circle $z(s)=(\cos s,\sin s)$, of geodesic curvature $1$.
Since the first component of $F_A$ vanishes on it, $\mu^{\mathrm{ab}}=(1,0)$ is
a constant multiplier with $C_{\dot z}\mu^{\mathrm{ab}}=0$, so the circle is
abnormal.

The second component $2+\cos s$ is nowhere zero on the circle, so $C_{\dot z}\neq0$
at every point, hence
$K_{\dot z}=\langle\dot z\rangle$ everywhere and \eqref{first_order_test} holds
automatically: the first-order test sees nothing. The full criterion does. Since $(dx\wedge dy)(\dot z,v)=-\langle z,v\rangle$ and
$\nabla_{\dot z}\dot z=-z$, in the $\nabla$-parallel trivialization of
\eqref{strictness_linear} one has
\[
\Phi(s)\mu=-\mu_2\,(2+\cos s)\,z^\flat(s),
\qquad
b(s)=z^\flat(s),
\]
so that the residual is
\begin{equation}
\label{rank_two_residual}
\Phi(s)\mu-b(s)=-r(s)\,z^\flat(s),
\qquad
r(s)=1+\mu_2\,(2+\cos s).
\end{equation}
The covector $z^\flat(s)$ is nowhere zero, so \eqref{strictness_linear} reads
\[
\mu_2\,(2+\cos s)=-1
\qquad\text{for all }s,
\]
which has no constant solution: the circle is strictly abnormal.

Because $z^\flat$ never vanishes, \eqref{rank_two_residual} also reduces the
whole algorithm of Proposition~\ref{strictness_algorithm} to the scalar function
$r$: if $r$ and its first $k-1$ derivatives vanish at $s$, then the $k$-th
derivative of the residual is $-r^{(k)}(s)z^\flat(s)$, so
$\mathcal W_k(s)=\{\mu:\ r^{(j)}(s)=0,\ j\leq k\}$. Now
$\mathcal W_0(s)=\{\mu_2=-1/(2+\cos s)\}$ is nonempty for every $s$, and on it
$\mu_2\neq0$, so $r'(s)=-\mu_2\sin s$ vanishes only at $s\in\{0,\pi\}$:
$\mathcal W_1(s)=\emptyset$ off those two points. At them
$\mathcal W_1=\mathcal W_0$, while $r''(s)=-\mu_2\cos s\neq0$ there, so
$\mathcal W_2=\emptyset$. Thus $\bigcap_k\mathcal W_k=\mathcal W_2=\emptyset$ at
every $s$, in accordance with \eqref{rank_test} with $N=\dim\mathfrak g=2$ and
with no exceptional set needed; and the obstruction lies beyond the first stage
of the algorithm, since $\mathcal W_0$ is nowhere empty.

Replacing the second component of $F_A$ in \eqref{rank_two_connection} by the
constant $2$, that is, deleting the term $\tfrac12x^2$ from the second component
of $A$, changes the conclusion: then $\mu_2=-\tfrac12$ solves
\eqref{strictness_linear} and the same circle is non-strictly abnormal, with
$\mu_1$ free. The two cases have the same abnormal extremal and the same
first-order data, and are distinguished only by \eqref{rank_test}.

\begin{remark}
\label{rank_two_design}
Both components of \eqref{rank_two_connection} have been chosen to vanish to a
prescribed order on the circle, and it is worth saying which property each one
carries. The first is what makes the circle abnormal, and its \emph{double}
zero on $\{x^2+y^2=1\}$ makes $\nabla\langle\mu^{\mathrm{ab}},F_A\rangle$ vanish
there as well; this is used in Remark~\ref{rank_two_legendre} below, and a simple
zero would not do. The second is bounded away from zero, which is what makes the
first-order test \eqref{first_order_test} vacuous at every point, and
nonconstant, which is what makes the parallelism obstruction bite.
\end{remark}

\begin{remark}
\label{strictness_general_D}
We have assumed $\mathcal D=\ker A$, so that the reduced problem is Riemannian on
$B$, and this hypothesis is used in an essential way rather than for
convenience. It enters through Proposition~\ref{abnormal_leaf}, whose conclusion
$p_z\equiv0$ on the abnormal leaf follows from $p\in\mathcal D^\circ$ only
because $\mathcal D$ is the full horizontal space. If $\mathcal D$ is a proper
subbundle of $\ker A$, then $p\in\mathcal D^\circ$ gives merely
$p_z\in(\rho_*\mathcal D)^\circ$, so the abnormal base costate need not vanish;
\eqref{abnormal_wong} is then not the pair of equations
$C_{\dot z}\mu=0$, $\nabla^A_{\dot z}\mu=0$ alone, and the identification
$p_z=\dot z^\flat$ on the normal leaf, on which the proof of
Theorem~\ref{strictness_criterion} rests, likewise fails: only the restriction of
$p_z$ to $\rho_*\mathcal D$ is determined. Strictness then becomes a coupled
system in $(p_z,\mu)$ rather than a single acceleration equation for $\mu$, and
neither Theorem~\ref{strictness_criterion} nor
Corollary~\ref{strictness_rank_one} extends as stated. One expects the
replacement of \eqref{strictness_equation} to say that the projected curve is a
magnetic \emph{sub}-Riemannian geodesic for a parallel charge, with
$\nabla_{\dot z}\dot z$ replaced by the normal sub-Riemannian geodesic operator
of the reduced structure on $B$; but that is a conjecture and not a corollary of
what is proved here, and it would in any case require the full reduced adjoint
system together with its own normal and abnormal branches. See
Remark~\ref{D_equals_kerA_essential} for an explicit case in which the
hypothesis fails.
\end{remark}


\subsection{The reduced index form, and a trichotomy}
\label{index_form}

Section~\ref{strictness} decided whether an abnormal extremal carries a normal
lift. It did not decide whether it minimizes, and Proposition
\ref{criticality_optimality} disclaimed any second-order statement. We now supply
one, in the same reduced variables, and the outcome is a clean division of
labor: the second-order condition turns out to be \emph{metric-free}, so that it
is independent of the criterion of Section~\ref{strictness} rather than a refinement
of it.

Throughout, the hypotheses are those of Theorem~\ref{strictness_criterion}: the
horizontal system \eqref{horizontal_system} with $\mathcal D=\ker A$ carrying the
pullback of a Riemannian metric $g$ on $B$, and $z(\cdot)$ the base projection of
an abnormal extremal with $\nabla^A$-parallel multiplier $\mu\neq0$, so that
$C_{\dot z}\mu=0$ by \eqref{abnormal_wong}.

\subsubsection*{The endpoint map, reduced}

\begin{lemma}
\label{reduced_endpoint}
Fix $x_0\in P$ over $z_0\in B$ and let $E\colon u(\cdot)\mapsto x(T)$ be the
endpoint map of \eqref{horizontal_system}. Then $E$ factors through the base
curve: $x(\cdot)$ is the horizontal lift of $z(\cdot)$ through $x_0$, and
\[
E(u)=\bigl(z(T),\ \operatorname{hol}_A(z|_{[0,T]})\bigr)
\]
in a local trivialization. If $G$ is abelian, with $A=\theta-a$ for
$a\in\Omega^1(B;\mathfrak g)$ locally and $F_A=-da$, the fiber component is
\begin{equation}
\label{fibre_displacement}
J[z]=\int_0^Ta(\dot z)\,dt .
\end{equation}
An abnormal multiplier is a $\mu\in\mathfrak g^*\setminus\{0\}$ annihilating the
image of $dE$ on variations fixing $z(0)$ and $z(T)$, and
$\delta(\mu\cdot J)[v]=\int_0^T\langle\mu,F_A(\dot z,v)\rangle\,dt$, so
criticality is \eqref{abnormal_wong}.
\end{lemma}

\dem
The controlled field is horizontal, so $A(\dot x)=0$ and $x$ is the horizontal
lift; the fiber displacement of a horizontal lift is the holonomy, which in the
abelian case is \eqref{fibre_displacement}. For the first variation, with
$J=\int a_i(z)\dot z^i\,dt$ and $v$ vanishing at the endpoints, an integration by
parts gives $\delta J=\int(\partial_ja_i-\partial_ia_j)v^j\dot z^i\,dt
=\int(da)(v,\dot z)\,dt=\int\langle\mu,F_A(\dot z,v)\rangle\,dt$ after pairing
with $\mu$.
\findem

\subsubsection*{The index form}

\begin{theorem}
\label{reduced_index_form}
Let $G$ be \emph{abelian} and let $z(\cdot)$ be an abnormal extremal with
multiplier $\mu$. The Hessian of
$\mu\cdot E$ on variations $v$ of the base curve vanishing at the endpoints is
\begin{equation}
\label{index_form_eq}
Q_\mu(v)
=
\int_0^T
\Bigl\langle
\mu,\ F_A\bigl(\nabla_{\dot z}v,\,v\bigr)
+
\bigl(\nabla_vF_A\bigr)(\dot z,v)
\Bigr\rangle\,dt .
\end{equation}
Here $\nabla_{\dot z}v$ is the covariant derivative of the variation field
along $z(\cdot)$ with respect to the Levi-Civita connection $\nabla$ of $g$,
and $\nabla F_A$ is the covariant derivative of the section
$F_A\in\Gamma(\Lambda^2T^*B\otimes\widetilde{\mathfrak g})$ with respect to
the connection induced on $\Lambda^2T^*B\otimes\widetilde{\mathfrak g}$ by
$\nabla$ on the $T^*B$-factors and by $\nabla^A$ on the
$\widetilde{\mathfrak g}$-factor, that is,
\[
\bigl(\nabla_vF_A\bigr)(w_1,w_2)
=
\nabla^A_v\bigl(F_A(W_1,W_2)\bigr)
-F_A\bigl(\nabla_vW_1,w_2\bigr)
-F_A\bigl(w_1,\nabla_vW_2\bigr)
\]
for vector fields $W_1,W_2$ extending $w_1,w_2$. The right-hand side of
\eqref{index_form_eq} does not change if $\nabla$ is replaced by any other
torsion-free connection on $B$; in particular $Q_\mu$ does \emph{not} depend on
the metric $g$, but only on $A$ and $F_A$.
\end{theorem}

Two remarks on the domain of $Q_\mu$. Since abnormality is the pointwise
identity $\langle\mu,F_A(\dot z,\cdot)\rangle=0$, the first variation
$\delta(\mu\cdot J)[v]=\int_0^T\langle\mu,F_A(\dot z,v)\rangle\,dt$ vanishes for
\emph{every} $v$ with $v(0)=v(T)=0$, so the Hessian is already well defined,
independently of a choice of complement, on that whole space; this is why
\eqref{index_form_eq} may be stated there. The index form of Agrachev and
Sarychev is, however, the restriction of $Q_\mu$ to the intrinsic domain
\begin{equation}
\label{kernel_dE}
\ker dE
=
\Bigl\{v:\ v(0)=v(T)=0,\ \int_0^TF_A(\dot z,v)\,dt=0\in\mathfrak g\Bigr\},
\end{equation}
the variations preserving the full endpoint to first order, and it is on
\eqref{kernel_dE} that the index and nullity of \cite[\S4]{AS96} are computed.
The two spaces differ only when $\dim\mathfrak g\geq2$, and definiteness on the
larger one is the stronger hypothesis; we shall use only the pointwise
consequence \eqref{glc_reduced}, for which the distinction is immaterial.

\dem
Take $G$ abelian and expand \eqref{fibre_displacement} at $z+\varepsilon v$. The
coefficient of $\varepsilon^2$ is
\[
\int_0^T\Bigl[\tfrac12\partial_k\partial_ja_i\,v^kv^j\dot z^i
+\partial_ja_i\,v^j\dot v^i\Bigr]dt ,
\]
and, writing $\omega=da$ and using
$\partial_k\partial_ja_i-\partial_i\partial_ja_k=\partial_j\omega_{ki}$, this
equals
\[
\tfrac12\int_0^T\bigl[\omega(v,\dot v)+(\partial_v\omega)(v,\dot z)\bigr]dt
+\tfrac12\Bigl[\partial_ja_i\,v^iv^j\Bigr]_0^T ,
\]
the boundary term vanishing. Since $\omega=-F_A$, twice this is
$\int[\langle\mu,F_A(\dot v,v)+(\partial_vF_A)(\dot z,v)\rangle]dt$. Replacing
the coordinate derivatives by covariant ones with a torsion-free $\nabla$ changes
the integrand by $-\langle\mu,F_A(\dot z,\Gamma(v,v))\rangle$ alone, the terms in
$\Gamma(\dot z,v)$ canceling by symmetry of $\Gamma$; and that remainder
vanishes because $\langle\mu,F_A(\dot z,\cdot)\rangle=0$ along an abnormal.
The same cancellation proves the independence of the torsion-free connection:
if $\nabla'=\nabla+S$ with $S$ a symmetric $(1,2)$-tensor field, then
\[
F_A\bigl(\nabla'_{\dot z}v,v\bigr)-F_A\bigl(\nabla_{\dot z}v,v\bigr)
=F_A\bigl(S(\dot z,v),v\bigr),
\]
\[
\bigl(\nabla'_vF_A\bigr)(\dot z,v)-\bigl(\nabla_vF_A\bigr)(\dot z,v)
=-F_A\bigl(S(v,\dot z),v\bigr)-F_A\bigl(\dot z,S(v,v)\bigr);
\]
the first two changes cancel by the symmetry of $S$, and the last is
annihilated by $\mu$ along an abnormal extremal, since $C_{\dot z}\mu=0$. The
$\widetilde{\mathfrak g}$-factor of the connection on
$\Lambda^2T^*B\otimes\widetilde{\mathfrak g}$ is $\nabla^A$ in every case,
fixed by $A$, so the metric enters only through the $TB$-factor, whose choice
has just been shown to be immaterial.
\findem

\begin{remark}
\label{nonabelian_index}
The abelian hypothesis in Theorem~\ref{reduced_index_form} is not an artifact of
the proof. For nonabelian $G$ the fiber component of the endpoint map is the
path-ordered exponential $\mathrm{Pexp}\bigl(\int_0^Ta(\dot z)\,dt\bigr)$,
which reduces to \eqref{fibre_displacement} only in the abelian case, and the second-order term of its Magnus--Chen
expansion contains, besides the terms of \eqref{index_form_eq}, the time-ordered
double integral
\begin{equation}
\label{holonomy_commutator_term}
\Bigl\langle
\mu_0,\ \int_0^T\!\!\int_0^t\bigl[\nu(\sigma),\nu(t)\bigr]\,d\sigma\,dt
\Bigr\rangle,
\qquad
\nu(t)=\operatorname{Ad}_{U(t)}\bigl(F_A(\dot z,v)(t)\bigr),
\end{equation}
where $U(t)$ is the holonomy up to time $t$ and $\mu_0$ the multiplier at the
initial point. Abnormality gives $\langle\mu_0,\nu(t)\rangle\equiv0$, that is,
$\nu(t)\in\ker\mu_0$, which does \emph{not} annihilate
\eqref{holonomy_commutator_term}: the term vanishes identically precisely when
$\langle\mu_0,[\ker\mu_0,\ker\mu_0]\rangle=0$, in particular whenever
$\mathfrak g$ is abelian. Nor is it removed by the Goh condition. It is not
negligible either: for $B=\mathbb R^2$, $G=SO(3)$, $P=B\times G$ and the constant
local connection form $a=c\,e_3\,dx+e_1\,dy$, so that
$F_A=\tfrac12[a,a]=c\,e_2\,dx\wedge dy$, the line $z(t)=(t,0)$ with
$\mu_0=e_3^*$ is abnormal (note that $\operatorname{ad}^*_{ce_3}e_3^*=0$, so
$\mu_0$ is indeed $\nabla^A$-parallel), satisfies the full Goh condition
$C_\mu\equiv0$, and has $\mathcal L_\mu(v)=c\,v_y^2$. Take $c>0$. Formula
\eqref{index_form_eq} would then give $Q_\mu(v)=c\int_0^Tv_y^2\,dt>0$ for every
$T$, whereas the true second variation acquires from
\eqref{holonomy_commutator_term} a L\'evy-area functional whose kernel is
proportional to $\sin\bigl(c(t-\sigma)\bigr)$ and whose smallest eigenvalue
becomes negative for $cT$ large. The
\emph{pointwise} consequence \eqref{glc_reduced} survives for general $G$, since
under a concentration of width $\varepsilon$ the term
\eqref{holonomy_commutator_term} is $O(\varepsilon^2)$ against the $O(\varepsilon)$
of $\int\mathcal L_\mu$; what needs $G$ abelian is the global conclusion of
Corollary~\ref{reduced_glc}. Deriving the nonabelian index form, and identifying
the hypotheses under which \eqref{holonomy_commutator_term} cancels, we leave
open.
\end{remark}

\begin{lemma}
\label{bracket_legendre}
Let $V,Z$ be vector fields on $B$ and let $\lambda=\mu\circ A\in\mathcal D^\circ$
correspond to $\mu$ under \eqref{cotangent_quotient}. At a point at which
$\nabla V=\nabla Z=0$,
\begin{equation}
\label{bracket_to_curvature}
\lambda\Bigl(\bigl[V^h,[Z^h,V^h]\bigr]\Bigr)
=
-\bigl\langle\mu,(\nabla_VF_A)(Z,V)\bigr\rangle .
\end{equation}
\end{lemma}

\dem
By Lemma~\ref{bracket_curvature}, $[Z^h,V^h]=[Z,V]^h-\bigl(F_A(Z,V)\bigr)^\dagger$,
where $\sigma^\dagger$ denotes the vertical field of a section $\sigma$ of
$\widetilde{\mathfrak g}$. Since $[V^h,\sigma^\dagger]=(\nabla^A_V\sigma)^\dagger$,
\[
A\bigl([V^h,[Z^h,V^h]]\bigr)
=
-F_A(V,[Z,V])-\nabla^A_V\bigl(F_A(Z,V)\bigr).
\]
Expanding $\nabla^A_V(F_A(Z,V))=(\nabla_VF_A)(Z,V)+F_A(\nabla_VZ,V)+F_A(Z,\nabla_VV)$
and $[Z,V]=\nabla_ZV-\nabla_VZ$, and evaluating where $\nabla V=\nabla Z=0$,
leaves $-(\nabla_VF_A)(Z,V)$; the term $F_A(Z,\nabla_VV)$ is in any case killed by
$\mu$ on an abnormal.
\findem

Identity \eqref{bracket_to_curvature} is the translation announced above: the
quadratic form $v\mapsto\lambda([v,[\dot\gamma,v]])$ of Agrachev and Sarychev
\cite[(4.15)]{AS96}, which is the generalized Legendre form of an abnormal
extremal, becomes on a principal bundle a form built from the covariant
derivative of the curvature. We give it a name.

\begin{definition}
\label{reduced_legendre}
The \emph{reduced Legendre form} of an abnormal extremal at time $t$ is the
quadratic form
\begin{equation}
\label{legendre_form}
\mathcal L_\mu(v)=\bigl\langle\mu,(\nabla_vF_A)(\dot z,v)\bigr\rangle,
\qquad
v\in T_zB .
\end{equation}
\end{definition}

Its domain is the full tangent space of the base, modulo the velocity, and not
the kernel $K_{\dot z}$ of the curvature pairing. This is worth stating
explicitly, because $K_{\dot z}$ is the natural domain of the \emph{first-order}
obstruction \eqref{first_order_test} and the analogy is misleading. The reason
is that the generalized Legendre form of \cite[(4.15)]{AS96} is tested on all
control directions: taking $v=\varphi(t)e(t)$ with $e$ a $\nabla$-parallel field
and $\varphi$ a bump function makes the first term of \eqref{index_form_eq}
vanish, since $F_A(\nabla_{\dot z}v,v)=\varphi'\varphi\,F_A(e,e)=0$, and leaves
$Q_\mu(\varphi e)=\int\varphi^2\,\mathcal L_\mu(e)\,dt$, so that
$\mathcal L_\mu(e)\geq0$ is forced for \emph{every} $e\in T_zB$. Restricting the
test to $K_{\dot z}$ would make it vacuous whenever $K_{\dot z}=\langle\dot
z\rangle$, that is, in exactly the situations the test is meant to decide.

Two properties justify passing to the quotient by $\dot z$. First,
$\mathcal L_\mu(\dot z)=0$ by antisymmetry of $F_A$. Second, and less
immediately, $\mathcal L_\mu$ descends to $T_zB/\langle\dot z\rangle$ under the
full Goh condition $C_\mu=0$: polarizing \eqref{legendre_form} and
differentiating the abnormality relation
$\langle\mu,F_A(\dot z,V)\rangle\equiv0$ along $z(\cdot)$, using
$\nabla^A_{\dot z}\mu=0$, gives
\[
2\mathcal L_\mu(v,\dot z)
=\bigl\langle\mu,(\nabla_{\dot z}F_A)(\dot z,v)\bigr\rangle
=-\bigl\langle\mu,F_A(\nabla_{\dot z}\dot z,v)\bigr\rangle
=-C_\mu\bigl(\nabla_{\dot z}\dot z,v\bigr),
\]
which vanishes when $C_\mu=0$. Note that this is a further reason why
$K_{\dot z}$ is not the right domain: $\mathcal L_\mu$ need not descend to
$K_{\dot z}/\langle\dot z\rangle$ in the absence of Goh either.

\begin{corollary}
\label{reduced_glc}
Along an abnormal extremal the generalized Legendre--Clebsch condition reads
\begin{equation}
\label{glc_reduced}
\mathcal L_\mu\succeq0
\quad\text{on }T_zB,
\end{equation}
after normalizing the sign of $\mu$; this holds for every structure group, by
the localization argument in the proof below together with
Remark~\ref{nonabelian_index}.
Suppose in addition that $G$ is abelian, that the Goh condition holds in its
full form, $C_\mu=0$ along the extremal, and that $\mathcal L_\mu$ is uniformly
definite on $T_zB/\langle\dot z\rangle$. Then the extremal is
$W^{1,\infty}$-rigid, hence a local length minimizer.
\end{corollary}

The full Goh condition $C_\mu=0$ is a genuine additional hypothesis and is not
supplied by the constraint algorithm. What
Theorem~\ref{abnormal_curvature} produces as a secondary constraint is only the
tangency relation \eqref{goh_condition}, that is, $C_\mu(\dot z,\cdot)=0$; the
two coincide when $\dim B=2$, where a degenerate skew form on a plane vanishes,
and they differ as soon as $\dim B>2$. This is the distinction already recorded
after Theorem~\ref{abnormal_curvature}. The full condition holds in all the
examples below.

\dem
For \eqref{glc_reduced}, take $v=\varphi e$ as in the discussion following
Definition~\ref{reduced_legendre}; the correction needed to place $v$ in
$\ker dE$ is of second order in the width of $\varphi$ and does not affect the
sign. This argument is local and survives for nonabelian $G$, since by
Remark~\ref{nonabelian_index} the extra holonomy term is of second order in that
width against the first order of $\int\mathcal L_\mu$.

For the second assertion, $C_\mu=0$ annihilates the first term of
\eqref{index_form_eq}, so that
\begin{equation}
\label{index_form_exact}
Q_\mu(v)=\int_0^T\mathcal L_\mu(v)\,dt
\end{equation}
exactly, with no nonlocal quadratic remainder. This is where the abelian
hypothesis is used, and it is the whole of the argument:
Remark~\ref{nonabelian_index} exhibits the term that would otherwise appear on
the right of \eqref{index_form_exact}. Uniform definiteness means
$\mathcal L_\mu(v)\geq\beta\,\bigl|\,\overline v\,\bigr|^2$, with
$\overline v$ the class of $v$ in $T_zB/\langle\dot z\rangle$ and $\beta>0$
independent of $t$; by \eqref{index_form_exact},
\begin{equation}
\label{coercivity}
Q_\mu(v)\geq\beta\,\|\overline v\|^2_{L^2}
\end{equation}
for every variation $v$ vanishing at the endpoints. Writing $w=\dot v$ for the
corresponding control variation, so that $v$ is its primitive vanishing at
$t=0$, the right-hand side of \eqref{coercivity} is $\beta$ times the squared
$H^{-1}$ norm of $w$ in the convention of \cite[\S4]{AS96}, in which the
duality is taken against test functions not required to vanish at $t=T$, so that
\[
\|\dot v\|_{H^{-1}}
=\sup\Bigl\{\int_0^Tv\psi\,dt\ :\ \|\psi\|_{L^2}\leq1\Bigr\}
=\|v\|_{L^2}
\]
for $v(0)=0$; the estimate is taken on the
quotient by the reparametrization direction $\dot z$ along which the
unparametrized curve does not move. Estimate \eqref{coercivity} is therefore the
weak positive definiteness of the second variation hypothesized in
\cite[Thm.~4.8]{AS96}, whose conclusion is $W^{1,\infty}$-rigidity of the whole
arc, that is, isolation up to reparametrization.
\findem

\begin{remark}[attribution and terminology]
\label{rigidity_hypothesis}
\label{algorithm_vs_index}
\label{index_form_attribution}
\label{rigidity_terminology}
The ingredients are classical and are used as such. The index form of an abnormal
extremal, its metric independence, the Goh and generalized Legendre conditions
and the rigidity theorem are due to Agrachev and Sarychev \cite{AS96}, in the
Lie-bracket form $\lambda([v,[\dot\gamma,v]])$ on the total space; the
description of the abnormal extremals of a bundle sub-Riemannian structure as the
horizontal lifts of the curves with $\langle\mu,F_A(\dot z,\cdot)\rangle=0$ and
$\nabla^A_{\dot z}\mu=0$ is Montgomery's \cite[\S3.2]{Mo95}; conditions that rule out strictly abnormal
minimizers using the distribution alone (medium fatness) are in
\cite{AS99}. What is done here is
to write the second-order theory in the same reduced variables as the first-order
theory, so that all three rows of Theorem~\ref{trichotomy} below are computations with
$F_A$, $\nabla F_A$ and $g$ on the base rather than with iterated brackets
upstairs. Three points of interpretation follow. First, the hypotheses of the second assertion of Corollary~\ref{reduced_glc} should not be weakened casually:
Theorem~4.8 of \cite{AS96} assumes weak positive definiteness of the full second
variation, and what pointwise Goh together with the strong Legendre condition gives in general
is the finite-codimension positivity of \cite[Prop.~4.5]{AS96}, from which
rigidity follows for sufficiently short subarcs \cite[Thm.~5.5]{AS96} but not for
the whole arc; what rescues the
global conclusion in the abelian case is the exactness of
$Q_\mu=\int\mathcal L_\mu$ under full Goh, displayed in the proof, whereas for nonabelian $G$ the remainder
\eqref{holonomy_commutator_term} is present and, by the example of
Remark~\ref{nonabelian_index}, genuinely destroys positivity on long arcs. Second, the constraint
algorithm of Section~\ref{singular} produces, as its secondary constraints, the
\emph{equations} of Goh and of Legendre--Clebsch \cite{Ba08}, as Theorem~\ref{abnormal_curvature}
records for the first of them, but it is insensitive to the replacement of $\mu$ by $-\mu$ and so cannot produce a
\emph{sign}; the index
form \eqref{index_form_eq} supplies exactly that missing datum, and
\eqref{glc_reduced} is its first consequence. Third, \emph{rigid} is applied by Hsu
\cite[Prop.~8]{Hs92} and by Bryant and Hsu \cite{BH93} to an integral curve
admitting no compactly supported variations, and by Agrachev and Sarychev
\cite[Thm.~4.8]{AS96} to a curve isolated, up to reparametrization, in the
$W^{1,\infty}$ topology; Corollary~\ref{reduced_glc} asserts the second. The $C^0$-isolation among geodesics of
bounded length in \cite[Thm.~1.3]{ABCK97} is a further, distinct statement, and we do not assert that it is equivalent to
$W^{1,\infty}$-rigidity in the sense of \cite{AS96}: the admissible classes and the topologies differ, and the two conclusions
should be read separately.
\end{remark}

\subsubsection*{The trichotomy}

\begin{theorem}
\label{trichotomy}
For the horizontal system \eqref{horizontal_system} the following three questions
about an admissible curve are answered by three reduced tests, of which the first
two involve only the connection and its curvature and the third involves the
metric:
\begin{center}
\renewcommand{\arraystretch}{1.3}
\begin{tabular}{lll}
\hline
question & reduced test & data used\\
\hline
is it abnormal?
& iff $C_{\dot z}\mu=0$, $\nabla^A_{\dot z}\mu=0$ for some $\mu\neq0$
& $A,F_A$\\
is it rigid?
& sufficient: $C_\mu=0$ and $\mathcal L_\mu$ definite on $T_zB/\langle\dot z\rangle$
& $A,F_A,\nabla F_A$\\
is it strictly abnormal?
& iff no parallel $\mu^{\mathrm n}$ solves \eqref{strictness_equation}
& $A,F_A,g$\\
\hline
\end{tabular}
\end{center}
The first is a precondition for the other two, and the two remaining tests are
independent of one another: the outcome of one determines neither the outcome
nor the applicability of the other.
\end{theorem}

Three qualifications belong with the table. The second row is a
\emph{sufficient} certificate and not a criterion: by
Corollary~\ref{reduced_glc} and Remark~\ref{rigidity_hypothesis} it certifies
rigidity of the whole arc for abelian $G$, and of sufficiently short subarcs in
general, while the corresponding \emph{necessary} condition is the weaker
$\mathcal L_\mu\succeq0$ of \eqref{glc_reduced}. The first and third rows
require the connection $A$ itself and not only its curvature, since both involve
$\nabla^A$-parallel transport of $\mu$, and $F_A$ does not determine $A$. And
the metric enters the third row only, through $\nabla_{\dot z}\dot z$, which is
the content of Remark~\ref{why_curvature_blind}.

\dem
The three tests are \eqref{abnormal_wong}, Corollary~\ref{reduced_glc} and
Theorem~\ref{strictness_criterion}; the assertions about the data used are
Theorem~\ref{reduced_index_form} and the definition of $\nabla$. For the
independence, the family of Remark~\ref{abck_deformation}, worked out below, is
rigid for every $\varepsilon$ and changes strictness at $\varepsilon=0$, so
rigidity determines nothing about strictness; and the rank-two example of
Section~\ref{strictness} is strictly abnormal while its reduced Legendre form
vanishes identically, so that the second-order test is silent about its
rigidity, by Remark~\ref{rank_two_legendre}. The strictness test can therefore be
decisive where the second-order test is not.
\findem

\subsubsection*{The examples}

For the flat Martinet problem of Section~\ref{martinet}, $F_A=-x\,dx\wedge dy$
and the abnormal is $x\equiv0$ with $\dot z=(0,1)$, so
\begin{equation}
\label{martinet_legendre}
\mathcal L_\mu(v)=\mu\,v_x^2 ,
\end{equation}
definite on $T_zB/\langle\dot z\rangle$, which is here the line spanned by the
class of $\partial_x$, and uniformly so since $\mu$ is a nonzero constant. The
group is $\mathbb R$, hence abelian, and $C_\mu=0$ because $F_A$ vanishes on the
abnormal, so Corollary~\ref{reduced_glc} applies to the whole arc: the abnormal
is rigid and locally minimizing. This is
consistent with, and independent of, \cite[Thm.~1.2]{ABCK97}, where minimality is
obtained from the explicit description of the sphere: there the abnormal
geodesic is removed from the cut locus, so it minimizes for all time.

For the deformation of Remark~\ref{abck_deformation}, with
$F_{A_\varepsilon}$ as displayed there, the metric \eqref{abck_metric}, the
domain \eqref{abck_domain} and the abnormal $y\equiv0$,
$\dot z=(1,0)$,
\begin{equation}
\label{abck_legendre}
\mathcal L_\mu(v)=\mu\,v_y^2
\qquad\text{for every }\varepsilon,
\end{equation}
using $f_\varepsilon'(0)=1$ for
$f_\varepsilon(y)=\dfrac{y(2+\varepsilon y)}{2(1+\varepsilon y)^2}$; again definite on
$T_zB/\langle\dot z\rangle$, and again with $\mathbb R$ abelian and $C_\mu=0$ on
the abnormal, so the abnormal is rigid for every $\varepsilon$. This is the
counterpart, for the whole family and not only for $\varepsilon\neq0$, of the
$C^0$-isolation of \cite[Thm.~1.3]{ABCK97}; the two statements are not the same,
and we do not identify them (see Remark~\ref{rigidity_terminology}). Together with
Proposition~\ref{abck_strictness} this is the promised independence: the location
of the abnormal and its rigidity are the same for all $\varepsilon$, and only
strictness changes. In the language of Theorem~\ref{trichotomy}, the deformation
moves nothing in the first two rows and everything in the third.

\begin{remark}
\label{rank_two_legendre}
The rank-two example of Section~\ref{strictness} shows the other degeneracy.
There $\langle\mu^{\mathrm{ab}},F_A\rangle=(x^2+y^2-1)^2\,dx\wedge dy$, and
since $\nabla(x^2+y^2-1)^2=4(x^2+y^2-1)\,z$ vanishes on the unit circle,
\[
\mathcal L_{\mu^{\mathrm{ab}}}\equiv0
\qquad\text{along the circle.}
\]
The Goh condition holds in its full form there, $C_{\mu^{\mathrm{ab}}}=0$, since
$\langle\mu^{\mathrm{ab}},F_A\rangle$ itself vanishes on the circle. So
\eqref{glc_reduced} holds, but only with equality: the reduced Legendre form is
degenerate, Corollary~\ref{reduced_glc} does not apply, and the second-order test
says nothing about rigidity. Strictness, by contrast, was decided there. This is
the promised complementarity: the two tests can each be decisive where the other
is not.

The double zero of the first curvature component is what produces the
degeneracy, as noted in Remark~\ref{rank_two_design}. Had that component been
$x^2+y^2-1$, with a simple zero, one would have obtained
$\mathcal L_{\mu^{\mathrm{ab}}}(v)=-2\langle z,v\rangle^2$, whose kernel is
exactly $\langle\dot z\rangle$; on the one-dimensional quotient
$T_zB/\langle\dot z\rangle$, spanned by the radial class, that form is definite
after normalizing $\mu^{\mathrm{ab}}$ to $(-1,0)$, so
Corollary~\ref{reduced_glc} would certify the circle as rigid, hence locally
minimizing --- a strictly abnormal local minimizer, decided by both tests. Both
variants are instructive, and it is the domain $T_zB/\langle\dot z\rangle$ of
Definition~\ref{reduced_legendre}, rather than $K_{\dot z}/\langle\dot z\rangle$,
that makes the second computation meaningful: on $K_{\dot z}$ the test would be
empty in both variants, for the accidental reason that $K_{\dot
z}=\langle\dot z\rangle$ there, and would therefore distinguish nothing.
\end{remark}

\begin{remark}
\label{D_equals_kerA_essential}
The standing hypothesis $\mathcal D=\ker A$ of this subsection and of
Section~\ref{strictness} is essential rather than convenient, and the example of
\cite{Hs92} shows why. The Hilbert--Cartan system on $\mathbb R^5$, generated by
\[
\theta^1=dy-p\,dx,
\qquad
\theta^2=dp-q\,dx,
\qquad
\theta^3=dz-q^2\,dx ,
\]
is invariant under translation in $z$, so that $G=\mathbb R$, $B=\mathbb R^4$ with
coordinates $(x,y,p,q)$, and
\[
A=dz-q^2\,dx,
\qquad
F_A=2q\,dx\wedge dq .
\]
But $\mathcal D$ has rank two while $\ker A$ has rank four, and the conclusion of
Theorem~\ref{abnormal_curvature} fails accordingly. The family $\mathcal R$ of
\cite[(2.16)]{Hs92} consists of the integral curves with $q''=0$, that is,
$q(x)=a_2+a_3x$ with $y$ the corresponding cubic in $x$; it carries the five
parameters $(a_0,a_1,a_2,a_3,b_0)$. Along such a curve, with $\dot x=1$, the
reduced distribution contains $\partial_q$, and
$F_A(\dot z,\partial_q)=2q(x)$, which vanishes identically along the curve
only on the three-parameter subfamily $q\equiv0$. The curvature pairing
therefore does not detect the generic member of $\mathcal R$. What goes wrong is that when $\mathcal D$ is a
proper subbundle of $\ker A$ the brackets of $\mathcal D$ are not generated by
$F_A$ alone, and the reduced problem is sub-Riemannian rather than Riemannian on
the base, as anticipated in Remark~\ref{strictness_general_D}.
\end{remark}

\begin{remark}
Section~\ref{heisenberg} and the present example exhibit the two extremes
permitted by Theorem~\ref{abnormal_curvature}: the curvature is nowhere zero and
there are no abnormal extremals, or it degenerates on a hypersurface and the
abnormal extremals are exactly the admissible curves lying in that hypersurface.
Together with the spacecraft of Section~\ref{sec:spacecraft}, in which the
curvature vanishes identically but the coadjoint term is nontrivial, the three
examples exhaust the qualitative possibilities for the two terms of \eqref{RHP}.
\end{remark}

\bibliographystyle{plain}
\bibliography{references}

\end{document}